\documentclass{article}
\usepackage{a4wide}
\usepackage{graphicx} 
\usepackage{amsthm,amsmath,amssymb}
\usepackage{framed}
\usepackage{enumerate}
\usepackage{mathtools}
\usepackage{mathrsfs}
\usepackage{color, hyperref}
\usepackage{braket}
\usepackage{here}
\usepackage{mleftright}
\usepackage{typearea}
\usepackage{cleveref}
\usepackage{authblk}

\usepackage{tikz}
\usepackage{pgfplots}
\usetikzlibrary{arrows.meta, decorations.markings, patterns, shapes.geometric, shapes.misc}
\pgfplotsset{compat=1.18}

\allowdisplaybreaks

\theoremstyle{definition}
\newtheorem{dfn}{Definition}
\newtheorem{thm}[dfn]{Theorem}
\newtheorem{lem}[dfn]{Lemma}
\newtheorem{prp}[dfn]{Proposition}
\newtheorem{rmk}[dfn]{Remark}

\newtheorem{asm}[dfn]{Assumption}
\newtheorem{cor}[dfn]{Corollary}
\usepackage[backend=biber, giveninits=true, 
sorting=anyt, isbn=false, doi=false, url=false]{biblatex}
\newcommand{\tr}{\operatorname{tr}}

\newcommand{\ii}{\mathrm{i}}
\newcommand{\dif}{\mathrm{d}}
\newcommand{\sss}{\mathrm{ss}}
\newcommand{\wslim}[1]{\underset{#1}{\mathrm{w}^*\mathchar`-\lim~}}
\newcommand{\loc}{\mathrm{loc}}
\newcommand{\supp}{\operatorname{supp}}
\newcommand{\id}{\mathrm{id}}
\newcommand{\cb}{\mathrm{cb}}
\newcommand{\NESS}{\mathrm{NESS}}
\newcommand{\TANESS}{\mathrm{TANESS}}
\newcommand{\C}{\mathbb{C}}
\newcommand{\dist}{\operatorname{dist}}

\title{Nonequilibrium steady state in Lindblad dynamics for infinite quantum spin systems}
\author[1]{Kenji Shimomura\footnote{Corresponding author, email: \url{kenji.shimomura@yukawa.kyoto-u.ac.jp}}}
\author[2]{Nagisa Hara}
\author[3]{Seiichiro Kusuoka}
\affil[1]{Center for Gravitational Physics and Quantum Information, Yukawa Institute for Theoretical Physics, Kyoto University, Japan}
\affil[2]{Department of Mathematics and Statistics,
University of Ottawa, Canada}
\affil[3]{Department of Mathematics, Graduate School of Science, Kyoto University, Japan}

\begin{document}
\maketitle

\begin{abstract}
    We consider Lindblad dynamics of quantum spin systems on infinite lattices and define a nonequilibrium steady state (NESS) and a time-averaged nonequilibrium steady state (TANESS) on the basis of $C^*$-algebraic formalism.
    Generically, the NESS on an infinite system does not equal the thermodynamic limit of NESSs on finite systems.
    We give a sufficient condition that they coincide with each other, in terms of both a condition number, which quantifies the normality of a Liouvillian, and some spectral gaps on finite subsystems. 
    To appreciate the importance of the condition number, we provide an example in which the spectral gaps have nonzero lower bounds uniformly for any finite subsystems but a thermodynamic limit and a long-time limit (or a long-time average) do not commute with each other.
\end{abstract}

\tableofcontents

\section{Introduction}
Open quantum systems provide a framework for describing the dynamics of quantum systems coupled to external environments (see e.g. \cite{Breuer2002,Attal2006a,Attal2006b,Attal2006c,Rivas2011} as standard text books).
Due to dissipation or decoherence to the surroundings, the system often relaxes to a steady state that is generally out of equilibrium.
Recently, quantum phases of matter in NESSs have attracted a lot of attention in physics.
Phases of matter are something characterized by macroscopic properties, so they should be defined in a sufficiently large system, ideally an infinitely large system.
For equilibrium states and ground states of Hamiltonian dynamics, phases of matter in infinitely large systems have been investigated intensively on the basis of operator algebra (see e.g. \cite{Bratteli1987,Bratteli1997,Naaijkens2017} as text books).
In contrast, operator-algebraic research into NESS phases of matter is still in its early stages.
Toward rigorous characterization of various phases, we would like to provide an operator algebraic definition of NESSs in open quantum systems on infinite lattices.
More concretely, we aim to formulate the following physical setting mathematically: a considered quantum spin system is on an infinite lattice and coupled to an environment with dissipation.
When the system evolves over time starting from a certain initial state, one would expect it to relax into a steady state after a sufficiently long time.
This is our physical picture of NESS.

There are mainly two approaches to give operator algebraic description of open quantum systems.
First, we can take the environmental degrees of freedom into explicit consideration and introduce a unitary time evolution of the whole system consisting of the system and the environment.
This approach was proposed by Ruelle in \cite{Ruelle2000,Ruelle2001} and then developed by Jakšić and Pillet in \cite{Jaksic2001,Jaksic2002a,Jaksic2002}.
They considered a system with finite size that couples to different thermal baths with infinite size.
The whole system was described by a $C^*$-algebra, the time evolution over that was determined by a $C^*$-dynamics $\alpha_t$, and the baths were initially set in Kubo-Martin-Schwinger (KMS) states.
Within the formalism, an NESS for an initial state $\omega$ was defined as the long-time average of the evolved state, i.e., the cluster point of a net,
\begin{align*}
    \left(\frac{1}{T}\int_0^T \omega\circ\alpha_t \dif t\right)_{T\ge 0}.
\end{align*}
This approach is faithful to physical setup of open quantum systems, but it is generally hard to calculate the NESS.

Second, we can consider a closed form of non-unitary time evolution reduced into the system ignoring the environmental degrees of freedom.
This approach of so-called quantum dynamical semigroup supposes time evolution of observables to be completely-positive and unit-preserving.
For Markovian processes, a general form of Liouvillian of such dynamics was deduced by Gorini, Kossakowski, and Sudarshan in \cite{Gorini1976} for $N$-level systems in the Schr\"{o}dinger picture and by Lindblad in \cite{Lindblad1976,Lindblad1976b} for von Neumann algebras in the Heisenberg picture, independently.
Within the formalism, an NESS in a finite-dimensional system is usually defined by a density operator belonging to the kernel of the generator.
This approach is convenient for effective analysis of wide variety of dissipations but lacks a general definition of NESSs for infinite systems from the operator algebraic viewpoint.
One of the reasons is that it is unclear whether the Liouvillian on infinite lattices is well-defined due to the problem of a convergence.
In other words, a Liouvillian on an infinite lattice should be a sum of infinitely many local terms and should therefore be unbounded;
in general, the domain of an unbounded operator over a Banach space is a proper subset of the whole space due to the closed graph theorem, and there is not a canonical extension of the domain;
so, there is some ambiguity in defining NESS as the kernel of the unbounded Liouvillian.
Another reason is that some states over an operator algebra do not have a corresponding density operator, such as the tracial state or equivalently the infinite-temperature state over the quasi-local algebra.
Hence, Lindblad evolution in the Schr\"{o}dinger picture may fail to catch an infinite-system NESS that cannot be described by a density operator.

In this paper, we combine some methods in these approaches to give an operator algebraic definition of NESSs in Lindblad dynamics for quantum spin systems on infinite lattices.
We exploit the Lindblad form of Liouvillian to describe dynamics for the system without explicit consideration of the environment.
The existence of Lindblad dynamics $(\gamma_t^\Gamma)_{t\ge 0}$ on an infinite lattice $\Gamma$, i.e., the existence of thermodynamic limit $\Lambda\to\Gamma$ of quantum dynamical semigroups $(\gamma_t^\Lambda)_{t\ge 0}$, is guaranteed by the Lieb-Robinson bound for Lindblad generators under a locality condition, given in \cite{Nachtergaele2011} (see Theorem~\ref{thm:existence_of_dynamics}).
Then, imitating the manner of Ruelle and Jakšić-Pillet, we define an NESS and a time-averaged NESS (TANESS) as a long-time limit and a long-time average of the evolved state from an initial state, respectively, in Definition~\ref{dfn:NESS} and \ref{dfn:TANESS}.
This definition is, in fact, a mathematical representation of our physical setting of NESS as outlined at the beginning.

In our definition, the order of a long-time limit (or a long-time average) and a thermodynamic limit matters.
We find a sufficient condition that these limits commute with each other in Theorems~\ref{thm:main_theorem} and \ref{thm:main_theorem2}.
In the theorems, essential concepts are both three types of spectral gap $\Delta_\Lambda,\Delta^{\text{ex}}_\Lambda$, and $\Delta^{\text{p}}_\Lambda$ for a Liouvillian $\mathcal{L}_\Lambda$ on each finite subsystem $\Lambda$ and a condition number $\kappa_{\Lambda'}(\mathcal{V}_\Lambda^\epsilon)=\|\mathcal{V}_\Lambda^\epsilon\restriction\mathfrak{A}_{\Lambda'}\|_{\mathfrak{A}_{\Lambda'}\to\mathbb{C}^{d^{2|\Lambda|}}}\|(\mathcal{V}_\Lambda^\epsilon)^{-1}\|_{\mathbb{C}^{d^{2|\Lambda|}}\to\mathfrak{A}_\Lambda}$ of an invertible linear map $\mathcal{V}_\Lambda^\epsilon$ that diagonalizes a normalization of Liouvillian $\mathcal{L}_\Lambda$ by a small value $\epsilon$.
The map $\mathcal{V}_\Lambda^\epsilon\restriction\mathfrak{A}_{\Lambda'}$ is the restriction of $\mathcal{V}_\Lambda^\epsilon$ onto an algebra $\mathfrak{A}_{\Lambda'}$ with $\Lambda'\subset\Lambda$.
The condition number quantifies departure from normality of the Liouvillian.
The spectral gap $\Delta_\Lambda$ (or $\Delta^{\text{p}}_\Lambda$) represents the distance between the nonzero spectrum of the Liouvillian $\mathcal{L}_\Lambda$
and the imaginary axis (or the origin) of the complex plain, and $\Delta^{\text{ex}}_\Lambda$ does the distance between the set of non-semisimple eigenvalues of the Liouvillian $\mathcal{L}_\Lambda$ and the imaginary axis (see Figure~\ref{fig:gaps}).
Theorem~\ref{thm:main_theorem} (or \ref{thm:main_theorem2}) states that if the spectral gaps $\Delta_\Lambda$ (or $\Delta^{\text{p}}_\Lambda$) and $\Delta^{\text{ex}}_\Lambda$ have uniform nonzero bounds in finite subsystems $\Lambda$ from below and the condition number has a uniform bound in finite subsystems $\Lambda$ from above, then a thermodynamic limit commutes with a long-time limit (or a long-time average) in the NESS (or the TANESS) for any initial state.
To appreciate the importance of the condition number, we provide an example in which the spectral gaps have nonzero lower bounds uniformly for any finite subsystems but a thermodynamic limit and a long-time limit (or a long-time average) do not commute with each other, in Section~\ref{sec:example}.
Thus, the condition number plays a central role in characterizing the NESS and the TANESS as well as the spectral gaps.

The paper is organized as follows.
In Section~\ref{sec:Preliminary}, we introduce the $C^*$-algebraic setup of quantum spin systems and Lindblad dynamics by reviewing previous studies.
In Section~\ref{sec:Def_of_NESS}, we give a definition of NESS for quantum spin systems on infinite lattices.
In Section~\ref{sec:Main_theorems}, we establish and prove our main theorems on the order of a long-time limit (or a long-time average) and a thermodynamic limit.
In Section~\ref{sec:example}, for an example we calculate the expectation value of an observable in the NESS and solve the eigenvalue problems of the Liouvillian.
In appendices, we summarize basic definitions and properties about a net and a $C^*$-algebra used in the main text for physicists.
Physical aspects of the NESS such as spontaneous symmetry breaking will be investigated in another paper \cite{Shimomura2025}, the contents of which have been announced in \cite{Shimomura2026}.

\section{Preliminary on quantum spin systems and Lindblad dynamics}\label{sec:Preliminary}

In this section, we review the $C^*$-algebraic setup of quantum spin systems and Lindblad dynamics.

\subsection{Quantum spin systems}
We prepare a $C^*$-algebraic setup to describe quantum spin systems according to \cite{Bratteli1987,Bratteli1997,Naaijkens2017}.

Let $\Gamma$ be a countable set and $d$ a positive integer.
We call a point $j\in\Gamma$ a site.
Physically, $\Gamma$ stands for the lattice on which the quantum spin system is considered and $d$ does for the dimension of a Hilbert space of a spin at each site.
The symbol $\Lambda\Subset\Gamma$ denotes that $\Lambda$ is a finite subset of $\Gamma$.
To each finite subset $\Lambda\Subset\Gamma$, we assign a finite-dimensional $C^*$-algebra
\begin{align*}
    \mathfrak{A}_\Lambda
    \coloneqq \bigotimes_{j\in\Lambda}\mathrm{M}_d
\end{align*}
with the $C^*$-norm $\|\bullet\|$, where $\mathrm{M}_d$ is the $d$-dimensional matrix ring over $\mathbb{C}$ equipped with the Hermitian conjugate $*$ as the involution.
In the present case
the $C^*$-norm $\|\hat{A}\|$ for $\hat{A}\in\mathfrak{A}_\Lambda$ is characterized as the largest singular value of the operator $\hat{A}$.
For the empty set $\Lambda=\emptyset$, we assign $\mathfrak{A}_\emptyset\coloneqq\mathbb{C}$. 
The dimension of $\mathfrak{A}_\Lambda$ is $d^{2|\Lambda|}$.
Every $\mathfrak{A}_\Lambda$ has an identity:  $\hat{I}_\Lambda\coloneqq\bigotimes_{j\in\Lambda}\hat{I}_d$ for $\Lambda\neq\emptyset$ and $\hat{I}_\emptyset\coloneqq 1$.
When two finite subsets $\Lambda_1,\Lambda_2$ have the inclusion relation $\Lambda_1\subset\Lambda_2\Subset\Gamma$, we can introduce an inclusion map $\iota_{\Lambda_1,\Lambda_2}\colon\mathfrak{A}_{\Lambda_1}\hookrightarrow\mathfrak{A}_{\Lambda_2}$ by
\begin{align*}
    \iota_{\Lambda_1,\Lambda_2}(\hat{A})
    \coloneqq \hat{A}\otimes\hat{I}_{\Lambda_2\setminus\Lambda_1}
    \in\mathfrak{A}_{\Lambda_2},
    \quad
    \hat{A}
    \in\mathfrak{A}_{\Lambda_1},
\end{align*}
where $\hat{I}_d$ is the $d$-dimensional identity matrix.
For $\Lambda_1\subset\Lambda_2\Subset\Gamma$, hereafter, we consider that $\mathfrak{A}_{\Lambda_1}$ and $\mathfrak{A}_{\Lambda_2}$ have the inclusion relation $\mathfrak{A}_{\Lambda_1}\subset\mathfrak{A}_{\Lambda_2}$ and often identify an operator $\hat{A}$ in $\mathfrak{A}_{\Lambda_1}$ as $\iota_{\Lambda_1,\Lambda_2}(\hat{A})$ in $\mathfrak{A}_{\Lambda_2}$.
Since this inclusion map is isometric:
\begin{align*}
    \|\iota_{\Lambda_1,\Lambda_2}(\hat{A})\|
    = \|\hat{A}\|,
    \quad
    \hat{A}
    \in\mathfrak{A}_{\Lambda_1},
\end{align*}
we do not distinguish the norm of $\mathfrak{A}_\Lambda$ from that of $\mathfrak{A}_{\Lambda'}$ for two finite subsets $\Lambda,\Lambda'\Subset\Gamma$.
For any finite subset $\Lambda\Subset\Gamma$, we often call a self-adjoint operator in $\mathfrak{A}_\Lambda$ a local observable.

A normed $*$-algebra with the norm $\|\bullet\|$ is defined by
\begin{align*}
    \mathfrak{A}_{\loc}
    \coloneqq \bigcup_{\Lambda\Subset\Gamma}\mathfrak{A}_\Lambda.
\end{align*}
The identity $\hat{I}$ exists in $\mathfrak{A}_{\loc}$.
For $\hat{A}\in\mathfrak{A}_{\loc}$, the \emph{support} of $\hat{A}$ is the smallest finite subset $\Lambda\Subset\Gamma$ such that $\hat{A}\in\mathfrak{A}_\Lambda$ holds and we write it by $\supp\hat{A}$.
We define a \emph{quantum spin system} as the completion of $\mathfrak{A}_{\loc}$ with respect to the norm:
\begin{align*}
    \mathfrak{A}
    \coloneqq \overline{\bigcup_{\Lambda\Subset\Gamma}\mathfrak{A}_\Lambda}^{\|\bullet\|}.
\end{align*}
The quantum spin system $\mathfrak{A}$ is a unital $C^*$-algebra with the identity $\hat{I}$ and the $C^*$-norm $\|\bullet\|$ that are the canonical extensions of $\hat{I}\in\mathcal{A}_{\loc}$ and $\|\bullet\|$ on $\mathfrak{A}_{\loc}$.
We often call a self-adjoint operator in $\mathfrak{A}$ an observable.
The pair of $\mathfrak{A}$ and a net $(\mathfrak{A}_\Lambda)_{\Lambda\Subset\Gamma}$ is a \emph{quasi-local algebra}.


\subsection{Complete positivity and complete boundedness}\label{subsec:CP&CB}
Next, we introduce the positivity and complete positivity of linear maps on the quantum spin system $\mathfrak{A}$.
A linear map $\mathcal{T}$ from $\mathfrak{A}$ to itself is said to be \emph{positive} if $\mathcal{T}(\hat{A})$ is positive for every positive element $\hat{A}$ of $\mathfrak{A}$, and $\mathcal{T}$ is said to be \emph{completely positive} if the linear map $\mathcal{T}\otimes\id_{\mathrm{M}_n}\colon\mathfrak{A}\otimes\mathrm{M}_n\to\mathfrak{A}\otimes\mathrm{M}_n,\hat{A}\otimes\hat{B}\to\mathcal{T}(\hat{A})\otimes\hat{B}$ is positive for every positive integer $n$. 
A completely positive map is positive, but the converse is not true.

We also remark upon a norm of linear maps over $\mathfrak{A}_\Lambda$ for a finite subset $\Lambda\Subset\Gamma$.
See \cite{Paulsen2003} for details of this topic.
We define the norm $\|\bullet\|_{B(\mathfrak{A}_\Lambda)}$ by
\begin{align*}
    \|\mathcal{T}_\Lambda\|_{B(\mathfrak{A}_\Lambda)}
    \coloneqq \sup\{\|\mathcal{T}_\Lambda(\hat{A})\|\mid\hat{A}\in\mathfrak{A}_\Lambda,\|\hat{A}\|\le 1\}
\end{align*}
for every linear map $\mathcal{T}_\Lambda\colon\mathfrak{A}_\Lambda\to\mathfrak{A}_\Lambda$.
We write the set of bounded linear maps from $\mathfrak{A}_\Lambda$ to itself by $B(\mathfrak{A}_\Lambda)$.
Note that the vector space $\mathfrak{A}_\Lambda$ is finite-dimensional.
We often extend the domain $\mathfrak{A}_\Lambda$ of $\mathcal{T}_\Lambda$ to $\mathfrak{A}_{\Lambda'}$ with $\Lambda\subset\Lambda'\Subset\Gamma$ by identifying $\mathcal{T}_\Lambda$ as $\mathcal{T}_\Lambda\otimes\id_{\mathfrak{A}_{\Lambda'\setminus\Lambda}}$.
It is noteworthy that this extension alters the value of the norm, i.e., 
there exists a linear map $\mathcal{T}_\Lambda\in B(\mathfrak{A}_\Lambda)$ and a finite subset $\Lambda'\Subset\Gamma$ including $\Lambda$ such that
\begin{align*}
    \|\mathcal{T}_\Lambda\|_{B(\mathfrak{A}_\Lambda)}
    \neq \|\mathcal{T}_\Lambda\otimes\id_{\mathfrak{A}_{\Lambda'\setminus\Lambda}}\|_{B(\mathfrak{A}_{\Lambda'})}
\end{align*}
Hence, we should distinguish $\|\bullet\|_{B(\mathfrak{A}_\Lambda)}$ from $\|\bullet\|_{B(\mathfrak{A}_{\Lambda'})}$.
For a finite subset $\Lambda\Subset\Gamma$, a linear map $\mathcal{T}_\Lambda$ from $\mathfrak{A}_\Lambda$ to itself is said to be \emph{completely bounded} if for all positive integer $n$, the linear maps $\mathcal{T}_\Lambda\otimes\id_{\mathrm{M}_n}\colon\mathfrak{A}_\Lambda\otimes\mathrm{M}_n\to\mathfrak{A}_\Lambda\otimes\mathrm{M}_n$ are uniformly bounded in $n$, i.e.,
\begin{align*}
    \|\mathcal{T}_\Lambda\|_{\cb}
    \coloneqq \sup_{n\ge 1}\|\mathcal{T}_\Lambda\otimes\id_{\mathrm{M}_n}\|_{B(\mathfrak{A}_\Lambda\otimes\mathrm{M}_n)}
    < \infty.
\end{align*}
The following proposition implies that $\mathcal{T}_\Lambda$ is always completely bounded because $\mathfrak{A}_\Lambda$ is of finite dimension.
\begin{prp}[{\cite[Exercise 3.11]{Paulsen2003}}]
    For a finite subset $\Lambda\Subset\Gamma$ and a linear map $\mathcal{T}_\Lambda\colon\mathfrak{A}_\Lambda\to\mathfrak{A}_\Lambda=\mathrm{M}_{d^{|\Lambda|}}$, it holds that
    \begin{align*}
        \|\mathcal{T}_\Lambda\|_{\cb}\le d^{|\Lambda|}\|\mathcal{T}_\Lambda\|_{B(\mathfrak{A}_\Lambda)}.
    \end{align*}
\end{prp}


In the last part of this subsection, we introduce a specific way to extend the domain of a completely bounded map over $\mathfrak{A}_\Lambda$ into $\mathfrak{A}$.
Hereafter, we often encounter a situation in which we have a local map over $\mathfrak{A}_\Lambda$ for a finite subset $\Lambda\Subset\Gamma$ and want to extend it into the whole quantum spin system $\mathfrak{A}$.
Arveson's or Wittstock's extension theorem assures the existence of such an extension preserving the complete positivity or the complete boundedness (see e.g. \cite{Paulsen2003}).
Throughout this paper, in particular, we will adapt a specific way of extension as follows: For a given linear map $\mathcal{T}_\Lambda\colon\mathfrak{A}_\Lambda\to\mathfrak{A}_\Lambda$, we take the tensor product between $\mathcal{T}_\Lambda$ and identities to define the linear map $\tilde{\mathcal{T}}_\Lambda\colon\mathfrak{A}_{\text{loc}}\to\mathfrak{A}_{\text{loc}},\hat{A}\mapsto\mathcal{T}_\Lambda\otimes\id_{\supp\hat{A}\setminus\Lambda}(\hat{A})$.  
Since $\mathcal{T}_\Lambda$ is completely bounded, $\tilde{T}_\Lambda$ is a bounded map over $\mathfrak{A}_{\text{loc}}$, so we can get a unique extension $\tilde{\tilde{\mathcal{T}}}_\Lambda$ of $\tilde{T}_\Lambda$ onto the quantum spin system $\mathfrak{A}=\overline{\mathfrak{A}_{\text{loc}}}^{\|\bullet\|}$ with $\|\tilde{\tilde{\mathcal{T}}}_\Lambda\|_{B(\mathfrak{A})}=\|\mathcal{T}_\Lambda\|_{\text{cb}}$.
Moreover, if $\mathcal{T}_\Lambda$ is completely positive, the Kraus representation $\mathcal{T}_\Lambda(\hat{A})=\sum_j \hat{K}_j\hat{A}\hat{K}_j^*$ with $\hat{K}_j\in\mathfrak{A}_\Lambda$ is naturally extended to $\tilde{\tilde{\mathcal{T}}}_\Lambda\colon\mathfrak{A}\to\mathfrak{A}$ by embedding each $\hat{K}_j$ into $\mathfrak{A}$, and hence the extension $\tilde{\tilde{\mathcal{T}}}_\Lambda$ is also completely positive.
Thus, we obtain the bounded extension of $\mathcal{T}_\Lambda$ over $\mathfrak{A}$, which is completely positive if so is $\mathcal{T}_\Lambda$.

\subsection{Lindblad dynamics and existence of the thermodynamic limit}
We introduce the generator of Lindblad dynamics on the quantum spin systems $\mathfrak{A}$ in accordance with \cite{Nachtergaele2011}.

For each finite subset $Z\Subset\Gamma$, we give the following:
\begin{itemize}
    \item a self-adjoint operator $\hat{\Phi}(Z)\in\mathfrak{A}_Z$ with $\hat{\Phi}(Z)^* = \hat{\Phi}(Z)$;
    \item $N(Z)$ number of operators $\hat{L}_\alpha(Z)\in\mathfrak{A}_Z$ for $\alpha=1,2,\ldots,N(Z)$.
\end{itemize}
Then, we define the bounded linear maps $\mathcal{L}_\Lambda\colon\mathfrak{A}_\Lambda\to\mathfrak{A}_\Lambda$ for each finite subset $\Lambda\Subset\Gamma$ by
\begin{align}\label{eq:def_of_Liouvillian}
    \mathcal{L}_\Lambda(\hat{A})
    \coloneqq \sum_{Z\subset\Lambda}\Psi_Z(\hat{A}),
    \quad
    \hat{A}\in\mathfrak{A}_\Lambda,
\end{align}
with
\begin{align*}
    \Psi_Z(\hat{A})
    \coloneqq \ii[\hat{\Phi}(Z),\hat{A}] + \sum_{\alpha=1}^{N(Z)}\left(\hat{L}_\alpha(Z)^*\hat{A}\hat{L}_\alpha(Z) - \frac{1}{2}\{\hat{L}_\alpha(Z)^*\hat{L}_\alpha(Z),\hat{A}\}\right) \notag,
\end{align*}
where $[\hat{A},\hat{B}]\coloneqq\hat{A}\hat{B}-\hat{B}\hat{A}$ and $\{\hat{A},\hat{B}\}\coloneqq\hat{A}\hat{B}+\hat{B}\hat{A}$.
Here, we define $\Psi_Z(\hat{A})$ for $\hat{A}\in\mathfrak{A}_\Lambda$ and $Z\subset\Lambda$ by identifying $\Psi_Z$ as the bounded linear map $\Psi_Z\otimes\id_{\mathfrak{A}_{\Lambda\setminus Z}}$ from $\mathfrak{A}_\Lambda$ to itself.
Note that we are allowed to take $N(Z)=\infty$ if the infinite sum in $\Psi_Z(\hat{A})$ converges with respect to the operator norm for each $\hat{A}\in\mathfrak{A}_{\loc}$. 

The map $\mathcal{L}_\Lambda$ stands for the Lindblad generator on the finite subset $\Lambda$.
We call $\mathcal{L}_\Lambda$ a \textit{Liouvillian}.
The time evolution of $\hat{A}\in\mathfrak{A}_\Lambda$ on a finite subset $\Lambda$ is provided by the exponential $e^{t\mathcal{L}_\Lambda}$ of the bounded linear map $\mathcal{L}_\Lambda$ for the time $t\ge 0$,
which is a unit-preserving completely positive map over $\mathfrak{A}_\Lambda$ \cite{Lindblad1976}.
Following the way in Section~\ref{subsec:CP&CB}, we extend the map $e^{t\mathcal{L}_\Lambda}$ over $\mathfrak{A}_\Lambda$ to $\gamma_t^\Lambda$ over $\mathfrak{A}$ keeping it unit-preserving and completely positive.
We intend to define Lindblad dynamics over the infinite system $\Gamma$ by taking the thermodynamic limit of $\gamma_t^\Lambda$.
To ensure the existence of the thermodynamic limit, for now we impose locality conditions to $\Gamma$ and Liouvillians, based on \cite{Nachtergaele2011}.

\begin{asm}[Locality of Liouvillian]\label{asm:locality}
    Let the countable set $\Gamma$ be a metric space equipped with the distance $\dist\colon\Gamma\times\Gamma\to[0,\infty)$.
    \begin{enumerate}[(1)]
        \item There exists a non-increasing function $F\colon[0,\infty)\to[0,\infty)$ satisfying
        \begin{align*}
            \sup_{x\in\Gamma}\sum_{y\in\Gamma}F(\dist(x,y))
            < \infty
        \end{align*}
        and
        \begin{align*}
            \sup_{x,y\in\Gamma}\sum_{z\in\Gamma}\frac{F(\dist(x,z))F(\dist(z,y))}{F(\dist(x,y))}
            < \infty.
        \end{align*}
        \item There exists $\mu>0$ such that 
        \begin{align*}
            \sup_{x,y\in\Gamma}\sum_{\substack{Z\Subset\Gamma \\ \text{s.t. }x,y\in Z}}\frac{\|\Psi_Z\|_{\cb}}{F_\mu(\dist(x,y))}
            < \infty
        \end{align*}
        with $F_\mu(a)\coloneqq e^{-\mu a}F(a)$.
    \end{enumerate} 
\end{asm}
Under this two assumptions, the Lieb-Robinson bound was shown, which resulted in the existence of Lindblad dynamics over $\Gamma$ as follows.
\begin{thm}[\cite{Nachtergaele2011} Theorem 3]\label{thm:existence_of_dynamics}
    Let $(\mathcal{L}_\Lambda)_{\Lambda\Subset\Gamma}$ be a family of Liouvillians on finite subsets satisfying Assumption \ref{asm:locality}.
    Then, there exists a strongly continuous semigroup $(\gamma_t^\Gamma)_{t\ge 0}$ of unit-preserving and completely positive maps on the quantum spin system $\mathfrak{A}$ such that 
    \begin{align}\label{eq:TDL_of_dynamics}
        \lim_{\Lambda}\|\gamma_t^{\Lambda}(\hat{A}) - \gamma_t^\Gamma(\hat{A})\|
        = 0
    \end{align}
    for any $t\ge 0$ and any $\hat{A}\in\mathfrak{A}$.
    Moreover, the convergence above is uniform over every finite interval $[0,T]$ for $T\in(0,\infty)$.
\end{thm}

We remark that the original statement in \cite{Nachtergaele2011} has several differences from that above. 
First, we now consider the case the Liouvillians are independent of $t$, unlike the original one.
Second, we are handling the convergence of the net $(\gamma_t^\Lambda(\hat{A}))_{\Lambda\Subset\Gamma}$ in Eq.~\eqref{eq:TDL_of_dynamics}, while the original paper argues the convergence of subsequences $(\gamma_t^{\Lambda_n})_n$ for any exhausting increasing sequences $(\Lambda_n)_{n=1}^\infty$ of finite subsets $\Lambda_n\Subset\Gamma$.
The extension from the proof of the sequence version to that of the net version is straightforward; see Theorem 6.2.11 in \cite{Bratteli1997}, for instance.
Third, we included the uniform convergence in the statement, which was commented in the proof part of \cite{Nachtergaele2011}.

In Theorem~\ref{thm:existence_of_dynamics}, we fix an operator $\hat{A}$ and take the thermodynamic limit of $\gamma_t^\Lambda$.
Still, we can also increase the support of operator to the whole $\Gamma$.
When taking the thermodynamic limit of the dynamics, we can simultaneously increase the support of operator to the whole $\Gamma$.
\begin{cor}\label{cor:limit_of_operator}
    Let $(\mathcal{L}_\Lambda)_{\Lambda\Subset\Gamma}$ be a family of Liouvillians on finite subsets satisfying Assumption \ref{asm:locality}.
    Suppose that an operator $\hat{A}\in\mathfrak{A}$ is approximated by a sequence $(\hat{A}_n)_{n=1}^\infty$ on $\mathfrak{A}_{\loc}$ with support $\supp\hat{A}_n=\Lambda_n\Subset\Gamma$, i.e.,
    \begin{align*}
        \lim_{n\to\infty}\|\hat{A}_n-\hat{A}\|
        = 0.
    \end{align*}
    Then, it holds that
    \begin{align*}
       \gamma_t^\Gamma(\hat{A})
       = \lim_{n\to\infty}\gamma_t^{\Lambda_n}(\hat{A}_n).
    \end{align*}
\end{cor}
\begin{proof}
    Since $\|\gamma_t^{\Lambda_n}(\hat{B})\|\le\|\hat{B}\|$ for any $\hat{B}\in\mathfrak{A}$ and $n\in\mathbb{N}$, we have
    \begin{align*}
        \|\gamma_t^{\Lambda_n}(\hat{A}_n) - \gamma_t^\Gamma(\hat{A})\|
        \le \|\hat{A}_n-\hat{A}\| + \|\gamma_t^{\Lambda_n}(\hat{A})-\gamma_t^{\Gamma}(\hat{A})\|
        \xrightarrow{n\to\infty} 0.
    \end{align*}
\end{proof}

\section{Definition of nonequilibrium steady state}\label{sec:Def_of_NESS}
In this section, we introduce our definition of nonequilibrium steady states and their time-averaged counterpart.

\begin{dfn}[Nonequilibrium steady state]\label{dfn:NESS}
    Let $\omega$ be a state over the quantum spin system $\mathfrak{A}$.
    For an initial state $\omega$, we define the \emph{nonequilibrium steady state (NESS)} $\omega_{\sss}$ by a weak-$*$ cluster point of a net $(\omega\circ\gamma_t^\Gamma)_{t>0}$ of states.
    We write the total set of the NESS for an initial state $\omega$ by $\Sigma_{\NESS}(\omega)$, i.e.,
    \begin{align*}
        \Sigma_{\NESS}(\omega)
        \coloneqq \mleft\{\omega_{\sss} \;\middle|\; \text{$\exists(t_q)_q\colon$a subnet of $(t)_{t>0}$ s.t. }\omega_{\sss}=\wslim{q}\omega\circ\gamma_{t_q}^\Gamma \mright\}.
    \end{align*}
\end{dfn}

\begin{dfn}[Time-averaged nonequilibrium steady state]\label{dfn:TANESS}
    Let $\omega$ be a state on the quantum spin systems $\mathfrak{A}$.
    For an initial state $\omega$, we define the \emph{time-averaged nonequilibrium steady state (TANESS)} $\omega_{\sss}^{\mathrm{ave}}$ by a weak-$*$ cluster point of a net $\left(\frac{1}{T}\int_0^T \omega\circ\gamma_t^\Gamma \dif t\right)_{T>0}$ of states, where we define the integral of a linear-functional-valued continuous function $\varphi_t$ over $\mathfrak{A}$ by 
    \begin{align*}
        \left(\int_0^T \varphi_t\dif t\right)(\hat{A})
        \coloneqq \int_0^T \varphi_t(\hat{A})\dif t
    \end{align*}
    for $\hat{A}\in\mathfrak{A}$.
    We write the total set of the TANESS for an initial state $\omega$ by $\Sigma_{\TANESS}(\omega)$, i.e.,
    \begin{align*}
        \Sigma_{\TANESS}(\omega)
        \coloneqq \mleft\{\omega_{\sss}^{\mathrm{ave}} \;\middle|\; \text{$\exists(T_q)_q\colon$a subnet of $(t)_{t>0}$ s.t. }\omega_{\sss}^{\mathrm{ave}}=\wslim{q}\frac{1}{T_q}\int_0^{T_q}\omega\circ\gamma_t^\Gamma\dif t \mright\}.
    \end{align*}
\end{dfn}

\begin{rmk}\label{rmk:Bochner}
    We comment on the integral $\int_0^T\omega\circ\gamma_t^\Gamma\dif t$.
    As mentioned above, it is defined via the integral $\int_0^T\omega\circ\gamma_t^\Gamma(\hat{A})\dif t$ of a complex-valued function $t\in\mathbb{R}\mapsto \omega\circ\gamma_t^\Gamma(\hat{A})\in\mathbb{C}$.
    
    Meanwhile, we can also characterize it in another manner. 
    The semigroup $(\gamma_t^\Gamma)_{t\ge 0}$ is strongly continuous, so the Bochner integral $\int_0^T \gamma_t^\Gamma(\hat{A})\dif t$ is well-defined for each $\hat{A}\in\mathfrak{A}$.
    Thereby we can define a map $\int_0^T\gamma_t^\Gamma(\bullet)\dif t \colon \hat{A}\in\mathfrak{A}\mapsto\int_0^T \gamma_t^\Gamma(\hat{A})\dif t\in\mathfrak{A}$ and the composition $\omega\circ\int_0^T\gamma_t^\Gamma(\bullet)\dif t\colon\mathfrak{A}\to\mathbb{C}$.
    In fact, $\omega\circ\int_0^T\gamma_t^\Gamma(\bullet)\dif t$ is equivalent to $\int_0^T\omega\circ\gamma_t^\Gamma\dif t$.
    Indeed, for each $\hat{A}\in\mathfrak{A}$, we have
    \begin{align*}
        \omega\circ\int_0^T\gamma_t^\Gamma(\hat{A})\dif t
        = \int_0^T \omega\circ\gamma_t^\Gamma(\hat{A})\dif t
    \end{align*}
    because $\omega$ is bounded.

    In the same way, we can define the map $\int_0^T\gamma_t^\Lambda(\bullet)\dif t\colon\mathfrak{A}\to\mathfrak{A}$ and show
    \begin{align*}
        \omega\circ\int_0^T\gamma_t^\Lambda(\bullet)\dif t
        = \int_0^T\omega\circ\gamma_t^\Lambda\dif t
    \end{align*}
    for any finite subset $\Lambda\Subset\Gamma$.
    The argument here is necessary for the proof of the main theorem on TANESS.
\end{rmk}

For unitary dynamics where all of the dissipators $\hat{L}_\alpha$ vanish, the TANESS defined here is equivalent to the conventional notion of NESS in \cite{Jaksic2002}. 
The NESS and TANESS for an initial state always exist by Theorem 2.3.15 in \cite{Bratteli1987} (see Theorem \ref{thm:compactness_of_states_set} in Appendix) and are also states over $\mathfrak{A}$, but are not necessarily unique, i.e.,
\begin{align*}
    |\Sigma_{\NESS}(\omega)|
    \ge 1,
    \quad
    |\Sigma_{\TANESS}(\omega)|
    \ge 1.
\end{align*}

If a state $\omega_1$ over $\mathfrak{A}$ is an NESS (or a TANESS) for an initial state $\omega$, then the evolved state $\omega_1\circ\gamma_t^\Gamma$ for an interval $t\ge 0$ is also an NESS (or a TANESS) though it does not necessarily equal the original NESS (or TANESS) $\omega_1$.
Therefore, we can say that an NESS (or a TANESS) is not necessarily invariant under the time evolution $\gamma_t^\Gamma$ but the set $\Sigma_{\NESS}(\omega)$ (or $\Sigma_{\TANESS}(\omega)$) is closed under the time evolution.
In the particular case when the NESS (or TANESS) is unique for an initial state, it is invariant under the time evolution.

The expectation value of an operator $\hat{A}$ under an NESS $\omega_{\sss}$ or a TANESS $\omega_{\sss}^{\mathrm{ave}}$ is computed by
\begin{align*}
    \omega_{\sss}(\hat{A})
    = \lim_q \lim_\Lambda \omega\circ\gamma_{t_q}^\Lambda(\hat{A}),
    \quad
    \omega_{\sss}^{\mathrm{ave}}(\hat{A})
    = \lim_q \frac{1}{T_q} \int_0^{T_q}\lim_\Lambda\omega\circ\gamma_{t}^\Lambda(\hat{A})\dif t,
\end{align*}
respectively.
We can use Proposition~\ref{cor:limit_of_operator} to calculate the thermodynamic limit $\Lambda\to\Gamma$. 
It is noteworthy that the thermodynamic limit $\Lambda\to\Gamma$ is taken first and then the long-time limit $t_q\to\infty$ or $T_q\to\infty$ is done second.
This order of the limits does not commute in general.
In Section~\ref{sec:example}, we will see an example where different orders of the limits yield distinct values.
Before the case study, we will establish a sufficient condition that the limits commute each other in the next section.

\section{A sufficient condition that thermodynamic and long-time limits commute}\label{sec:Main_theorems}
As noted in the previous section, the order of the thermodynamic limit $\Lambda\to\Gamma$ and the long-time limit $t_q\to\infty$ or $T_q\to\infty$ is a sensitive issue for defining NESSs and TANESSs. 
In this section, we elucidate a sufficient condition that they commute each other.
We describe the statement of main theorems for NESSs and TANESSs in Sec.~\ref{subsec:main_theorems}, prepare some lemmas in Sec.~\ref{subsec:behavior_on_finite_subsets}, and then prove the main theorems in Sec.~\ref{subsec:proof_of_main_theorems}.
In Sec.~\ref{subsec:TDL_of_finite_system_NESS_and_TANESS} we give a remark about finite system NESSs and TANESSs.
Lastly in Sec.~\ref{subsec:application_of_main_theorems} we show that the main theorems are applicable to onsite Liouvillians.

\subsection{Main theorems}\label{subsec:main_theorems}
In this subsection, we introduce the statement of main theorems and give some remarks on them.
First, we see a theorem about NESSs.
Hereafter, we define $\min\emptyset\coloneqq+\infty$.

\begin{thm}\label{thm:main_theorem}
    Let $(\mathcal{L}_\Lambda)_{\Lambda\Subset\Gamma}$ be a family of Liouvillians on finite subsets satisfying Assumption \ref{asm:locality}.
    For each finite subset $\Lambda\Subset\Gamma$, let $\sigma_{B(\mathfrak{A}_\Lambda)}(\mathcal{L}_\Lambda)$ be the spectrum of the linear map $\mathcal{L}_\Lambda$ from the finite-dimensional vector space $\mathfrak{A}_\Lambda$ to itself
    and $\sigma_{B(\mathfrak{A}_\Lambda)}^{\mathrm{ex}}$ the set of all non-semisimple eigenvalues of $\mathcal{L}_\Lambda$.
    Suppose that
    \begin{enumerate}[($\text{A}$1)]
        \item there exists a positive number $\Delta$ such that for any finite subset $\Lambda\Subset\Gamma$ we have
        \begin{align}\label{eq:line_gap}
            \Delta_\Lambda
            \coloneqq \min\{ |\operatorname{Re}\lambda| \mid \lambda\in\sigma_{B(\mathfrak{A}_\Lambda)}(\mathcal{L}_\Lambda)\setminus\{0\} \}
            \ge \Delta
            > 0;
        \end{align}
        \item there exist (finite) positive numbers $\Delta^{\mathrm{ex}}$ and $\epsilon\in (0,\Delta^{\mathrm{ex}})$ such that we have both
        \begin{align}\label{eq:existence_of_ex_gap}
            \Delta_\Lambda^{\mathrm{ex}}
            \coloneqq \min\{|\operatorname{Re}\lambda| \mid\lambda\in\sigma_{B(\mathfrak{A}_\Lambda)}^{\mathrm{ex}}(\mathcal{L}_\Lambda)\}
            \ge \Delta^{\mathrm{ex}}
            > 0
        \end{align}
        for any finite subset $\Lambda\Subset\Gamma$ and
        \begin{align*}
            \kappa_{\Lambda'}^\epsilon
            \coloneqq \sup_{\substack{\Lambda\Subset\Gamma \\ \text{s.t. }\Lambda'\subset\Lambda}}\|\mathcal{V}_\Lambda^\epsilon\restriction\mathfrak{A}_{\Lambda'}\|_{\mathfrak{A}_{\Lambda'}\to\C^{d^{2|\Lambda|}}}\|(\mathcal{V}_\Lambda^\epsilon)^{-1}\|_{\C^{d^{2|\Lambda|}}\to\mathfrak{A}_\Lambda}
            < \infty
        \end{align*}
        for any finite subset $\Lambda'\Subset\Gamma$, where $\mathcal{V}_\Lambda^\epsilon\in B(\mathfrak{A}_\Lambda,\C^{d^{2|\Lambda|}})$ is an invertible linear map that transforms $(\Delta^{\mathrm{ex}}-\epsilon)^{-1}\mathcal{L}_\Lambda$ into the Jordan canonical form $\mathcal{J}_\Lambda^{\epsilon}\in\mathrm{M}_{d^{2|\Lambda|}}$,
        \begin{align*}
            \mathcal{J}_\Lambda^\epsilon
            = \mathcal{V}_\Lambda^\epsilon \circ \frac{1}{\Delta^{\mathrm{ex}}-\epsilon}\mathcal{L}_\Lambda \circ (\mathcal{V}_\Lambda^\epsilon)^{-1},
        \end{align*}
        and the norms $\|\bullet\|_{\mathfrak{A}_\Lambda\to\C^{d^{2|\Lambda|}}}$ and $\|\bullet\|_{\C^{d^{2|\Lambda|}}\to\mathfrak{A}_\Lambda}$ are defined by
        \begin{gather*}
            \|\mathcal{S}\|_{\mathfrak{A}_\Lambda\to\C^{d^{2|\Lambda|}}}
            = \sup\{ {\|\mathcal{S}(\hat{A})\|}_2 \mid \hat{A}\in\mathfrak{A}_\Lambda,\ \|\hat{A}\|\le 1 \},
            \quad
            \mathcal{S}\colon\mathfrak{A}_\Lambda\to\C^{d^{2|\Lambda|}}, \\
            \|\mathcal{T}\|_{\C^{d^{2|\Lambda|}}\to\mathfrak{A}_\Lambda}
            = \sup\{ \|\mathcal{T}(\boldsymbol{v})\| \mid \boldsymbol{v}\in\C^{d^{2|\Lambda|}},\ {\|\boldsymbol{v}\|}_2\le 1 \},
            \quad
            \mathcal{T}\colon\C^{d^{2|\Lambda|}}\to\mathfrak{A}_\Lambda.
        \end{gather*}
        Here, $\|\bullet\|_2$ is the 2-norm of numerical vectors.
        The map $\mathcal{V}_\Lambda^\epsilon\restriction\mathfrak{A}_{\Lambda'}$ is the restriction of $\mathcal{V}_\Lambda^\epsilon$ onto $\mathfrak{A}_{\Lambda'}(\subset\mathfrak{A}_\Lambda)$.
    \end{enumerate}
    Then, for any state $\omega$ over the quantum spin system $\mathfrak{A}$, the limit
    \begin{align}\label{eq:LTLthenTDL}
        \wslim{\Lambda} \wslim{t\to\infty}\omega\circ\gamma_t^{\Lambda}
    \end{align}
    exists and is the unique NESS for the initial state $\omega$. 
\end{thm}

Roughly speaking, Theorem~\ref{thm:main_theorem} says that the thermodynamic limit and the long-time limit commute under the assumptions (A1) and (A2).
We can express the long-time limit $\wslim{t\to\infty}\omega\circ\gamma_t^\Lambda$ by using the spectral projection of the Liouvillian $\mathcal{L}_\Lambda$ over $\mathfrak{A}_\Lambda$ as shown in Lemma~\ref{lem:NESS_in_finite_systems} below, so commuting the limits may simplify calculating the NESS.
The limit in Eq.~\eqref{eq:LTLthenTDL} means the thermodynamic limit of finite system NESSs, which should be distinguished from the NESS in the infinite system $\Gamma$.
Furthermore, the theorem also yields the uniqueness of the NESS for a given initial state.
In particular, we can calculate the NESS by taking the long-time limit $t\to\infty$ without considering a certain subnet.
Note that Theorem~\ref{thm:main_theorem} allows different NESSs to exist for different initial states. 

\begin{figure}
    \centering
    \begin{tikzpicture}[scale=1.]
        
        \draw[-{Latex[length=3mm]}] (-4,0) -- (3,0) node[right] {$\text{Re}$};
        \draw[-{Latex[length=3mm]}] (2,-3) -- (2,3) node[above] {$\text{Im}$};
        
        \draw[blue!70, dashed, thick] (-2,-3) -- (-2,3);
        \draw[blue!70, dashed, thick] (1,-3) -- (1,3);
    
        \foreach \x/\y in {-3/1, -3/-1, 1/2, 1/-2, -1/0, -0.5/0.5, -0.5/-0.5, 0.5/0, 2/0} {
            \node[cross out, draw=red, thick, minimum size=6pt, inner sep=0pt] at (\x,\y) {};
        }
        
        \foreach \x/\y in {-2/1.5, -2/-1.5, -2.5/0} {
            \node[diamond, draw=red, thick, fill=red!30, minimum size=10pt, inner sep=0pt] at (\x,\y) {};
        }
        
        \draw[blue!70, latex-latex] (-2,2.5) -- (2,2.5) node[midway, above] {$\Delta^{\text{ex}}_\Lambda$};
        \draw[blue!70, latex-latex] (2,0) -- (0.7,-0.75) node[midway, below] {$\Delta^{\text{p}}_\Lambda$};
        \draw[blue!70, latex-latex] (2,-2.5) -- (1,-2.5) node[midway, below] {$\Delta_\Lambda$};
        
        \draw[blue!70, thick, dashed] (2,-1.5) to[out=180, in=270] (0.5,0) to[out=90, in=180] (2,1.5);
        
        \node[draw, rectangle, left, align=center] at (-2,3) {$\sigma_{B(\mathfrak{A}_\Lambda)}(\mathcal{L}_\Lambda)$};
    
        \node[right, align=center] at (2,0.2) {$0$};
        
        \node at (0,-4) {
            \begin{tabular}{@{}c@{\hspace{1em}}c@{\hspace{1em}}l@{}}
                \tikz{\node[cross out, draw=red, thick, minimum size=6pt, inner sep=0pt] {};} & \textcolor{red}{:} & \textcolor{red}{semisimple eigenvalue} \\[1pt]
                \tikz{\node[diamond, draw=red, thick, fill=red!30, minimum size=10pt, inner sep=0pt] {};} & \textcolor{red}{:} & \textcolor{red}{non-semisimple eigenvalue}
            \end{tabular}
        };
    
    \end{tikzpicture}
    \caption{A schematic figure for three kinds of gaps in Theorems~\ref{thm:main_theorem} and \ref{thm:main_theorem2}: $\Delta_\Lambda$, $\Delta_\Lambda^{\mathrm{p}}$, and $\Delta_\Lambda^{\mathrm{ex}}$ for a finite subset $\Lambda\Subset\Gamma$.}
    \label{fig:gaps}
\end{figure}
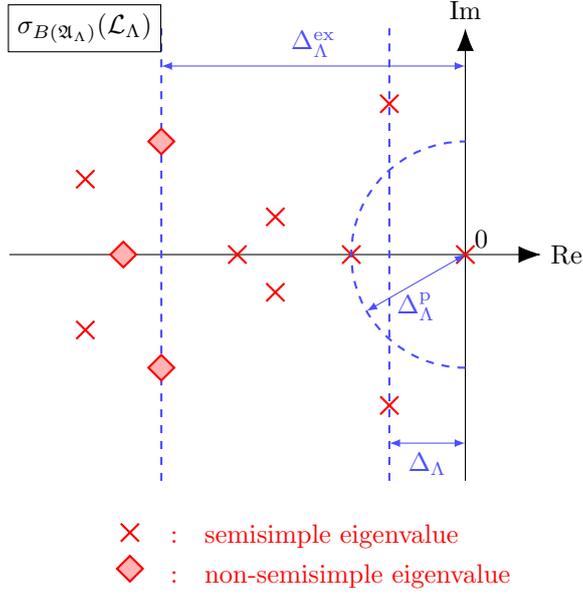

We also give some comments on the two assumptions in Theorem~\ref{thm:main_theorem}.
See Figure ~\ref{fig:gaps} for the gaps $\Delta_\Lambda$ and $\Delta_\Lambda^{\mathrm{ex}}$ that appear in the assumptions.
The first assumption (A1) states that the Liouvillians have a uniform spectral gap for any finite subset $\Lambda\Subset\Gamma$.
It has been well-known in various branches of mathematics that a spectral gap in a generator of time evolution controls the dynamic behavior of the corresponding system, such as mixing times in Markov chains (see e.g. \cite{Levin2017}) and decay of correlations in closed quantum many-body systems described as $C^*$-dynamical systems (see e.g. \cite{Hastings2006}), so that this assumption (A1) is a common characteristic with several other settings.
Specifically, the notion of spectral gap in $C^*$-dynamical systems plays a central role in operator algebraic investigations for closed quantum lattice systems as in \cite{Affleck1988,Matsui2001,Bachmann2012,Matsui2013,Nachtergaele2019,Nachtergaele2022,Kapustin2022,Ogata2023,Hastings2015,Bachmann2020,Prodan2016}.
Furthermore, some methods are known to evaluate the spectral gap.
An example is logarithmic Sobolev inequalities, which have been used for various models from classical stochastic processes to quantum dynamical semigroups, such as Markov processes on classical spin systems in \cite{Holley1987,Stroock1992}, exclusion processes in \cite{Bertini2004}, and Lindblad dynamics in \cite{Kastoryano2013}.

In contrast, the second assumption (A2) appears only in the case of open quantum systems considered in this paper because closed quantum systems, which have self-adjoint Liouvillians, automatically satisfy this.
The quantity 
\begin{align*}
    \kappa_{\Lambda'}(\mathcal{V}_\Lambda^\epsilon)
    = \|\mathcal{V}_\Lambda^\epsilon\restriction\mathfrak{A}_{\Lambda'}\|_{\mathfrak{A}_{\Lambda'}\to\C^{d^{2|\Lambda|}}}\|(\mathcal{V}_\Lambda^\epsilon)^{-1}\|_{\C^{d^{2|\Lambda|}}\to\mathfrak{A}_\Lambda}
\end{align*}
is so-called \emph{condition number} of $\mathcal{V}_\Lambda^\epsilon$.
The condition number measures the normality of an operator and controls the stability of spectrum of an operator against perturbations (see e.g. \cite{Trefethen2005}), so it has often been utilized in applied mathematics, particularly the field of numerical analysis \cite{Saad2011}, and non-Hermitian physics \cite{Nakai2024,Shimomura2024,Ughrelidze2024}.
The assumption (A2) requests the condition numbers for the invertible matrices that transform the (normalized) Liouvillians into the Jordan canonical forms should be uniformly bounded for any finite subsets $\Lambda\Subset\Gamma$.
In Section~\ref{sec:example}, we will see that a concrete model has nonzero $\Delta,\Delta^{\text{ex}}$ and unbounded condition number, and thereby violates the assumption (A2), yielding an NESS that differs from Eq.~\eqref{eq:LTLthenTDL}.

When we define the condition number of $\mathcal{V}_\Lambda^\epsilon$, there is an ambiguity of the transformation $\mathcal{V}_\Lambda^\epsilon\mapsto\mathcal{A}\circ\mathcal{V}_\Lambda^\epsilon$ by a linear map $\mathcal{A}\in B(\mathbb{C}^{d^{2|\Lambda|}})$ respecting a symmetry $\mathcal{A}\circ\mathcal{J}_\Lambda^\epsilon\circ\mathcal{A}^{-1}=\mathcal{J}_\Lambda^\epsilon$.
Since $\mathcal{A}$ is not necessarily unitary, the transformation generically changes the value of condition number, $\kappa_\Lambda(\mathcal{A}\circ\mathcal{V}_\Lambda^\epsilon)\neq\kappa_\Lambda(\mathcal{V}_\Lambda^\epsilon)$.
Even when $\kappa_\Lambda(\mathcal{V}_\Lambda^\epsilon)$ is not bounded in $\Lambda$, we may have bounded $\kappa_\Lambda(\mathcal{A}_\Lambda^\epsilon\circ\mathcal{V}_\Lambda^\epsilon)$ after an adequate transformation $\mathcal{V}_\Lambda^\epsilon\mapsto\mathcal{A}_\Lambda^\epsilon\circ\mathcal{V}_\Lambda^\epsilon$.
A standard choice of $\mathcal{A}\circ\mathcal{V}_\Lambda^\epsilon$ has been discussed in \cite{Sluis1969,Nakai2024} for diagonalizable $\mathcal{L}_\Lambda$.

The second assumption (A2) is worth further considering.
The number $\Delta^{\mathrm{ex}}$ is a uniform gap of non-semisimple eigenvalues of $\mathcal{L}_\Lambda$ for any $\Lambda\Subset\Gamma$, so we call it the \emph{exceptional gap} named after the exceptional point introduced in \cite{Kato1995}.
The existence of $\Delta^{\mathrm{ex}}$ that meets Eq.~\eqref{eq:existence_of_ex_gap} is automatically achieved under the first assumption (A1).
Still, it is nontrivial whether we can take an adequate value of $\Delta^{\mathrm{ex}}$ for which the condition number $\kappa_\Lambda(\mathcal{V}_\Lambda^\epsilon)$ is bounded by a finite constant.

If $\mathcal{L}_\Lambda$ is diagonalizable for all $\Lambda\Subset\Gamma$, Theorem~\ref{thm:main_theorem} is reduced to the following form.
\begin{cor}\label{cor:main_theorems}
    Let $(\mathcal{L}_\Lambda)_{\Lambda\Subset\Gamma}$ be a family of Liouvillians on finite subsets satisfying Assumption \ref{asm:locality}.
    For each finite subset $\Lambda\Subset\Gamma$, let $\sigma_{B(\mathfrak{A}_\Lambda)}(\mathcal{L}_\Lambda)$ be the spectrum of the linear map $\mathcal{L}_\Lambda$ from the finite-dimensional vector space $\mathfrak{A}_\Lambda$ to itself.
    Suppose that
    \begin{enumerate}[($\text{A}$1)]
        \setcounter{enumi}{-1}
        \item for any finite subset $\Lambda\Subset\Gamma$, $\mathcal{L}_\Lambda$ is diagonalizable;
        \item there exists a positive number $\Delta$ such that for any finite subset $\Lambda\Subset\Gamma$ we have
        \begin{align*}
            \Delta_\Lambda
            \coloneqq \min\{ |\operatorname{Re}\lambda| \mid \lambda\in\sigma_{B(\mathfrak{A}_\Lambda)}(\mathcal{L}_\Lambda)\setminus\{0\} \}
            \ge \Delta
            > 0;
        \end{align*}
        \item for any finite subset $\Lambda'\Subset\Gamma$ we have 
        \begin{align*}
            \kappa_{\Lambda'}
            \coloneqq \sup_{\substack{\Lambda\Subset\Gamma \\ \text{s.t. }\Lambda'\subset\Lambda}}\|\mathcal{V}_\Lambda\restriction\mathfrak{A}_{\Lambda'}\|_{\mathfrak{A}_{\Lambda'}\to\C^{d^{2|\Lambda|}}}\|(\mathcal{V}_\Lambda)^{-1}\|_{\C^{d^{2|\Lambda|}}\to\mathfrak{A}_\Lambda}
            < \infty,
        \end{align*}
        where $\mathcal{V}_\Lambda\in B(\mathfrak{A}_\Lambda,\C^{d^{2|\Lambda|}})$ is an invertible linear map that diagonalizes $\mathcal{L}_\Lambda$.
        Here, the map $\mathcal{V}_\Lambda\restriction\mathfrak{A}_{\Lambda'}$ is the restriction of $\mathcal{V}_\Lambda$ onto $\mathfrak{A}_{\Lambda'}(\subset\mathfrak{A}_\Lambda)$.
    \end{enumerate}
    Then, for any state $\omega$ over the quantum spin system $\mathfrak{A}$, the limit
    \begin{align*}
        \wslim{\Lambda} \wslim{t\to\infty}\omega\circ\gamma_t^{\Lambda}
    \end{align*}
    exists and is the unique NESS for the initial state $\omega$. 
\end{cor}

We also have a similar condition on the TANESS as Theorem~\ref{thm:main_theorem}.

\begin{thm}\label{thm:main_theorem2}
    Let $(\mathcal{L}_\Lambda)_{\Lambda\Subset\Gamma}$ be a family of Liouvillians on finite subsets satisfying Assumption \ref{asm:locality}.
    For each finite subset $\Lambda\Subset\Gamma$, let $\sigma_{B(\mathfrak{A}_\Lambda)}(\mathcal{L}_\Lambda)$ be the spectrum of the linear map $\mathcal{L}_\Lambda$ from the finite-dimensional vector space $\mathfrak{A}_\Lambda$ to itself
    and $\sigma_{B(\mathfrak{A}_\Lambda)}^{\mathrm{ex}}$ the set of all non-semisimple eigenvalues of $\mathcal{L}_\Lambda$.
    Suppose that
    \begin{enumerate}[($\text{B}$1)]
        \item there exists a positive number $\Delta^{\mathrm{p}}$ such that for any finite subset $\Lambda\Subset\Gamma$ we have
        \begin{align}\label{eq:point_gap}
            \Delta_\Lambda^{\mathrm{p}}
            \coloneqq \min\{ |\lambda| \mid \lambda\in\sigma_{B(\mathfrak{A}_\Lambda)}(\mathcal{L}_\Lambda)\setminus\{0\} \}
            \ge \Delta^{\mathrm{p}}
            > 0;
        \end{align}
        \item there exist (finite) positive numbers $\Delta^{\mathrm{ex}}$ and $\epsilon\in (0,\Delta^{\mathrm{ex}})$ such that we have both
        \begin{align}\label{eq:existence_of_ex_gap2}
            \Delta_\Lambda^{\mathrm{ex}}
            \coloneqq \min\{|\operatorname{Re}\lambda| \mid\lambda\in\sigma_{B(\mathfrak{A}_\Lambda)}^{\mathrm{ex}}(\mathcal{L}_\Lambda)\}
            \ge \Delta^{\mathrm{ex}}
            > 0
        \end{align}
        for any finite subset $\Lambda\Subset\Gamma$ and
        \begin{align*}
            \kappa_{\Lambda'}^\epsilon
            \coloneqq \sup_{\substack{\Lambda\Subset\Gamma \\ \text{s.t. }\Lambda'\subset\Lambda}}\|\mathcal{V}_\Lambda^\epsilon\restriction\mathfrak{A}_{\Lambda'}\|_{\mathfrak{A}_{\Lambda'}\to\C^{d^{2|\Lambda|}}}\|(\mathcal{V}_\Lambda^\epsilon)^{-1}\|_{\C^{d^{2|\Lambda|}}\to\mathfrak{A}_\Lambda}
            < \infty
        \end{align*}
        for any finite subset $\Lambda'\Subset\Gamma$, where $\mathcal{V}_\Lambda^\epsilon\in B(\mathfrak{A}_\Lambda,\C^{d^{2|\Lambda|}})$ is an invertible linear map that transforms $(\Delta^{\mathrm{ex}}-\epsilon)^{-1}\mathcal{L}_\Lambda$ into the Jordan canonical form $\mathcal{J}_\Lambda^{\epsilon}\in\mathrm{M}_{d^{2|\Lambda|}}$,
        \begin{align*}
            \mathcal{J}_\Lambda^\epsilon
            = \mathcal{V}_\Lambda^\epsilon \circ \frac{1}{\Delta^{\mathrm{ex}}-\epsilon}\mathcal{L}_\Lambda \circ (\mathcal{V}_\Lambda^\epsilon)^{-1},
        \end{align*}
        and the norms $\|\bullet\|_{\mathfrak{A}_\Lambda\to\C^{d^{2|\Lambda|}}}$ and $\|\bullet\|_{\C^{d^{2|\Lambda|}}\to\mathfrak{A}_\Lambda}$ are defined by
        \begin{gather*}
            \|\mathcal{S}\|_{\mathfrak{A}_\Lambda\to\C^{d^{2|\Lambda|}}}
            = \sup\{ {\|\mathcal{S}(\hat{A})\|}_2 \mid \hat{A}\in\mathfrak{A}_\Lambda,\ \|\hat{A}\|\le 1 \},
            \quad
            \mathcal{S}\colon\mathfrak{A}_\Lambda\to\C^{d^{2|\Lambda|}}, \\
            \|\mathcal{T}\|_{\C^{d^{2|\Lambda|}}\to\mathfrak{A}_\Lambda}
            = \sup\{ \|\mathcal{T}(\boldsymbol{v})\| \mid \boldsymbol{v}\in\C^{d^{2|\Lambda|}},\ {\|\boldsymbol{v}\|}_2\le 1 \},
            \quad
            \mathcal{T}\colon\C^{d^{2|\Lambda|}}\to\mathfrak{A}_\Lambda.
        \end{gather*}
        Here, $\|\bullet\|_2$ is the 2-norm of numerical vectors.
        The map $\mathcal{V}_\Lambda^\epsilon\restriction\mathfrak{A}_{\Lambda'}$ is the restriction of $\mathcal{V}_\Lambda^\epsilon$ onto $\mathfrak{A}_{\Lambda'}(\subset\mathfrak{A}_\Lambda)$.
    \end{enumerate}
    Then, for any state $\omega$ over the quantum spin system $\mathfrak{A}$, the limit
    \begin{align}\label{eq:LTLthenTDL2}
        \wslim{\Lambda} \wslim{T\to\infty}\frac{1}{T}\int_0^T\omega\circ\gamma_t^{\Lambda}\dif t
    \end{align}
    exists and is the unique TANESS for the initial state $\omega$. 
\end{thm}

Theorem~\ref{thm:main_theorem2} implies that the thermodynamic limit and the long-time limit commute under the conditions (B1) and (B2), similarly to Theorem~\ref{thm:main_theorem}.
In addition, it also implies the thermodynamic limit and the integral of time averaging commute.
Again, note that Theorem~\ref{thm:main_theorem2} allows different TANESSs to exist for different initial states.

Let us compare the assumptions in Theorem~\ref{thm:main_theorem2} with those in Theorem~ \ref{thm:main_theorem}.
See Figure~\ref{fig:gaps} again for comparison of three gaps $\Delta_\Lambda$, $\Delta_\Lambda^{\mathrm{p}}$, and $\Delta_\Lambda^{\mathrm{ex}}$.
Whereas the gap $\Delta_\Lambda$ in the assumption (A1) is the distance between the set $\sigma_{B(\mathfrak{A}_\Lambda)}(\mathcal{L}_\Lambda)\setminus\{0\}$ and the imaginary axis, the gap $\Delta_\Lambda^{\mathrm{p}}$ in the assumption (B1) is that between $\sigma_{B(\mathfrak{A}_\Lambda)}(\mathcal{L}_\Lambda)\setminus\{0\}$ and zero.
Hence, $\Delta_\Lambda$ and $\Delta$ are called the \textit{line gap}, while $\Delta_\Lambda^{\mathrm{p}}$ and $\Delta^{\mathrm{p}}$ the \textit{point gap} \cite{Kawabata2019}.
Theorems~\ref{thm:main_theorem} and \ref{thm:main_theorem2} yield that the point gap controls the TANESS, while the line gap does the NESS.
The condition (B1) does not imply the existence of 
exceptional gap in Eq.~\eqref{eq:existence_of_ex_gap2}, unlike (A1).

\subsection{Behavior of Lindblad dynamics on finite subsets}\label{subsec:behavior_on_finite_subsets}
In preparation for the proof of Theorems~\ref{thm:main_theorem} and \ref{thm:main_theorem2}, we show some lemmas about Lindblad dynamics on finite subsets.
The first lemma concerns the spectrum of a Liouvillian $\mathcal{L}_\Lambda$ for finite $\Lambda\Subset\Gamma$.
Within this subsection, Assumption~\ref{asm:locality} on the locality of Liouvillians is not required.

\begin{lem}\label{lem:spectrum_in_finite_systems}
    For each finite subset $\Lambda\Subset\Gamma$, the Liouvillian $\mathcal{L}_\Lambda\colon\mathfrak{A}_\Lambda\to\mathfrak{A}_\Lambda$ has an eigenvalue $0$, every purely imaginary eigenvalue (including $0$) of $\mathcal{L}_\Lambda$ is semisimple, and all of the eigenvalues of $\mathcal{L}_\Lambda$ have non-positive real parts.
\end{lem}
\begin{proof}
    The proof of this lemma is inspired by the proof in \cite[Lemma 4]{Minganti2018}.
    The essence of this proof is that for every $t\ge 0$
    \begin{align}\label{eq:unity_of_norm_of_time_evolution}
        \|e^{t\mathcal{L}_\Lambda}\|_{B(\mathfrak{A}_\Lambda)}
        = \|e^{t\mathcal{L}_\Lambda}(\hat{I})\|
        = \|\hat{I}\|
        = 1
    \end{align}
    holds because $\mathcal{L}_\Lambda$ generates the completely positive map $e^{t\mathcal{L}_\Lambda}$ over the finite dimensional $C^*$-algebra $\mathfrak{A}_\Lambda$.

    We can check from $\mathcal{L}_\Lambda(\hat{I})=0$ that the Liouvillian $\mathcal{L}_\Lambda$ has the eigenvalue 0. 

    We prove the remaining statement of the lemma by contradiction.
    Let $\ii r\in\ii\mathbb{R}$ be a purely imaginary eigenvalue of $\mathcal{L}_\Lambda$.
    We assume that the eigenvalue $\ii r$ is not semisimple and thus there exists a nonzero operator $\hat{A}\in\ker(\mathcal{L}_\Lambda-\ii r\id_{\mathfrak{A}_\Lambda})^2\setminus\ker(\mathcal{L}_\Lambda-\ii r\id_{\mathfrak{A}_\Lambda})\subset\mathfrak{A}_\Lambda$.
    For any $t\ge 0$ we have
    \begin{align*}
        \|e^{t\mathcal{L}_\Lambda}(\hat{A})\| - \|\hat{A}\|
        &= \|e^{-\ii rt}e^{t\mathcal{L}_\Lambda}(\hat{A})\| - \|\hat{A}\| \\
        &= \|\hat{A}+t(\mathcal{L}_\Lambda-\ii r\id_{\mathfrak{A}_\Lambda})(\hat{A})\| - \|\hat{A}\| \\
        &\ge t\|(\mathcal{L}_\Lambda-\ii r\id_{\mathfrak{A}_\Lambda})(\hat{A})\| -2\|\hat{A}\|.
    \end{align*}
    Since $(\mathcal{L}_\Lambda-\ii r\id_{\mathfrak{A}_\Lambda})(\hat{A})\neq 0$ by construction, $\|e^{t\mathcal{L}_\Lambda}(\hat{A})\|>\|\hat{A}\|$ holds for a finite $t$ satisfying
    \begin{align*}
        t>\frac{2\|\hat{A}\|}{\|(\mathcal{L}_\Lambda-\ii r\id_{\mathfrak{A}_\Lambda})(\hat{A})\|},
    \end{align*}
    which contradicts Eq.~\eqref{eq:unity_of_norm_of_time_evolution}.

    Next, we assume that $\mathcal{L}_\Lambda$ has an eigenvalue $\lambda$ such that $\operatorname{Re}\lambda > 0$.
    Let $\hat{B}\in\mathfrak{A}_\Lambda\setminus\{0\}$ be an eigenoperator associated with the eigenvalue $\lambda$, $\mathcal{L}_\Lambda(\hat{B})=\lambda\hat{B}$. 
    For any $t>0$ we have
    \begin{align*}
        \|e^{t\mathcal{L}_\Lambda}(\hat{B})\| - \|\hat{B}\|
        = \|e^{t\lambda}\hat{B}\| - \|\hat{B}\|
        = (e^{t\operatorname{Re}\lambda}-1)\|\hat{B}\|
        > 0,
    \end{align*}
    which contradicts Eq.~\eqref{eq:unity_of_norm_of_time_evolution}.
\end{proof}

Next, we argue the long-time limit of $\omega\circ\gamma_t^\Lambda$ and $T^{-1}\int_0^T\omega\circ\gamma_t^\Lambda\dif t$ for finite $\Lambda\Subset\Gamma$.
\begin{lem}\label{lem:NESS_in_finite_systems}
    Let $\Pi_\Lambda\colon\mathfrak{A}\to\mathfrak{A}$ be the extension onto the quantum spin chain $\mathfrak{A}$ of the spectral projection $\pi_\Lambda\colon\mathfrak{A}_\Lambda\to\mathfrak{A}_\Lambda$ of the Liouvillian $\mathcal{L}_\Lambda$ associated with the eigenvalue 0.
    For each finite $\Lambda\Subset\Gamma$ and any state $\omega$, the following holds.
    \begin{enumerate}[(1)]
        \item If $\Delta_\Lambda>0$, then it holds that
        \begin{align}\label{eq:finite_NESS}
            \lim_{t\to\infty}\|\omega\circ\Pi_\Lambda-\omega\circ\gamma_t^\Lambda\|
            = 0
        \end{align}
        or in particular
        \begin{align*}
            \omega\circ\Pi_\Lambda
            = \wslim{t\to\infty} \omega\circ\gamma_t^\Lambda.
        \end{align*}
        \item It holds that 
        \begin{align}\label{eq:finite_TANESS}
            \lim_{T\to\infty}\left\| \omega\circ\Pi_\Lambda - \frac{1}{T}\int_0^T\omega\circ\gamma_t^\Lambda\dif t \right\|
            = 0
        \end{align}
        or in particular
        \begin{align*}
            \omega\circ\Pi_\Lambda
            = \wslim{T\to\infty} \frac{1}{T}\int_0^T\omega\circ\gamma_t^\Lambda\dif t.
        \end{align*}
        In particular, $\omega\circ\Pi_\Lambda$ is a state on $\mathfrak{A}$.
    \end{enumerate}
\end{lem}
\begin{proof}
    Fix a finite subset $\Lambda\Subset\Gamma$.
    Since the assertion is obvious if $\mathcal{L}_\Lambda=0$, we suppose $\mathcal{L}_\Lambda\neq 0$ hereafter in this proof.

    First, we assume $\Delta_\Lambda>0$ and show Eq.~\eqref{eq:finite_NESS}.
    We transform $(\Delta_\Lambda/2)^{-1}\mathcal{L}_\Lambda$ into the Jordan canonical form,
    \begin{align}\label{eq:Jordan_canonical_form1}
        \bigoplus_{h=0}^m \left[ (\Delta_\Lambda/2)^{-1}\lambda_h I_h + N_h \right]
        = \mathcal{V}\circ(\Delta_\Lambda/2)^{-1}\mathcal{L}_\Lambda\circ\mathcal{V}^{-1},
    \end{align}
    where $\lambda_0,\lambda_1,\ldots,\lambda_m$ are all distinct eigenvalues of $\mathcal{L}_\Lambda$, $I_h$ is the identity matrix, $N_h$ is a nilpotent matrix whose upper diagonal entries are one or zero but all other entries are zero, and $\mathcal{V}$ is an invertible linear map from $\mathfrak{A}_\Lambda$ to $\C^{d^{2|\Lambda|}}$.
    Note that $\|N_h\|\le 1$ holds because
    \begin{align*}
        \|N_h\|^2
        \le \|N_h\|_1 \|N_h\|_\infty
        = \left( \max_j\sum_i \left|[N_h]_{i,j}\right| \right) \left( \max_i\sum_j \left|[N_h]_{i,j}\right| \right)
        \le 1,
    \end{align*}
    where $[N_h]_{i,j}$ denotes the $(i,j)$-element of $N_h$.
    Since $\mathcal{L}_\Lambda$ has the semisimple eigenvalue 0, we can take $\lambda_0=0$ and $N_0=0$ without loss of generality.
    From $\|N_h\|\le 1$ and $\operatorname{Re}\lambda_h\le -\Delta_\Lambda$ ($h=1,2,\ldots,m$), it follows that
    \begin{align*}
        \|e^{t\mathcal{L}_\Lambda} - \pi_\Lambda\|_{B(\mathfrak{A}_\Lambda)}
        &= \left\| \mathcal{V}^{-1} \circ \left( 0I_0 \oplus \bigoplus_{h=1}^m e^{t\lambda_h}e^{(\Delta_\Lambda/2)tN_h} \right) \circ \mathcal{V} \right\|_{B(\mathfrak{A}_\Lambda)} \\
        &\le \|\mathcal{V}^{-1}\|_{\C^{d^{2|\Lambda|}}\to\mathfrak{A}_\Lambda}\|\mathcal{V}\|_{\mathfrak{A}_\Lambda\to\C^{d^{2|\Lambda|}}}\max_{h=1,2,\ldots,m}|e^{t\lambda_h}|{\|e^{(\Delta_\Lambda/2)tN_h}\|} \\
        &\le \kappa_\Lambda(\mathcal{V})\max_{h=1,2,\ldots,m}e^{t\operatorname{Re}\lambda_h}e^{(\Delta_\Lambda/2)t\|N_h\|} \\
        &\le \kappa_\Lambda(\mathcal{V})e^{-t\Delta_\Lambda/2}.
    \end{align*}
    Here, we used the condition number 
    \begin{align*}
        \kappa_\Lambda(\mathcal{V}) 
        = \|\mathcal{V}\restriction\mathfrak{A}_\Lambda\|_{\mathfrak{A}_\Lambda\to\mathbb{C}^{d^{2|\Lambda|}}}\|\mathcal{V}^{-1}\|_{\mathbb{C}^{d^{2|\Lambda|}}\to\mathfrak{A}_\Lambda} 
        = \|\mathcal{V}\|_{\mathfrak{A}_\Lambda\to\mathbb{C}^{d^{2|\Lambda|}}}\|\mathcal{V}^{-1}\|_{\mathbb{C}^{d^{2|\Lambda|}}\to\mathfrak{A}_\Lambda}.
    \end{align*}
    Hence, we have
        \begin{align*}
            \|\omega\circ\gamma_t^\Lambda - \omega\circ\Pi_\Lambda\|
            &\le \|\gamma_t^\Lambda-\Pi_\Lambda\|_{B(\mathfrak{A})} \\
            &= \|e^{t\mathcal{L}_\Lambda}-\pi_\Lambda\|_{\cb} \\
            &\le d^{|\Lambda|}\|e^{t\mathcal{L}_\Lambda}-\pi_\Lambda\|_{B(\mathfrak{A}_\Lambda)} \\
            &\le d^{|\Lambda|}\kappa_\Lambda(\mathcal{V})e^{-t\Delta_\Lambda/2}
            \xrightarrow{t\to\infty} 0
        \end{align*}
    for any state $\omega$ on $\mathfrak{A}$,
    which means that the net $(\omega\circ\gamma_t^\Lambda)_{t\ge 0}$ of states converges to $\omega\circ\Pi_\Lambda$ in the norm topology, and hence in particular the weak-$*$ topology.

    Next, we show Eq.~\eqref{eq:finite_TANESS} without the assumption $\Delta_\Lambda>0$.
    We have $\Delta_\Lambda^{\mathrm{p}}>0$ by the definition in Eq.~\eqref{eq:point_gap} and $\Delta_\Lambda^{\mathrm{ex}}>0$ from Lemma~\ref{lem:spectrum_in_finite_systems}. 
    We can take a finite positive value $\Delta'_\Lambda$ satisfying $0<\Delta'_\Lambda<\Delta_\Lambda^{\mathrm{ex}}$.
    Again, we transform $(\Delta'_\Lambda/2)^{-1}\mathcal{L}_\Lambda$ into the Jordan canonical form,
    \begin{align*}
        \bigoplus_{h=0}^m \left[ (\Delta'_\Lambda/2)^{-1}\lambda_h I_h + N_h \right]
        = \mathcal{V}'\circ(\Delta'_\Lambda/2)^{-1}\mathcal{L}_\Lambda\circ(\mathcal{V}')^{-1},
    \end{align*}
    where $\lambda_h$, $I_h$, and $N_h$ are same as in Eq.~\eqref{eq:Jordan_canonical_form1} and $\mathcal{V}'$ is an invertible linear map from $\mathfrak{A}_\Lambda$ to $\C^{d^{2|\Lambda|}}$.
    We evaluate the norm of the difference between $T^{-1}\int_0^T e^{t\mathcal{L}_\Lambda}\dif t$ and $\pi_\Lambda$ as
    \begin{align*}
        \left\| \frac{1}{T}\int_0^T e^{t\mathcal{L}_\Lambda}\dif t - \pi_\Lambda \right\|_{B(\mathfrak{A}_\Lambda)}
        &= \left\| (\mathcal{V}')^{-1} \circ \frac{1}{T}\int_0^T\left( 0I_0\oplus\bigoplus_{h=1}^m e^{t\lambda_h}e^{(\Delta'_\Lambda/2)tN_h} \right)\dif t \circ \mathcal{V}' \right\|_{B(\mathfrak{A}_\Lambda)} \\
        &\le \kappa_\Lambda(\mathcal{V}')\max_{h=1,2,\ldots,m}\frac{1}{T}\left\|\int_0^T e^{t\lambda_h}e^{(\Delta'_\Lambda/2)tN_h}\dif t\right\|.
    \end{align*}
    To calculate the term $T^{-1}\|\int_0^T e^{t\lambda_h}e^{(\Delta'_\Lambda/2)tN_h}\dif t\|$, we consider the following two cases.
    \begin{description}
        \item[If the eigenvalue $\lambda_h\neq 0$ is semisimple,] we have $N_h=0$ and hence obtain
        \begin{align}\label{eq:semisimple_case}
            \left\|\int_0^T e^{t\lambda_h}e^{(\Delta'_\Lambda/2)tN_h}\dif t\right\|
            = \left|\int_0^T e^{t\lambda_h}\dif t\right|
            = \frac{|1-e^{T\lambda_h}|}{|\lambda_h|}
            \le \frac{2}{|\lambda_h|}
            \le \frac{2}{\Delta_\Lambda^{\mathrm{p}}}.
        \end{align}
        \item[If the eigenvalue $\lambda_h\neq 0$ is not semisimple,] the fact that $\|N_h\|\le 1$ and $\operatorname{Re}\lambda_h\le -\Delta'_\Lambda<0$ yields
        \begin{align}\label{eq:non-semisimple_case}
            \left\|\int_0^T e^{t\lambda_h}e^{(\Delta'_\Lambda/2)tN_h}\dif t\right\|
            \le \int_0^T |e^{t\lambda_h}|e^{(\Delta'_\Lambda/2)t\|N_h\|}\dif t
            \le \int_0^T e^{-t\Delta'_\Lambda/2}\dif t
            = \frac{1-e^{-T\Delta'_\Lambda/2}}{\Delta'_\Lambda/2}
            \le \frac{2}{\Delta'_\Lambda}.
        \end{align}
    \end{description}
    Therefore, from Eqs.~\eqref{eq:semisimple_case} and \eqref{eq:non-semisimple_case}, we have
    \begin{align*}
        \left\| \frac{1}{T}\int_0^T e^{t\mathcal{L}_\Lambda}\dif t - \pi_\Lambda \right\|_{B(\mathfrak{A}_\Lambda)}
        \le \frac{2\kappa_\Lambda(\mathcal{V}')\max\{(\Delta_\Lambda^{\mathrm{p}})^{-1},(\Delta'_\Lambda)^{-1}\}}{T}.
    \end{align*}
    For each $\hat{A}\in\mathfrak{A}$, we consider the Bochner integral $\int_0^T \gamma_t^\Lambda(\hat{A})\dif t$ as mentioned in Remark~\ref{rmk:Bochner} and then perform the following evaluation,
    \begin{align*}
        \left| \frac{1}{T}\int_0^T\omega\circ\gamma_t^\Lambda(\hat{A})\dif t - \omega\circ\Pi_\Lambda(\hat{A}) \right|
        &= \left| \omega\circ\left( \frac{1}{T}\int_0^T\gamma_t^\Lambda(\hat{A})\dif t - \Pi_\Lambda(\hat{A}) \right) \right| \\
        &\le \left\| \frac{1}{T}\int_0^T\gamma_t^\Lambda(\hat{A})\dif t - \Pi_\Lambda(\hat{A}) \right\| \\
        &\le \left\| \frac{1}{T}\int_0^T e^{t\mathcal{L}_\Lambda}\dif t - \pi_\Lambda \right\|_{\cb}\|\hat{A}\| \\
        &\le d^{|\Lambda|}\left\| \frac{1}{T}\int_0^T e^{t\mathcal{L}_\Lambda}\dif t - \pi_\Lambda \right\|_{B(\mathfrak{A}_\Lambda)}\|\hat{A}\| \\
        &\le d^{|\Lambda|}\frac{2\kappa_\Lambda(\mathcal{V}')\max\{(\Delta_\Lambda^{\mathrm{p}})^{-1},(\Delta'_\Lambda)^{-1}\}}{T}\|\hat{A}\|,
    \end{align*}
    which implies that
    \begin{align*}
        \left\| \frac{1}{T}\int_0^T\omega\circ\gamma_t^\Lambda\dif t - \omega\circ\Pi_\Lambda \right\|
        \le d^{|\Lambda|}\frac{2\kappa_\Lambda(\mathcal{V}')\max\{(\Delta_\Lambda^{\mathrm{p}})^{-1},(\Delta'_\Lambda)^{-1}\}}{T}
        \xrightarrow{T\to\infty} 0.
    \end{align*}
    Therefore, the net $\left(T^{-1}\int_0^T\omega\circ\gamma_t^\Lambda\dif t\right)_{T\ge 0}$ of states converges to $\omega\circ\Pi_\Lambda$ in the norm topology, and hence in particular the weak-$*$ topology.
    Since the set of states over the $C^*$-algebra $\mathfrak{A}$ is closed in terms of the weak-$*$ topology, $\omega\circ\Pi_\Lambda$ is a state over $\mathfrak{A}$. 
\end{proof}

As seen in this lemma, the nonequilibrium steady state in the finite system is determined from the spectral projection of the Liouvillian, which can be calculated by an algebraic method of the Jordan canonical form.
It is remarkable that Theorems~\ref{thm:main_theorem} and 
\ref{thm:main_theorem2} show that the NESS and TANESS on the quantum spin system of an infinite size, originally defined as an fully analytic object, can be characterized partially in the algebraic way under some assumptions.

\subsection{Proof of main theorems}\label{subsec:proof_of_main_theorems}
Using the lemmas shown in Sec.~\ref{subsec:behavior_on_finite_subsets}, we prove the main theorems.

\begin{proof}[Proof of Theorem \ref{thm:main_theorem}]
    Let $\omega$ be an arbitrary state.
    As shown in Lemma~\ref{lem:NESS_in_finite_systems}, $\omega\circ\Pi_\Lambda=\wslim{t\to\infty}\omega\circ\gamma_t^\Lambda$ for every finite subset $\Lambda\Subset\Gamma$ is a state over $\mathfrak{A}$ and thus there exists a subnet $(\Lambda_\mu)_\mu$ of the net $(\Lambda)_{\Lambda\Subset\Gamma}$ such that the limit
    \begin{align*}
        \bar{\omega}
        = \wslim{\mu}\omega\circ\Pi_{\Lambda_\mu}
    \end{align*}
    exists as a state over $\mathfrak{A}$. 
    
    First, we see $\bar{\omega}=\omega_{\sss}$ for any $\omega_{\sss}\in\Sigma_{\NESS}(\omega)$.
    We take any operator $\hat{A}\in\mathfrak{A}_{\loc}$ with a finite support $\supp\hat{A}\Subset\Gamma$ and any finite subset $\Lambda$ of $\Gamma$ that includes $\supp\hat{A}$, i.e.,
    \begin{align*}
        \supp\hat{A}
        \subset \Lambda
        \Subset \Gamma.
    \end{align*}
    Then, we have
    \begin{align*}
        |\omega\circ\gamma_t^\Lambda(\hat{A})-\omega\circ\Pi_\Lambda(\hat{A})|
        \le \left\|(e^{t\mathcal{L}_\Lambda}-\pi_\Lambda)\restriction\mathfrak{A}_{\supp\hat{A}}\right\|_{\mathfrak{A}_{\supp\hat{A}}\to\mathfrak{A}_\Lambda}\|\hat{A}\|,
    \end{align*}
    where for a linear map $\mathcal{T}\colon\mathfrak{A}_{\Lambda'}\to\mathfrak{A}_{\Lambda}$ we define the norm $\|\mathcal{T}\|_{\mathfrak{A}_{\Lambda'}\to\mathfrak{A}_\Lambda}\coloneqq \sup\{ \|\mathcal{T}(\hat{A})\| \mid \hat{A}\in\mathfrak{A}_{\Lambda'},\|\hat{A}\|\le 1 \}$.
    We transform $(\Delta^{\mathrm{ex}}-\epsilon)^{-1}\mathcal{L}_\Lambda$ into the Jordan canonical form as
    \begin{align*}
        \bigoplus_{h=0}^m \left[ (\Delta^{\mathrm{ex}}-\epsilon)^{-1}\lambda_h I_h+N_h \right]
        = \mathcal{V}_\Lambda^\epsilon\circ(\Delta^{\mathrm{ex}}-\epsilon)^{-1}\mathcal{L}_\Lambda\circ(\mathcal{V}_\Lambda^\epsilon)^{-1},
    \end{align*}
    where we use the same notations as in Eq.~\eqref{eq:Jordan_canonical_form1} for $\lambda_h$, $I_h$, and $N_h$.
    From the assumption (A2), i.e.,
    \begin{align*}
        \|\mathcal{V}_\Lambda^\epsilon\restriction\mathfrak{A}_{\supp\hat{A}}\|_{\mathfrak{A}_{\supp\hat{A}}\to\C^{d^{2|\Lambda|}}}\|(\mathcal{V}_\Lambda^\epsilon)^{-1}\|_{\C^{d^{2|\Lambda|}}\to\mathfrak{A}_\Lambda}
        \le \kappa_{\supp\hat{A}}^\epsilon
        < \infty,
    \end{align*}
    it follows that
    \begin{align*}
        &\left\|(e^{t\mathcal{L}_\Lambda}-\pi_\Lambda)\restriction\mathfrak{A}_{\supp\hat{A}}\right\|_{\mathfrak{A}_{\supp\hat{A}}\to\mathfrak{A}_\Lambda} \\
        &= \left\| (\mathcal{V}_\Lambda^\epsilon)^{-1} \circ \left( 0I_0 \oplus \bigoplus_{h=1}^m e^{t\lambda_h}e^{(\Delta^{\mathrm{ex}}-\epsilon)tN_h} \right) \circ \mathcal{V}_\Lambda^\epsilon\restriction\mathfrak{A}_{\supp\hat{A}} \right\|_{\mathfrak{A}_{\supp\hat{A}}\to\mathfrak{A}_\Lambda} \\
        &\le \|(\mathcal{V}_\Lambda^\epsilon)^{-1}\|_{\C^{d^{2|\Lambda|}}\to\mathfrak{A}_\Lambda}\|\mathcal{V}_\Lambda^\epsilon\restriction\mathfrak{A}_{\supp\hat{A}}\|_{\mathfrak{A}_{\supp\hat{A}}\to\C^{d^{2|\Lambda|}}}\max_{h=1,2,\ldots,m}|e^{t\lambda_h}|{\|e^{(\Delta^{\mathrm{ex}}-\epsilon)tN_h}\|} \\
        &\le \kappa_{\supp\hat{A}}^\epsilon\max_{h=1,2,\ldots,m}e^{t\operatorname{Re}\lambda_h}e^{(\Delta^{\mathrm{ex}}-\epsilon)t\|N_h\|}
    \end{align*}
    To calculate the term $e^{t\operatorname{Re}\lambda_h}e^{(\Delta^{\mathrm{ex}}-\epsilon)t\|N_h\|}$, we consider the following two cases.
    \begin{description}
        \item[If the eigenvalue $\lambda_h\neq 0$ is semisimple,] we have $N_h=0$ and hence obtain from $\operatorname{Re}\lambda_h\le-\Delta_\Lambda\le-\Delta$
        \begin{align}\label{eq:semisimple_case_for_main_theorem1}
            e^{t\operatorname{Re}\lambda_h}e^{(\Delta^{\mathrm{ex}}-\epsilon)t\|N_h\|}
            \le e^{-\Delta t}.
        \end{align}
        \item[If the eigenvalue $\lambda_h\neq 0$ is not semisimple,] $\|N_h\|\le 1$ and $\operatorname{Re}\lambda_h\le -\Delta_\Lambda^{\mathrm{ex}}\le-\Delta^{\mathrm{ex}}$ follow
        \begin{align}\label{eq:non-semisimple_case_for_main_theorem1}
            e^{t\operatorname{Re}\lambda_h}e^{(\Delta-\epsilon)t\|N_h\|}
            \le e^{-\Delta^{\mathrm{ex}}t}e^{(\Delta^{\mathrm{ex}}-\epsilon)t}
            \le e^{-\epsilon t}.
        \end{align}
    \end{description}
    Therefore, from Eqs.~\eqref{eq:semisimple_case_for_main_theorem1} and \eqref{eq:non-semisimple_case_for_main_theorem1}, we have
    \begin{align*}
        |\omega\circ\gamma_t^\Lambda(\hat{A})-\omega\circ\Pi_\Lambda(\hat{A})|
        \le \kappa_{\supp\hat{A}}^\epsilon \|\hat{A}\|e^{-\min\{\Delta,\epsilon\} t}.
    \end{align*}
    For every set $\Lambda_\mu$ of the subnet $(\Lambda_\mu)_\mu$ that includes $\supp\hat{A}$, we obtain
    \begin{align*}
        &|\omega\circ\gamma_t^\Gamma(\hat{A})-\bar{\omega}(\hat{A})| \notag \\
        &\le |\omega\circ\gamma_t^\Gamma(\hat{A}) - \omega\circ\gamma_t^{\Lambda_\mu}(\hat{A})| + |\omega\circ\gamma_t^{\Lambda_\mu}(\hat{A}) - \omega\circ\Pi_{\Lambda_\mu}(\hat{A})| + |\omega\circ\Pi_{\Lambda_\mu}(\hat{A})-\bar{\omega}(\hat{A})| \\
        &\le \|\gamma_t^\Gamma(\hat{A}) - \gamma_t^{\Lambda_\mu}(\hat{A})\| + \kappa_{\supp\hat{A}}^\epsilon \|\hat{A}\| e^{-\min\{\Delta,\epsilon\}t} + |\omega\circ\Pi_{\Lambda_\mu}(\hat{A})-\bar{\omega}(\hat{A})|.
    \end{align*}
    Taking the limit in $\mu$ provides
    \begin{align*}
        |\omega\circ\gamma_t^\Gamma(\hat{A})-\bar{\omega}(\hat{A})|
        \le \kappa_{\supp\hat{A}}^\epsilon \|\hat{A}\| e^{-\min\{\Delta,\epsilon\}t}
    \end{align*}
    for any $t\ge 0$, which leads to
    \begin{align*}
        \lim_{t\to\infty}\omega\circ\gamma_t^\Gamma(\hat{A})
        = \bar{\omega}(\hat{A}).
    \end{align*}
    Therefore, for each NESS $\omega_{\sss}\in\Sigma_{\NESS}(\omega)$,     
    \begin{align*}
        \omega_{\sss}(\hat{A})
        = \bar{\omega}(\hat{A})
    \end{align*}
    holds for any $\hat{A}\in\mathfrak{A}_{\loc}$.
    Since every element $\hat{A}'$ of the quantum spin system $\mathfrak{A}$ is the (uniform) limit of some sequence $(\hat{A}_n)_n$ of local observables $\hat{A}_n\in\mathfrak{A}_{\loc}$, the continuity of the states $\omega_{\sss}$ and $\bar{\omega}$ results in
    \begin{align*}
        \omega_{\sss}(\hat{A}')
        = \lim_{n\to\infty}\omega_{\sss}(\hat{A}_n)
        = \lim_{n\to\infty}\bar{\omega}(\hat{A}_n)
        = \bar{\omega}(\hat{A}')
    \end{align*}
    for any $\hat{A}'\in\mathfrak{A}$, which implies $\omega_{\sss}=\bar{\omega}$.
    In particular, the NESS for the initial state $\omega$ is unique.
    
    The discussion above can be applied to arbitrary weakly-$*$ convergent subnets of the net $(\omega\circ\Pi_\Lambda)_{\Lambda\Subset\Gamma}$ of states.
    Specifically, every subnet of $(\omega\circ\Pi_\Lambda)_{\Lambda\Subset\Gamma}$ has some weakly-$*$ convergent subnet, the limit of which is the unique NESS $\omega_{\sss}$ for $\omega$.
    In view of Theorem~\ref{thm:subnet_of_subnet} in Appendix~\ref{sec:net} (also see \cite{Kelley1955} p.74 (c)), this implies the net $(\omega\circ\Pi_\Lambda)_{\Lambda\Subset\Gamma}$ weakly-$*$ converges to $\omega_{\sss}$, i.e.,
    \begin{align*}
        \omega_{\sss}
        = \wslim{\Lambda}\omega\circ\Pi_\Lambda
        = \wslim{\Lambda}\wslim{t\to\infty}\omega\circ\gamma_t^\Lambda.
    \end{align*}
\end{proof}

Theorem~\ref{thm:main_theorem2} is proved parallelly as that of Theorem~\ref{thm:main_theorem}, as follows.
\begin{proof}[Proof of Theorem~\ref{thm:main_theorem2}]
    Let $\omega$ be an arbitrary state.
    As shown in Lemma~\ref{lem:NESS_in_finite_systems}, $\omega\circ\Pi_\Lambda$ for every finite subset $\Lambda\Subset\Gamma$ is a state over $\mathfrak{A}$ and thus there exists a subnet $(\Lambda_\mu)_\mu$ of the net $(\Lambda)_{\Lambda\Subset\Gamma}$ such that the limit
    \begin{align*}
        \bar{\omega}
        = \wslim{\mu}\omega\circ\Pi_{\Lambda_\mu}
    \end{align*}
    exists as a state over $\mathfrak{A}$. 
    
    First, we see $\bar{\omega}=\omega_{\sss}^{\mathrm{ave}}$ for any $\omega_{\sss}^{\mathrm{ave}}\in\Sigma_{\TANESS}(\omega)$.
    We take any operator $\hat{A}\in\mathfrak{A}_{\loc}$ with a finite support $\supp\hat{A}\Subset\Gamma$ and any finite subset $\Lambda$ of $\Gamma$ that includes $\supp\hat{A}$, i.e.,
    \begin{align*}
        \supp\hat{A}
        \subset \Lambda
        \Subset \Gamma.
    \end{align*}
    Then, we have
    \begin{align*}
        &\left|\frac{1}{T}\int_0^T\omega\circ\gamma_t^\Lambda(\hat{A})\dif t - \omega\circ\Pi_\Lambda(\hat{A})\right| \\
        &\le \left\|\left(\frac{1}{T}\int_0^T e^{t\mathcal{L}_\Lambda}\dif t-\pi_\Lambda\right)\restriction\mathfrak{A}_{\supp\hat{A}}\right\|_{\mathfrak{A}_{\supp\hat{A}}\to\mathfrak{A}_\Lambda}\|\hat{A}\|.
    \end{align*}
    We transform $(\Delta^{\mathrm{ex}}-\epsilon)^{-1}\mathcal{L}_\Lambda$ into the Jordan canonical form as
    \begin{align*}
        \bigoplus_{h=0}^m \left[ (\Delta^{\mathrm{ex}}-\epsilon)^{-1}\lambda_h I_h+N_h \right]
        = \mathcal{V}_\Lambda^\epsilon\circ(\Delta^{\mathrm{ex}}-\epsilon)^{-1}\mathcal{L}_\Lambda\circ(\mathcal{V}_\Lambda^\epsilon)^{-1},
    \end{align*}
    where we use the same notations as \eqref{eq:Jordan_canonical_form1} for $\lambda_h$, $I_h$, and $N_h$.
    From the assumption (B2), i.e., 
    \begin{align*}
        \|\mathcal{V}_\Lambda^\epsilon\restriction\mathfrak{A}_{\supp\hat{A}}\|_{\mathfrak{A}_{\supp\hat{A}}\to\C^{d^{2|\Lambda|}}}\|(\mathcal{V}_\Lambda^\epsilon)^{-1}\|_{\C^{d^{2|\Lambda|}}\to\mathfrak{A}_\Lambda}
        \le \kappa_{\supp\hat{A}}^\epsilon
        < \infty,
    \end{align*}
    it follows that
    \begin{align*}
        &\left\|\left(\frac{1}{T}\int_0^T e^{t\mathcal{L}_\Lambda}\dif t-\pi_\Lambda\right)\restriction\mathfrak{A}_{\supp\hat{A}}\right\|_{\mathfrak{A}_{\supp\hat{A}}\to\mathfrak{A}_\Lambda} \\
        &= \left\| (\mathcal{V}_\Lambda^\epsilon)^{-1} \circ \frac{1}{T}\int_0^T \left( 0I_0 \oplus \bigoplus_{h=1}^m e^{t\lambda_h}e^{(\Delta^{\mathrm{ex}}-\epsilon)tN_h} \right) \dif t \circ \mathcal{V}_\Lambda^\epsilon\restriction\mathfrak{A}_{\supp\hat{A}} \right\|_{\mathfrak{A}_{\supp\hat{A}}\to\mathfrak{A}_\Lambda} \\
        &\le \kappa_{\supp\hat{A}}^\epsilon\max_{h=1,2,\ldots,m} \frac{1}{T}\left\|\int_0^Te^{t\lambda_h}e^{(\Delta^{\mathrm{ex}}-\epsilon)tN_h}\dif t\right\|
    \end{align*}
    To calculate the term $T^{-1}\|\int_0^T e^{t\lambda_h}e^{(\Delta^{\mathrm{ex}}-\epsilon)tN_h}\dif t\|$, we consider the following two cases.
    \begin{description}
        \item[If the eigenvalue $\lambda_h\neq 0$ is semisimple,] we have $N_h=0$ and hence obtain from $|\lambda_h|\ge\Delta_\Lambda^{\mathrm{p}}\ge\Delta^{\mathrm{p}}$ that
        \begin{align}\label{eq:semisimple_case_for_main_theorem2}
            \left\|\int_0^T e^{t\lambda_h}e^{(\Delta^{\mathrm{ex}}-\epsilon)tN_h}\dif t\right\|
            = \left|\int_0^T e^{t\lambda_h}\dif t\right|
            = \frac{|1-e^{T\lambda_h}|}{|\lambda_h|}
            \le \frac{2}{|\lambda_h|}
            \le \frac{2}{\Delta^{\mathrm{p}}}.
        \end{align}
        \item[If the eigenvalue $\lambda_h\neq 0$ is not semisimple,] the fact that $\|N_h\|\le 1$ and $\operatorname{Re}\lambda_h\le -\Delta_\Lambda^{\mathrm{ex}}\le-\Delta^{\mathrm{ex}}$ yields
        \begin{align}\label{eq:non-semisimple_case_for_main_theorem2}
            \left\|\int_0^T e^{t\lambda_h}e^{(\Delta^{\mathrm{ex}}-\epsilon)tN_h}\dif t\right\|
            \le \int_0^T |e^{t\lambda_h}|e^{(\Delta^{\mathrm{ex}}-\epsilon)t\|N_h\|}\dif t
            \le \int_0^T e^{-\epsilon t}\dif t
            = \frac{1-e^{-\epsilon T}}{\epsilon}
            \le \epsilon^{-1}.
        \end{align}
    \end{description}
    Therefore, from Eqs.~\eqref{eq:semisimple_case_for_main_theorem2} and \eqref{eq:non-semisimple_case_for_main_theorem2}, we have
    \begin{align*}
        \left\|\left(\frac{1}{T}\int_0^T e^{t\mathcal{L}_\Lambda}\dif t-\pi_\Lambda\right)\restriction\mathfrak{A}_{\supp\hat{A}}\right\|_{\mathfrak{A}_{\supp\hat{A}}\to\mathfrak{A}_\Lambda}
        \le \frac{\kappa_{\supp\hat{A}}^\epsilon\max\{2/\Delta^{\mathrm{p}},\epsilon^{-1}\}}{T}
    \end{align*}
    and thereby
    \begin{align*}
        \left|\frac{1}{T}\int_0^T\omega\circ\gamma_t^\Lambda(\hat{A})\dif t - \omega\circ\Pi_\Lambda(\hat{A})\right|
        \le \frac{\kappa_{\supp\hat{A}}^\epsilon}\max\{2/\Delta^{\mathrm{p}},\epsilon^{-1}\}{T}\|\hat{A}\|.
    \end{align*}
    For every set $\Lambda_\mu$ of the subnet $(\Lambda_\mu)_\mu$ that includes $\supp\hat{A}$, we obtain
    \begin{align}
        &\left|\frac{1}{T}\int_0^T\omega\circ\gamma_t^\Gamma(\hat{A})\dif t-\bar{\omega}(\hat{A})\right| \notag \\
        &\le \left|\frac{1}{T}\int_0^T\left[\omega\circ\gamma_t^\Gamma(\hat{A})- \omega\circ\gamma_t^{\Lambda_\mu}(\hat{A})\right]\dif t\right| + \left|\frac{1}{T}\int_0^T\omega\circ\gamma_t^{\Lambda_\mu}(\hat{A})\dif t - \omega\circ\Pi_{\Lambda_\mu}(\hat{A})\right| + |\omega\circ\Pi_{\Lambda_\mu}(\hat{A})-\bar{\omega}(\hat{A})| \notag\\
        &\le \frac{1}{T}\int_0^T\|\gamma_t^\Gamma(\hat{A}) - \gamma_t^{\Lambda_\mu}(\hat{A})\|\dif t + \frac{\kappa_{\supp\hat{A}}^\epsilon}\max\{2/\Delta^{\mathrm{p}},\epsilon^{-1}\}{T}\|\hat{A}\| + |\omega\circ\Pi_{\Lambda_\mu}(\hat{A})-\bar{\omega}(\hat{A})|.
        \label{eq:evaluate_TANESS}
    \end{align}
    Since the convergence $\|\gamma_t^\Gamma(\hat{A})-\gamma_t^{\Lambda_\mu}\|\xrightarrow{\mu}0$ is uniform on the closed interval $[0,T]$ as mentioned in Theorem~\ref{thm:existence_of_dynamics}, for each $T>0$ we have
    \begin{align*}
        \int_0^T\|\gamma_t^\Gamma(\hat{A})-\gamma_t^{\Lambda_\mu}(\hat{A})\|\dif t
        \xrightarrow{\mu} 0.
    \end{align*}
    Thus, taking the limit of $\mu$ in Eq.~\eqref{eq:evaluate_TANESS} provides
    \begin{align*}
        \left|\frac{1}{T}\int_0^T\omega\circ\gamma_t^\Gamma(\hat{A})\dif t-\bar{\omega}(\hat{A})\right|
        \le \frac{\kappa_{\supp\hat{A}}^\epsilon}\max\{2/\Delta^{\mathrm{p}},\epsilon^{-1}\}{T}\|\hat{A}\|
    \end{align*}
    for any $T>0$, which leads to
    \begin{align*}
        \lim_{T\to\infty}\frac{1}{T}\int_0^T\omega\circ\gamma_t^\Gamma(\hat{A})\dif t
        = \bar{\omega}(\hat{A}).
    \end{align*}
    For each TANESS $\omega_{\sss}^{\mathrm{ave}}\in\Sigma_{\TANESS}(\omega)$, therefore,
    \begin{align*}
        \omega_{\sss}^{\mathrm{ave}}(\hat{A})
        = \bar{\omega}(\hat{A})
    \end{align*}
    holds for any $\hat{A}\in\mathfrak{A}_{\loc}$.
    Since every element $\hat{A}'$ of the quantum spin system $\mathfrak{A}$ is the (uniform) limit of some sequence $(\hat{A}_n)_n$ of $\hat{A}_n\in\mathfrak{A}_{\loc}$, the continuity of the states $\omega_{\sss}^{\mathrm{ave}}$ and $\bar{\omega}$ results in
    \begin{align*}
        \omega_{\sss}^{\mathrm{ave}}(\hat{A}')
        = \lim_{n\to\infty}\omega_{\sss}^{\mathrm{ave}}(\hat{A}_n)
        = \lim_{n\to\infty}\bar{\omega}(\hat{A}_n)
        = \bar{\omega}(\hat{A}')
    \end{align*}
    for any $\hat{A}'\in\mathfrak{A}$, which implies $\omega_{\sss}^{\mathrm{ave}}=\bar{\omega}$.
    In particular, the TANESS for the initial state $\omega$ is unique.
    
    The discussion above can be applied to arbitrary weakly-$*$ convergent subnets of the net $(\omega\circ\Pi_\Lambda)_{\Lambda\Subset\Gamma}$ of states.
    Specifically, every subnet of $(\omega\circ\Pi_\Lambda)_{\Lambda\Subset\Gamma}$ has some weakly-$*$ convergent subnet, the limit of which is $\omega_{\sss}^{\mathrm{ave}}$, the unique TANESS for $\omega$.
    From Theorem \ref{thm:subnet_of_subnet} (also see \cite{Kelley1955} p.74 (c)), this implies the net $(\omega\circ\Pi_\Lambda)_{\Lambda\Subset\Gamma}$ weakly-$*$ converges to $\omega_{\sss}^{\mathrm{ave}}$,
    \begin{align*}
        \omega_{\sss}^{\mathrm{ave}}
        = \wslim{\Lambda}\omega\circ\Pi_\Lambda
        = \wslim{\Lambda}\wslim{t\to\infty}\omega\circ\gamma_t^\Lambda.
    \end{align*}
\end{proof}

\subsection{Thermodynamic limit of finite system NESS and TANESS}\label{subsec:TDL_of_finite_system_NESS_and_TANESS}
In this subsection, we examine in depth the state $\bar{\omega}$ appearing in the proofs of the main theorems.

For a fixed finite subset $\Lambda\Subset\Gamma$, Lemma~\ref{lem:NESS_in_finite_systems} states 
\begin{align*}
    \omega\circ\Pi_\Lambda
    = \wslim{T\to\infty}\frac{1}{T}\int_0^T\omega\circ\gamma_t^\Lambda\dif t
\end{align*}
for any state $\omega$ over $\mathfrak{A}$.
Thus, we can call the state $\omega\circ\Pi_\Lambda$ a finite system TANESS.
Moreover, if $\Delta_\Lambda>0$ holds, $\omega\circ\Pi_\Lambda$ is also regarded as a finite system NESS because we have $\omega\circ\Pi_\Lambda=\wslim{t\to\infty}\omega\circ\gamma_t^\Lambda$.

Let $\bar{\omega}$ be a weak-$*$ cluster point of the net $(\omega\circ\Pi_\Lambda)_{\Lambda\Subset\Gamma}$.
Namely, $\bar{\omega}$ is the weak-$*$ limit of a subnet $(\omega\circ\Pi_{\Lambda_\mu})_\mu$, i.e.,
\begin{align}\label{eq:omega_bar}
    \bar{\omega}
    = \wslim{\mu}\omega\circ\Pi_{\Lambda_\mu}.
\end{align}
The state $\bar{\omega}$ is a thermodynamic limit of finite system TANESS (and sometimes NESS).
This $\bar{\omega}$ is coincident with the NESS or the TANESS for the initial state $\omega$ under some certain assumptions in Theorem~\ref{thm:main_theorem} or Theorem~\ref{thm:main_theorem2}, respectively, but $\bar{\omega}$ differs from the NESS or the TANESS for $\omega$ in generic cases.
However, if we take an initial state as $\bar{\omega}$ itself, $\bar{\omega}$ is a unique NESS and TANESS for $\bar{\omega}$.

\begin{prp}
    Let $(\mathcal{L}_\Lambda)_{\Lambda\Subset\Gamma}$ be a family of Liouvillians on finite subsets satisfying Assumption \ref{asm:locality}.
    Then, the state $\bar{\omega}$ defined in Eq.~\eqref{eq:omega_bar} is invariant under the time evolution $\gamma_t^\Gamma$, i.e.,
    \begin{align*}
        \bar{\omega}\circ\gamma_t^\Gamma
        = \bar{\omega}
    \end{align*}
    for $t\ge 0$.
    In particular, $\bar{\omega}$ is a unique NESS and TANESS for the initial state $\bar{\omega}$, i.e.,
    \begin{align}
        \Sigma_{\NESS}(\bar{\omega})
        = \Sigma_{\TANESS}(\bar{\omega})
        = \{\bar{\omega}\}.
    \end{align}
\end{prp}
\begin{proof}
    Considering $\pi_\Lambda\circ\mathcal{L}_\Lambda=\mathcal{L}_\Lambda\circ\pi_\Lambda=0$, we obtain $\Pi_\Lambda\circ\gamma_t^\Lambda=\Pi_\Lambda$ for $\Lambda\Subset\Gamma$.
    Hence, we have 
    \begin{align*}
        |\bar{\omega}\circ\gamma_t^{\Gamma}(\hat{A})-\bar{\omega}(\hat{A})| 
        &\le |\bar{\omega}\circ\gamma_t^\Gamma(\hat{A})-\omega\circ\Pi_{\Lambda_\mu}\circ\gamma_t^\Gamma(\hat{A})| \\ 
        &\qquad+ |\omega\circ\Pi_{\Lambda_\mu}\circ\gamma_t^\Gamma(\hat{A})-\omega\circ\Pi_{\Lambda_\mu}\circ\gamma_t^{\Lambda_\mu}(\hat{A})| \\
        &\qquad+ |\omega\circ\Pi_{\Lambda_\mu}\circ\gamma_t^{\Lambda_\mu}(\hat{A})-\bar{\omega}(\hat{A})| \\
        &\le |\bar{\omega}\circ\gamma_t^\Gamma(\hat{A})-\omega\circ\Pi_{\Lambda_\mu}\circ\gamma_t^\Gamma(\hat{A})| \\
        &\qquad+ \|\gamma_t^\Gamma(\hat{A})-\gamma_t^{\Lambda_\mu}(\hat{A})\| \\
        &\qquad+ |\omega\circ\Pi_{\Lambda_\mu}(\hat{A})-\bar{\omega}(\hat{A})| \\
        &\xrightarrow{\mu} 0
    \end{align*}
    for any $\hat{A}\in\mathfrak{A}$.
\end{proof}

\subsection{Application of main theorems to onsite Liouvillians}\label{subsec:application_of_main_theorems}
We define \textit{onsite Liouvillians} as a family of Liouvllians $(\mathcal{L}_\Lambda)_{\Lambda\Subset\Gamma}$ described by
\begin{align}\label{eq:onsite_Liouvillian}
    \mathcal{L}_\Lambda
    = \sum_{j\in\Lambda}\mathcal{L}_{\{j\}}
\end{align}
for any finite subset $\Lambda\Subset\Gamma$.
It is clear that the onsite Liouvillians satisfy Assumption~\ref{asm:locality}.
We call $\mathcal{L}_{\{j\}}$ a \textit{single-site Liouvllian} for each site $j\in\Gamma$.
In this subsection, we show that diagonalizable onsite Liouvillians have bounded condition number, i.e., $\kappa_{\Lambda'}$ introduced in Condition~(A2) of Corollary~\ref{cor:main_theorems} is finite for every finite subset $\Lambda'\Subset\Gamma$.
This makes it possible to apply Corollary~\ref{cor:main_theorems} to diagonalizable onsite Liouvillians, which leads to commutativity of the thermodynamic and long-time limits in calculating the NESS if a spectral gap condition is satisfied.

The exact statement is as follow.

\begin{thm}\label{thm:onsite_Liouvillian}
    Let $(\mathcal{L}_\Lambda)_{\Lambda\Subset\Gamma}$ be onsite Liouvillians such that each single-site Liouvillian $\mathcal{L}_{\{j\}}$ is diagonalizable. 
    Then, there exists a family $(\mathcal{V}_\Lambda)_{\Lambda\Subset\Gamma}$ of invertible maps $\mathcal{V}_\Lambda\in B(\mathfrak{A}_\Lambda,\mathbb{C}^{d^{2|\Lambda|}})$ that diagonalize $\mathcal{L}_\Lambda$, such that for any finite subset $\Lambda'\Subset\Gamma$ we have
    \begin{align}
        \kappa_{\Lambda'}
        \coloneqq \|\mathcal{V}_\Lambda\restriction\mathfrak{A}_{\Lambda'}\|_{\mathfrak{A}_{\Lambda'}\to\C^{d^{2|\Lambda|}}}\|(\mathcal{V}_\Lambda)^{-1}\|_{\C^{d^{2|\Lambda|}}\to\mathfrak{A}_\Lambda}
        < \infty.
    \end{align}
\end{thm}

To prove this, we prepare the following two lemmas about the completely bounded (cb) norm $\|\bullet\|_{\text{cb}}$.

\begin{lem}\label{lem:expression_of_cb_norm}
    By the identification $\mathbb{C}^{n^2} = B(\mathbb{C},\mathbb{C}^{n^2})$, we think of $\mathbb{C}^{n^2}$ as a subspace of the $C^*$-algebra $B(\mathbb{C}^{n^2})=\mathrm{M}_{n^2}$.
    Thereby we define the norm in $\mathbb{C}^{n^2}\otimes\mathrm{M}_k=B(\mathbb{C}^k,\mathbb{C}^{n^2}\otimes\mathbb{C}^k)$ by
    \begin{align*}
        \|\hat{X}\|_{\mathbb{C}^k\to\mathbb{C}^{n^2}\otimes\mathbb{C}^k}
        \coloneqq \sup\{\|\hat{X}\boldsymbol{v}\|_2 \mid \boldsymbol{v}\in\mathbb{C}^k,\ \|\boldsymbol{v}\|_2\le 1\}.
    \end{align*}
    For a linear map $\mathcal{S}\colon\mathbb{C}^{n^2}\to\mathrm{M}_n$, we also define the cb norm by
    \begin{align*}
        \|\mathcal{S}\|_{\text{cb}}
        \coloneqq \sup_{k\ge 1} \sup\{\|\mathcal{S}\otimes\id_{\mathrm{M}_k}(\hat{X})\| \mid \hat{X}\in B(\mathbb{C}^k,\mathbb{C}^{n^2}\otimes\mathbb{C}^k),\ \|\hat{X}\|_{\mathbb{C}^k\to\mathbb{C}^{n^2}\otimes\mathbb{C}^k}\le 1\}.
    \end{align*}
    Let $(\boldsymbol{e}_\mu)_{\mu=0}^{n^2-1}$ be the standard basis of $\mathbb{C}^{n^2}$.
    Then, we have
    \begin{align}\label{eq:explicit_formula_for_cb_norm}
        \|\mathcal{S}\|_{\text{cb}}
        = \left\| \sum_{\mu=0}^{n^2-1} \mathcal{S}(\boldsymbol{e}_\mu)\mathcal{S}(\boldsymbol{e}_\mu)^* \right\|^{1/2}.
    \end{align}
\end{lem}
\begin{proof}
    For a given $\hat{X}\in\mathbb{C}^{n^2}\otimes\mathrm{M}_k$, we can linearly expand it by
    \begin{align*}
        \hat{X}
        = \sum_{\mu=0}^{n^2-1}\boldsymbol{e}_\mu\otimes\hat{X}_\mu
    \end{align*}
    for some $\hat{X}_\mu\in\mathrm{M}_k$, $\mu=0,1,\ldots,n^2-1$.
    For a given $\boldsymbol{v}\in\mathbb{C}^k$, we have
    \begin{align*}
        \|\hat{X}\boldsymbol{v}\|_2^2
        = \left\|\sum_{\mu=0}^{n^2-1}\boldsymbol{e}_\mu\otimes \hat{X}_\mu\boldsymbol{v}\right\|_2^2
        = \sum_{\mu=0}^{n^2-1}\braket{\hat{X}_\mu\boldsymbol{v},\hat{X}_\mu\boldsymbol{v}}
        = \Braket{\boldsymbol{v},\sum_{\mu=0}^{n^2-1}\hat{X}_\mu^* \hat{X}_\mu\boldsymbol{v}}.
    \end{align*}
    Hence, we obtain
    \begin{align*}
        \|\hat{X}\|_{\mathbb{C}^k\to\mathbb{C}^{n^2}\otimes\mathbb{C}^k}
        = \left\|\sum_{\mu=0}^{n^2-1}\hat{X}_\mu^* \hat{X}_\mu\right\|^{1/2}.
    \end{align*}

    On the other hand, recalling
    \begin{align*}
        \mathcal{S}\otimes\id_{\mathrm{M}_k}(\hat{X})
        &= \sum_{\mu=0}^{n^2-1}\mathcal{S}(\boldsymbol{e}_\mu)\otimes \hat{X}_\mu
        = \sum_{\mu=0}^{n^2-1}(\mathcal{S}(\boldsymbol{e}_\mu)\otimes\hat{I})(\hat{I}\otimes\hat{X}_\mu) \\
        &= \begin{pmatrix}
            \mathcal{S}(\boldsymbol{e}_0)\otimes\hat{I} & \mathcal{S}(\boldsymbol{e}_1)\otimes\hat{I} & \cdots & \mathcal{S}(\boldsymbol{e}_{n^2-1})\otimes\hat{I}
        \end{pmatrix}\begin{pmatrix}
            \hat{I}\otimes\hat{X}_0 \\
            \hat{I}\otimes\hat{X}_1 \\
            \vdots \\
            \hat{I}\otimes\hat{X}_{n^2-1}
        \end{pmatrix},
    \end{align*}
    we evaluate its norm as
    \begin{equation*}
        \|\mathcal{S}\otimes\id_{\mathrm{M}_k}(\hat{X})\|
        \leq \left\|\begin{pmatrix}
            \mathcal{S}(\boldsymbol{e}_0)\otimes\hat{I} & \mathcal{S}(\boldsymbol{e}_1)\otimes\hat{I} & \cdots & \mathcal{S}(\boldsymbol{e}_{n^2-1})\otimes\hat{I}
        \end{pmatrix}\right\| \left\|\begin{pmatrix}
            \hat{I}\otimes\hat{X}_0 \\
            \hat{I}\otimes\hat{X}_1 \\
            \vdots \\
            \hat{I}\otimes\hat{X}_{n^2-1}
        \end{pmatrix} \right\|.
    \end{equation*}
    
    We compute these two norms in the right hand side as
    \begin{align*}
        &\left\| \begin{pmatrix}
            \mathcal{S}(\boldsymbol{e}_0)\otimes\hat{I} & \mathcal{S}(\boldsymbol{e}_1)\otimes\hat{I} & \cdots & \mathcal{S}(\boldsymbol{e}_{n^2-1})\otimes\hat{I}
        \end{pmatrix} \right\| \notag \\
        &= \left\|\begin{pmatrix}
            \mathcal{S}(\boldsymbol{e}_0)^*\otimes\hat{I} \\
            \mathcal{S}(\boldsymbol{e}_1)^*\otimes\hat{I} \\
            \vdots \\
            \mathcal{S}(\boldsymbol{e}_{n^2-1})^*\otimes\hat{I}
        \end{pmatrix}\right\| \\
        &= \left\|\begin{pmatrix}
            \mathcal{S}(\boldsymbol{e}_0)\otimes\hat{I} & \mathcal{S}(\boldsymbol{e}_1)\otimes\hat{I} & \cdots & \mathcal{S}(\boldsymbol{e}_{n^2-1})\otimes\hat{I}
        \end{pmatrix}\begin{pmatrix}
            \mathcal{S}(\boldsymbol{e}_0)^*\otimes\hat{I} \\
            \mathcal{S}(\boldsymbol{e}_1)^*\otimes\hat{I} \\
            \vdots \\
            \mathcal{S}(\boldsymbol{e}_{n^2-1})^*\otimes\hat{I}
        \end{pmatrix}\right\|^{1/2} \\
        &= \left\| \left( \sum_{\mu=0}^{n^2-1}\mathcal{S}(\boldsymbol{e}_\mu)\mathcal{S}(\boldsymbol{e}_\mu)^* \right) \otimes \hat{I} \right\|^{1/2}
        = \left\| \sum_{\mu=0}^{n^2-1}\mathcal{S}(\boldsymbol{e}_\mu)\mathcal{S}(\boldsymbol{e}_\mu)^* \right\|^{1/2}
    \end{align*}
    and
    \begin{align*}
        \left\| \begin{pmatrix}
            \hat{I}\otimes\hat{X}_0 \\
            \hat{I}\otimes\hat{X}_1 \\
            \vdots \\
            \hat{I}\otimes\hat{X}_{n^2-1}
        \end{pmatrix} \right\|
        &= \left\| \begin{pmatrix}
            \hat{I}\otimes\hat{X}_0^* & \hat{I}\otimes\hat{X}_1^* & \cdots & \hat{I}\otimes\hat{X}_{n^2-1}^*
        \end{pmatrix} \begin{pmatrix}
            \hat{I}\otimes\hat{X}_0 \\
            \hat{I}\otimes\hat{X}_1 \\
            \vdots \\
            \hat{I}\otimes\hat{X}_{n^2-1}
        \end{pmatrix} \right\|^{1/2} \\
        &= \left\| \hat{I}\otimes\sum_{\mu=0}^{n^2-1}\hat{X}_\mu^*\hat{X}_\mu \right\|^{1/2}
        = \left\| \sum_{\mu=0}^{n^2-1}\hat{X}_\mu^*\hat{X}_\mu \right\|^{1/2}
        = \|\hat{X}\|_{\mathbb{C}^k\to\mathbb{C}^{n^2}\otimes\mathbb{C}^k}.
    \end{align*}
    Therefore, we obtain
    \begin{align*}
        \|\mathcal{S}\otimes\id_{\mathrm{M}_k}(\hat{X})\|
        \le \left\| \sum_{\mu=0}^{n^2-1}\mathcal{S}(\boldsymbol{e}_\mu)\mathcal{S}(\boldsymbol{e}_\mu)^* \right\|^{1/2} \|\hat{X}\|_{\mathbb{C}^k\to\mathbb{C}^{n^2}\otimes\mathbb{C}^k}.
    \end{align*}
    This implies
    \begin{align*}
        \|\mathcal{S}\otimes\id_{\mathrm{M}_k}\|
        \le \left\| \sum_{\mu=0}^{n^2-1}\mathcal{S}(\boldsymbol{e}_\mu)\mathcal{S}(\boldsymbol{e}_\mu)^* \right\|^{1/2}.
    \end{align*}
    Since the right hand side of the equation above is independent of $k\ge 1$, the upper bound of the cb norm of $\mathcal{S}$ is
    \begin{align*}
        \|\mathcal{S}\|_{\text{cb}}
        \le \left\| \sum_{\mu=0}^{n^2-1}\mathcal{S}(\boldsymbol{e}_\mu)\mathcal{S}(\boldsymbol{e}_\mu)^* \right\|^{1/2}.
    \end{align*} 

    Next, we compute the lower bound.
    Let $k=n^2$. 
    For a given $\hat{Y}\in \mathbb{C}^{n^2}\otimes\mathrm{M}_{n^2}$ we write
    \begin{align*}
        \hat{Y}
        \coloneqq \sum_{\mu=0}^{n^2-1} \boldsymbol{e}_\mu\otimes\hat{E}_{1,\mu+1},
    \end{align*}
    where $\hat{E}_{i,j}\in\mathrm{M}_{n^2}$ denotes the matrix whose $(i,j)$-entry is $1$ and all the other entries are $0$.
    We note that 
    \begin{align*}
        \|\hat{Y}\|_{\mathbb{C}^{n^2}\to\mathbb{C}^{n^2}\otimes\mathbb{C}^{n^2}}
        = \left\| \sum_{\mu=0}^{n^2-1} \hat{E}_{1,\mu+1}^*\hat{E}_{1,\mu+1} \right\|^{1/2}
        = \left\| \sum_{\mu=1}^{n^2} \hat{E}_{\mu,\mu} \right\|^{1/2}
        = \|\hat{I}\|
        = 1
    \end{align*}
    and
    \begin{align*}
        \|\mathcal{S}\otimes\id_{\mathrm{M}_{n^2}}(\hat{Y})\|
        &= \left\| \sum_{\mu=0}^{n^2-1}\mathcal{S}(\boldsymbol{e}_\mu)\otimes\hat{E}_{1,\mu+1} \right\|
        = \left\| \sum_{\mu=0}^{n^2-1}\mathcal{S}(\boldsymbol{e}_\mu)^*\otimes\hat{E}_{\mu+1,1} \right\| \\
        &= \left\| \sum_{\mu=0}^{n^2-1}\mathcal{S}(\boldsymbol{e}_\mu)\mathcal{S}(\boldsymbol{e}_\mu)^*\otimes\hat{E}_{1,1} \right\|^{1/2}
        = \left\| \sum_{\mu=0}^{n^2-1}\mathcal{S}(\boldsymbol{e}_\mu)\mathcal{S}(\boldsymbol{e}_\mu)^* \right\|^{1/2}.
    \end{align*}
    Hence, we obtain 
    \begin{align*}
        \|\mathcal{S}\|_{\text{cb}}
        \ge \|\mathcal{S}\otimes\id_{\mathrm{M}_{n^2}}\|
        \ge \frac{\|\mathcal{S}\otimes\id_{\mathrm{M}_{n^2}}(\hat{Y})\|}{\|\hat{Y}\|_{\mathbb{C}^{n^2}\to\mathbb{C}^{n^2}\otimes\mathbb{C}^{n^2}}}
        = \left\| \sum_{\mu=0}^{n^2-1}\mathcal{S}(\boldsymbol{e}_\mu)\mathcal{S}(\boldsymbol{e}_\mu)^* \right\|^{1/2}.
    \end{align*}
    Therefore, Equation~\eqref{eq:explicit_formula_for_cb_norm} follows. 
\end{proof}

\begin{lem}\label{lem:cb_norm_and_tensor}
    For a linear map $\mathcal{S}_j\colon\mathbb{C}^{n_j^2}\to\mathrm{M}_{n_j}$, $j=1,2$, 
    \begin{align*}
        \|\mathcal{S}_1\otimes\mathcal{S}_2\|_{\text{cb}}
        = \|\mathcal{S}_1\|_{\text{cb}}\|\mathcal{S}_2\|_{\text{cb}}.
    \end{align*}
\end{lem}
\begin{proof}
    This is just a special case of \cite[Theorem~12.3]{Paulsen2003} but we can show this directly by using Lemma~\ref{lem:expression_of_cb_norm}.
    Let $(\boldsymbol{e}_\mu)_{\mu=0}^{n_1^2-1}$ and $(\boldsymbol{f}_\nu)_{\nu=0}^{n_1^2-1}$ be the standard basis of $\mathbb{C}^{n_1^2}$ and $\mathbb{C}^{n_2^2}$ respectively. 
    We note that $(\boldsymbol{e}_\mu\otimes\boldsymbol{f}_\nu)_{\mu,\nu}$ is the basis of $\mathbb{C}^{(n_1n_2)^2}$. Hence, we have
    \begin{align*}
        \|\mathcal{S}_1\|_{\text{cb}}\|\mathcal{S}_2\|_{\text{cb}}
        &= \left\|\sum_{\mu=0}^{n_1^2-1}\mathcal{S}_1(\boldsymbol{e}_\mu)\mathcal{S}_1(\boldsymbol{e}_\mu)^*\right\|^{1/2} \left\|\sum_{\nu=0}^{n_2^2-1}\mathcal{S}_2(\boldsymbol{f}_\nu)\mathcal{S}_2(\boldsymbol{f}_\nu)^*\right\|^{1/2} \\
        &= \left\|\sum_{\mu=0}^{n_1^2-1}\sum_{\nu=0}^{n_2^2-1}\mathcal{S}_1(\boldsymbol{e}_\mu)\mathcal{S}_1(\boldsymbol{e}_\mu)^*\otimes \mathcal{S}_2(\boldsymbol{f}_\nu)\mathcal{S}_2(\boldsymbol{f}_\nu)^*\right\|^{1/2} \\
        &= \left\|\sum_{\mu=0}^{n_1^2-1}\sum_{\nu=0}^{n_2^2-1}\mathcal{S}_1\otimes\mathcal{S}_2(\boldsymbol{e}_\mu\otimes\boldsymbol{f}_\nu)\mathcal{S}_1\otimes\mathcal{S}_2(\boldsymbol{e}_\mu\otimes\boldsymbol{f}_\nu)^*\right\|^{1/2} \\
        &= \|\mathcal{S}_1\otimes\mathcal{S}_2\|_{\text{cb}}.
    \end{align*}
\end{proof}

\begin{proof}[Proof of Theorem~\ref{thm:onsite_Liouvillian}]
    By the assumption that each single-site Liouvillian of given onsite Liouvillians is diagonalizable, for each finite subset $\Lambda\Subset\Gamma$, $\mathcal{L}_\Lambda$ can be diagonalized by some invertible linear operator $\mathcal{V}_\Lambda\colon\mathfrak{A}_\Lambda\to\mathbb{C}^{d^{2|\Lambda|}}$.
    Moreover,
    $\mathcal{V}_\Lambda$ can be written as a tensor product of $\mathcal{V}_{\{j\}}$, i.e.,
    \begin{align*}
        \mathcal{V}_\Lambda
        = \bigotimes_{j\in\Lambda}\mathcal{V}_{\{j\}}.
    \end{align*}
    Let $(\boldsymbol{e}_\mu)_{\mu=0}^{d^2-1}$ be the standard basis of $\mathbb{C}^{d^2}$. Since $\mathcal{L}_{\{j\}}(\hat{I})=0$, we can choose each $\mathcal{V}_{\{j\}}$ that satisfies
    \begin{align}\label{eq:normalization_of_V_j}
        \mathcal{V}_{\{j\}}(\hat{I})
        = \boldsymbol{e}_0.
    \end{align}
    This implies 
    \begin{align*}
        \mathcal{V}_\Lambda(\hat{A})
        = \left(\bigotimes_{i\in\Lambda'} \mathcal{V}_{\{i\}}\right)(\hat{A})\otimes\boldsymbol{e}_0^{\otimes(\Lambda\setminus\Lambda')}
    \end{align*}
    for $\Lambda'\subset\Lambda\Subset\Gamma$ and $\hat{A}\in\mathfrak{A}_{\Lambda'}$.
    Hence, we obtain
    \begin{align*}
        \|\mathcal{V}_\Lambda \restriction\mathfrak{A}_{\Lambda'}\|_{\mathfrak{A}_{\Lambda'}\to\mathbb{C}^{d^{2|\Lambda|}}}
        = \left\| \bigotimes_{i\in\Lambda'} \mathcal{V}_{\{i\}} \right\|_{\mathfrak{A}_{\Lambda'}\to\mathbb{C}^{d^{2|\Lambda'|}}}
        \eqqcolon\mathcal{N}_{\Lambda'}.
    \end{align*} 
    We remark that $\mathcal{N}_{\Lambda'}$ is finite and independent of $\Lambda$.

    On the other hand, it follows from Lemma~\ref{lem:cb_norm_and_tensor} that
    \begin{align*}
        \|(\mathcal{V}_\Lambda)^{-1}\|_{\mathbb{C}^{d^{2|\Lambda|}}\to\mathfrak{A}_\Lambda}
        = \left\|\bigotimes_{j\in\Lambda}(\mathcal{V}_{\{j\}})^{-1}\right\|_{\mathbb{C}^{d^{2|\Lambda|}}\to\mathfrak{A}_\Lambda}
        \le \left\|\bigotimes_{j\in\Lambda}(\mathcal{V}_{\{j\}})^{-1}\right\|_{\text{cb}}
        = \prod_{j\in\Lambda}\|(\mathcal{V}_{\{j\}})^{-1}\|_{\text{cb}}.
    \end{align*}
    Therefore, we can evaluate the condition number of $\mathcal{V}_\Lambda$ as 
    \begin{align*}
        \kappa_{\Lambda'}(\mathcal{V}_\Lambda)
        = \|\mathcal{V}_\Lambda \restriction\mathfrak{A}_{\Lambda'}\|_{\mathfrak{A}_{\Lambda'}\to\mathbb{C}^{d^{2|\Lambda|}}}\|(\mathcal{V}_\Lambda)^{-1}\|_{\mathbb{C}^{d^{2|\Lambda|}}\to\mathfrak{A}_\Lambda}
        \le \mathcal{N}_{\Lambda'}\prod_{j\in\Lambda}\|(\mathcal{V}_{\{j\}})^{-1}\|_{\text{cb}}.
    \end{align*}
    Furthermore, Lemma~\ref{lem:expression_of_cb_norm} yields
    \begin{align*}
        \|(\mathcal{V}_{\{j\}})^{-1}\|_{\text{cb}}
        &= \left\|\sum_{\mu=0}^{d^2-1}(\mathcal{V}_{\{j\}})^{-1}(\boldsymbol{e}_\mu)(\mathcal{V}_{\{j\}})^{-1}(\boldsymbol{e}_\mu)^*\right\|^{1/2} \\
        &= \left\|\hat{I}+\sum_{\mu=1}^{d^2-1}(\mathcal{V}_{\{j\}})^{-1}(\boldsymbol{e}_\mu)(\mathcal{V}_{\{j\}})^{-1}(\boldsymbol{e}_\mu)^*\right\|^{1/2} \\
        &\le \sqrt{1+\left\|\sum_{\mu=1}^{d^2-1}(\mathcal{V}_{\{j\}})^{-1}(\boldsymbol{e}_\mu)(\mathcal{V}_{\{j\}})^{-1}(\boldsymbol{e}_\mu)^*\right\|}.
    \end{align*}

    Next, we recall that under the transformation
    \begin{align}\label{eq:DoF_of_gauge_transformation}
        \mathcal{V}_{\{j\}}
        \mapsto \operatorname{diag}(1,c_j,c_j,\ldots,c_j)\mathcal{V}_{\{j\}},
        \quad
        c_j\in\mathbb{C}\setminus\{0\}
    \end{align}
    the map $\mathcal{V}_{\{j\}}$ retains the property that it diagonalizes $\mathcal{L}_{\{j\}}$.
    Since the eigenvalue $0$ of $\mathcal{L}_{\{j\}}$ is semisimple as discussed in Lemma~\ref{lem:spectrum_in_finite_systems}, the condition~\eqref{eq:normalization_of_V_j} is invariant under this transformation.
    Now, given a family of positive numbers $(\delta_j)_{j\in\Gamma}\in(0,\infty)^\Gamma$, we can choose $c_j$ in \eqref{eq:DoF_of_gauge_transformation} to normalize each $\mathcal{V}_{\{j\}}$ as 
    \begin{align*}
        \left\|\sum_{\mu=1}^{d^2-1}(\mathcal{V}_{\{j\}})^{-1}(\boldsymbol{e}_\mu)(\mathcal{V}_{\{j\}})^{-1}(\boldsymbol{e}_\mu)^*\right\|
        = \delta_j.
    \end{align*}
    Since we assume that $\Gamma$ is countable, there exists a choice of $(\delta_j)_{j\in\Gamma}\in(0,\infty)^\Gamma$ such that $\prod_{j\in\Gamma}(1+\delta_j) < \infty$ holds.
    Hence, we obtain
    \begin{align*}
        \kappa_{\Lambda'}
        \le \mathcal{N}_{\Lambda'}\sqrt{\prod_{j\in\Gamma}(1+\delta_j)}
        < \infty
    \end{align*}
    for any $\Lambda'\Subset\Gamma$.
\end{proof}

\begin{rmk}
    In the proof, the assumption of diagonalizability is critical to the tensor product decomposition $\mathcal{V}_\Lambda=\bigotimes_{j\in\Lambda}\mathcal{V}_{\{j\}}$.
    When a single-site Liouvillian is not diagonalizable, $\mathcal{V}_\Lambda^\epsilon$ may not agree with the tensor product of $\mathcal{V}_{\{j\}}^\epsilon$.
    The Jordan canonical form $\mathcal{J}_\Lambda^\epsilon$ is similar to $\sum_{j\in\Lambda}\hat{I}\otimes\cdots\otimes\mathcal{J}_{\{j\}}^\epsilon\otimes\cdots\otimes\hat{I}$ as a matrix, but the transformation between them is not necessarily implemented by a unitary.
\end{rmk}

This theorem ensures that there are nontrivial dissipative models satisfying the assumptions of Theorems~\ref{thm:main_theorem} and \ref{thm:main_theorem2}.
Therefore, we can apply the main theorems to the diagonalized onsite Liouvillians, leading to the corollary below.

\begin{cor}\label{cor:onsite_Liouvillian}
    Let $(\mathcal{L}_\Lambda)_{\Lambda\Subset\Gamma}$ be onsite Liouvillians such that each single-site Liouvillian $\mathcal{L}_{\{j\}}$ is diagonalizable. 
    Then, the followings hold for any state $\omega$ over the quantum spin system $\mathfrak{A}$.
    \begin{enumerate}[(A)]
        \item If $\inf_{j\in\Gamma}\Delta_{\{j\}}>0$, the limit
        \begin{align*}
            \wslim{\Lambda}\wslim{t\to\infty} \omega\circ\gamma_t^\Lambda
        \end{align*}
        exists and is the unique NESS for the initial state $\omega$.
        \item If $\inf_{j\in\Gamma}\Delta_{\{j\}}^{\mathrm{p}}>0$, the limit
        \begin{align*}
            \wslim{\Lambda}\wslim{T\to\infty} \frac{1}{T}\int_0^T\omega\circ\gamma_t^\Lambda\mathrm{d}t
        \end{align*}
        exists and is the unique TANESS for the initial state $\omega$.
    \end{enumerate}
\end{cor}
\begin{proof}
    By Theorem~\ref{thm:onsite_Liouvillian}, Assumption~(A2) in Theorem~\ref{thm:main_theorem} or Corollary~\ref{cor:main_theorems} and Assumption~(B2) in Theorem~\ref{thm:main_theorem2} are satisfied.

    On the other hand, under the assumption of \eqref{eq:onsite_Liouvillian}, the spectrum of $\mathcal{L}_\Lambda$ is
    \begin{align*}
        \sigma_{B(\mathfrak{A}_\Lambda)}(\mathcal{L}_\Lambda)
        = \left\{ \sum_{j\in\Lambda}\lambda_j \ \middle| \ \lambda_j\in\sigma_{B(\mathfrak{A}_{\{j\}})}(\mathcal{L}_{\{j\}}) \right\}.
    \end{align*}
    Hence, for any finite subset $\Lambda\Subset\Gamma$, it follows that
    \begin{align*}
        \Delta_\Lambda
        \ge \inf_{j\in\Lambda}\Delta_{\{j\}}
        \ge \inf_{j\in\Gamma}\Delta_{\{j\}},
        \quad
        \Delta_\Lambda^{\text{p}}
        \ge \inf_{j\in\Lambda}\Delta_{\{j\}}^{\text{p}}
        \ge \inf_{j\in\Gamma}\Delta_{\{j\}}^{\text{p}}.
    \end{align*}
    Therefore, Assumption~(A1) in Theorem~\ref{thm:main_theorem} or Corollary~\ref{cor:main_theorems} holds if $\inf_{j\in\Gamma}\Delta_{\{j\}}>0$, and Assumption~(B1) in Theorem~\ref{thm:main_theorem2} does if $\inf_{j\in\Gamma}\Delta_{\{j\}}^{\text{p}}>0$.
\end{proof}

We should however remark that commutativity of the thermodynamic limit and the long-time limit (or the long-time average) can be proven without explicit use of Theorem~\ref{thm:main_theorem} (or Theorem~\ref{thm:main_theorem2}) even for nondiagonalizable onsite Liouvillians.
To be exact, the following proposition holds.
\begin{prp}\label{prp:onsite_Liouvillian}
    Let $(\mathcal{L}_\Lambda)_{\Lambda\Subset\Gamma}$ be onsite Liouvillians.
    Then, for any state $\omega$ over the quantum spin system $\mathfrak{A}$, the limit
    \begin{align*}
        \wslim{\Lambda}\wslim{T\to\infty} \frac{1}{T}\int_0^T\omega\circ\gamma_t^\Lambda\mathrm{d}t
    \end{align*}
    exists and is the unique NESS for the initial state $\omega$.
    
    Moreover, if $\Delta_{\{j\}}>0$ holds for each site $j\in\Gamma$, then, for any state $\omega$ over the quantum spin system $\mathfrak{A}$, the limit
    \begin{align*}
        \wslim{\Lambda}\wslim{t\to\infty} \omega\circ\gamma_t^\Lambda
    \end{align*}
    exists and is the unique NESS for the initial state $\omega$.
\end{prp}
\begin{proof}
    The proof strategy is the same as in Theorem~\ref{thm:main_theorem2}.
    As shown in Lemma~\ref{lem:NESS_in_finite_systems}, $\omega\circ\Pi_\Lambda$ for every finite subset $\Lambda\Subset\Gamma$ is a state over $\mathfrak{A}$ and thus there exists a subnet $(\Lambda_\mu)_\mu$ of the net $(\Lambda)_{\Lambda\Subset\Gamma}$ such that the limit
    \begin{align*}
        \bar{\omega}
        = \wslim{\mu}\omega\circ\Pi_{\Lambda_\mu}
    \end{align*}
    exists as a state over $\mathfrak{A}$. 

    First, we see $\bar{\omega}=\omega_{\sss}^{\text{ave}}$ for any $\omega_{\sss}^{\text{ave}}\in\Sigma_{\TANESS}(\omega)$.
    We take any operator $\hat{A}\in\mathfrak{A}_{\loc}$ with a finite support $\supp\hat{A}\Subset\Gamma$ and any finite subset $\Lambda$ of $\Gamma$ that includes $\supp\hat{A}$, i.e.,
    \begin{align*}
        \supp\hat{A}
        \subset \Lambda
        \Subset \Gamma.
    \end{align*}
    Then, from Lemma~\ref{lem:NESS_in_finite_systems}, we have
    \begin{align*}
        \omega\circ\Pi_\Lambda(\hat{A})
        = \lim_{T\to\infty} \frac{1}{T}\int_0^T\omega\circ\gamma_t^{\Lambda}(\hat{A})\mathrm{d}t.
    \end{align*}
    Moreover, since $(\mathcal{L}_\Lambda)_{\Lambda\Subset\Gamma}$ is onsite Liouvillians, 
    \begin{align*}
        \gamma_t^\Gamma(\hat{A})
        = \gamma_t^{\Lambda}(\hat{A})
    \end{align*}
    holds.
    Hence, for every set $\Lambda_\mu$ of the subnet $(\Lambda_\mu)_\mu$ that includes $\supp\hat{A}$, we obtain
    \begin{align*}
        \omega\circ\Pi_{\Lambda_\mu}(\hat{A})
        = \lim_{T\to\infty}\frac{1}{T}\int_0^T\omega\circ\gamma_t^\Gamma(\hat{A})\mathrm{d}t.
    \end{align*}
    Taking the limit in $\mu$ provides
    \begin{align*}
        \bar{\omega}(\hat{A})
        = \lim_{T\to\infty}\frac{1}{T}\int_0^T\omega\circ\gamma_t^\Gamma(\hat{A})\mathrm{d}t.
    \end{align*}
    Therefore, for each TANESS $\omega_{\sss}^{\text{ave}}\in\Sigma_{\TANESS}(\omega)$,
    \begin{align*}
        \bar{\omega}(\hat{A})
        = \omega_{\sss}^{\text{ave}}(\hat{A})
    \end{align*}
    holds for any $\hat{A}\in\mathfrak{A}_{\loc}$.

    Since every element $\hat{A}'$ of the quantum spin system $\mathfrak{A}$ is the (uniform) limit of some sequence $(\hat{A}_n)_n$ of local observables $\hat{A}_n\in\mathfrak{A}_{\loc}$, the continuity of the states $\omega_{\sss}^{\text{ave}}$ and $\bar{\omega}$ results in
    \begin{align*}
        \omega_{\sss}^{\text{ave}}(\hat{A}')
        = \lim_{n\to\infty}\omega_{\sss}^{\text{ave}}(\hat{A}_n)
        = \lim_{n\to\infty}\bar{\omega}(\hat{A}_n)
        = \bar{\omega}(\hat{A}')
    \end{align*}
    for any $\hat{A}'\in\mathfrak{A}$, which implies $\omega_{\sss}^{\text{ave}}=\bar{\omega}$.
    In particular, the TANESS for the initial state $\omega$ is unique.
    
    The discussion above can be applied to arbitrary weakly-$*$ convergent subnets of the net $(\omega\circ\Pi_\Lambda)_{\Lambda\Subset\Gamma}$ of states.
    Specifically, every subnet of $(\omega\circ\Pi_\Lambda)_{\Lambda\Subset\Gamma}$ has some weakly-$*$ convergent subnet, the limit of which is the unique TANESS $\omega_{\sss}^{\text{ave}}$ for $\omega$.
    In view of Theorem~\ref{thm:subnet_of_subnet} in Appendix~\ref{sec:net} (also see \cite{Kelley1955} p.74 (c)), this implies the net $(\omega\circ\Pi_\Lambda)_{\Lambda\Subset\Gamma}$ weakly-$*$ converges to $\omega_{\sss}^{\text{ave}}$, i.e.,
    \begin{align*}
        \omega_{\sss}^{\text{ave}}
        = \wslim{\Lambda}\omega\circ\Pi_\Lambda
        = \wslim{\Lambda}\wslim{t\to\infty}\omega\circ\gamma_t^\Lambda.
    \end{align*}

    The statement about the NESS can be shown in the same way.
    From the assumption $\Delta_{\{j\}}>0$, the spectral gap $\Delta_\Lambda$ of the onsite Liouvillians $(\mathcal{L}_\Lambda)_{\Lambda\Subset\Gamma}$ is positive for each finite subset $\Lambda\Subset\Gamma$, so we can apply Lemma~\ref{lem:NESS_in_finite_systems} to have
    \begin{align*}
        \omega\circ\Pi_\Lambda
        = \wslim{t\to\infty}\omega\circ\gamma_t^\Lambda.
    \end{align*}
    Thus, the above argument also holds true for the NESS.
\end{proof}

\begin{rmk}
    In \cite{Kastoryano2012}, Kastoryano, Reeb, and Wolf derived lower and upper bounds of pre-cutoff times for commuting Liouvillians, which includes onsite Liouvillians in our definition.
    Concretely, \cite[Theorem~7]{Kastoryano2012} states that
    \begin{gather*}
        c<1 \implies \lim_{|\Lambda|\to\infty}\eta_{\tr}[e^{ct_{1,n}\mathcal{L}_\Lambda}] = 1, \\
        c>1 \implies \lim_{|\Lambda|\to\infty}\eta_{\tr}[e^{ct_{2,n}\mathcal{L}_\Lambda}] = 0
    \end{gather*}
    for onsite Liouvillians $(\mathcal{L}_\Lambda)_{\Lambda\Subset\Gamma}$ with a uniform single-site Liouvillian $\mathcal{L}_{\{j\}}=\mathcal{L}\in B(\mathrm{M}_d)$, $j\in\Gamma$, where 
    \begin{align*}
        t_{1,\Lambda}
        \coloneqq \frac{\log|\Lambda|}{2\Delta_{\{j\}}},
        \quad
        t_{2,\Lambda}
        \coloneqq \frac{\log|\Lambda|}{\Delta_{\{j\}}}
    \end{align*}
    are pre-cutoff times for a finite subset $\Lambda\Subset\Gamma$.
    We note that in the present case $\Delta_{\{j\}}$ is independent of site $j$.
    Also, $\eta_{\tr}$ is a trace-norm contraction defined in \cite[Definition~2]{Kastoryano2012}; 
    in our notation, it is written as
    \begin{align*}
        \eta_{\tr}[e^{t\mathcal{L}_\Lambda}]
        \coloneqq \frac{1}{2}\sup\left\{ \|\omega_\Lambda\circ\pi_\Lambda - \omega_\Lambda\circ e^{t\mathcal{L}_\Lambda} \| \middle|\ \omega_\Lambda\colon\text{state over }\mathfrak{A}_\Lambda \right\}
    \end{align*}
    with the spectral projection $\pi_\Lambda\colon\mathfrak{A}_\Lambda\to\mathfrak{A}_\Lambda$ introduced in Lemma~\ref{lem:NESS_in_finite_systems}.
    Since the pre-cutoff time $t_{1,\Lambda}$ diverges as $|\Lambda|$ increases, their result implies that an evolved state $\omega_\Lambda\circ e^{t\mathcal{L}_\Lambda}$ in an arbitrary time $t\ge 0$ is distinguishable from the asymptotic state $\omega_\Lambda\circ\pi_\Lambda$ after the thermodynamic limit.

    It is remarkable that this result on the pre-cutoff time never contradicts Theorem~\ref{cor:onsite_Liouvillian} or Proposition~\ref{prp:onsite_Liouvillian} on commutativity of the thermodynamic limit and the long-time limit (or the long-time average).
    This is due to the difference of topology. 
    Throughout the discussion so far, we have used the weak-$*$ topology to define and calculate the NESS and the TANESS.
    On the other hand, the argument of pre-cutoff times is based on the norm topology in the dual space $\mathfrak{A}^*$, which is stronger than the weak-$*$ one.
    Therefore, the diverging pre-cutoff time regarding the norm distance is compatible with the commutable order of the thermodynamic and long-time limits in the weak-$*$ topology.
\end{rmk}

\section{An example}\label{sec:example}
In this section, we investigate a specific model for which the NESS (or TANESS) does not coincide with Eq.~\eqref{eq:LTLthenTDL} in Theorem~\ref{thm:main_theorem} (or Eq.~\eqref{eq:LTLthenTDL2} in Theorem~\ref{thm:main_theorem2}) for an initial state due to unboundedness of the condition number.
We can acquire exact time evolution of a certain observable and solve the eigenvalue problem of the Liouvillians to calculate the spectra and the condition number.

\subsection{Description of the model}
Let $\hat{\sigma}^0,\hat{\sigma}^1,\hat{\sigma}^2,\hat{\sigma}^3$ denote the $2\times 2$ identity matrix and the Pauli matrices, i.e.,
\begin{align*}
    \hat{\sigma}^0
    = \begin{pmatrix} 1 & 0 \\ 0 & 1 \end{pmatrix},
    \quad
    \hat{\sigma}^1
    = \begin{pmatrix} 0 & 1 \\ 1 & 0 \end{pmatrix},
    \quad
    \hat{\sigma}^2
    = \begin{pmatrix} 0 & -\ii \\ \ii & 0 \end{pmatrix},
    \quad
    \hat{\sigma}^3
    = \begin{pmatrix} 1 & 0 \\ 0 & -1 \end{pmatrix}.
\end{align*}
Throughout this section, we choose the countable set $\Gamma$ and the algebra $\mathfrak{A}_\Lambda$ for each finite subset $\Lambda\Subset\Gamma$ as
$\Gamma=\mathbb{N}$ and $\mathfrak{A}_\Lambda=B(\mathbb{C}^4)^{\otimes|\Lambda|}$.
We use the distance $\dist(x,y)=|x-y|$ ($x,y\in\mathbb{N}$) over $\Gamma$.
To define the Liouvillians, we give the following. For any $Z\Subset\Gamma$, we assign $\hat{\Psi}(Z)=0$ and
\begin{align*}
    \hat{L}(Z)
    = \begin{dcases*}
        \hat{\sigma}^0\otimes\hat{G}\otimes\hat{\sigma}^0 & if $Z=\{j,j+1\}$, $j\in\mathbb{N}$; \\
        \hat{G} & if $Z=\{j\}$, $j\in\mathbb{N}$; \\ 
        0 & otherwise,
    \end{dcases*}
\end{align*}
with
\begin{align}\label{eq:def_G}
    \hat{G}
    = \frac{1}{2}(\hat{\sigma}^1\otimes\hat{\sigma}^0 + \ii\hat{\sigma}^2\otimes\hat{\sigma}^3)
    \in B(\mathbb{C}^4).
\end{align}
From them, we have a Liouvillian $\mathcal{L}_\Lambda\colon\mathfrak{A}_\Lambda\to\mathfrak{A}_\Lambda$ defined as in Eq.~\eqref{eq:def_of_Liouvillian}.
Unless otherwise noted, we deal only with this Liouvillian throughout this section.

The model above satisfies Assumption~\ref{asm:locality}.
Indeed, Assumption~\ref{asm:locality}~(1) is satisfied by taking $F(a)=(1+a)^{-\alpha}$ with $\alpha>1$ and Assumption~\ref{asm:locality}~(2) holds for any $\mu\ge 0$ because
\begin{align*}
    \|\Psi_Z\|_{\cb}
    \le 2\|\hat{L}(Z)\|^2
    = \begin{dcases*}
        2 & if $Z=\{j\}$ or $\{j,j+1\}$, $j\in\mathbb{N}$; \\
        0 & otherwise.
    \end{dcases*}
\end{align*}
Thereby, the dynamics $\gamma_t^\Gamma$ over the whole quantum spin system $\mathfrak{A}$ exists, so we can define the NESS and the TANESS on the dynamics.

\subsection{Time evolution of an observable and the expectation value for NESS and TANESS}
In this subsection, we calculate $\omega\circ\gamma_t^{\Lambda_n}(\hat{O})$ for a particular observable $\hat{O}=\hat{\sigma}_{j=1}^3\coloneqq\hat{\sigma}^3\otimes\hat{\sigma}^0\otimes\hat{\sigma}^0\otimes\hat{\sigma}^0\otimes\cdots\in\mathfrak{A}$ and subsets $\Lambda_n=\{1,2,\ldots,n\}\Subset\Gamma$, $n\in\mathbb{N}$.
The calculation of $\omega\circ\gamma_t^{\Lambda_n}$ is simply a matter of linear algebra.

For any self-adjoint operator $\hat{A}\in\mathfrak{A}_\Lambda$, the Liouvillian on $\Lambda_n$ has the form
\begin{align}\label{eq:expression_of_L_Lambda_n}
    \mathcal{L}_{\Lambda_n}(\hat{A})
    = \frac{1}{2}\sum_{m=1}^{2n-1} \left( (\hat{\sigma}^0)^{\otimes(m-1)}\otimes\hat{G}^*\otimes(\hat{\sigma}^0)^{\otimes(2n-m-1)}\left[\hat{A}, (\hat{\sigma}^0)^{\otimes(m-1)}\otimes\hat{G}\otimes(\hat{\sigma}^0)^{\otimes(2n-m-1)} \right] + \text{h.c.} \right),
\end{align}
where $\text{h.c.}$ denotes the Hermitian conjugate of all preceding terms in the parentheses.
For calculation of the time-evolution, the following lemma is useful.

\begin{lem}\label{lem:invariant_subspace}
    The $2n$-dimensional subspace
    \begin{align*}
        \operatorname{span}_\mathbb{C}\left\{ \hat{D}_m \mid 1\le m\le 2n \right\}
    \end{align*}
    of $\mathfrak{A}_{\Lambda_n}$ is an invariant subspace of $\mathcal{L}_{\Lambda_n}$, where
    \begin{align*}
        \hat{D}_m
        \coloneqq (\hat{\sigma}^0)^{\otimes (m-1)}\otimes\hat{\sigma}^3\otimes(\hat{\sigma}^0)^{\otimes(2n-m)}\in\mathfrak{A}_{\Lambda_n}.
    \end{align*}
\end{lem}
\begin{proof}
    Since on $\mathbb{C}^4$ it holds that
    \begin{align*}
        \hat{G}^* \left[\hat{\sigma}^0\otimes\hat{\sigma}^0, \hat{G} \right]
        = 0, \quad
        \hat{G}^* \left[\hat{\sigma}^0\otimes\hat{\sigma}^3, \hat{G} \right]
        = 0, \quad
        \hat{G}^* \left[\hat{\sigma}^3\otimes\hat{\sigma}^0, \hat{G} \right]
        = -\hat{\sigma}^3\otimes\hat{\sigma}^0 + \hat{\sigma}^0\otimes\hat{\sigma}^3,
    \end{align*}
    we obtain
    \begin{align*}
        \mathcal{L}_{\Lambda_n}(\hat{D}_m)
        = \begin{dcases*}
            -\hat{D}_m + \hat{D}_{m+1} & if $1\le m<2n$ \\
            0 & if $m = 2n$,
        \end{dcases*}
    \end{align*}
    and hence
    \begin{align}\label{eq:reduced_Liouvillian}
        \begin{pmatrix}
            \mathcal{L}_{\Lambda_n}(\hat{D}_1) \\
            \mathcal{L}_{\Lambda_n}(\hat{D}_2) \\
            \mathcal{L}_{\Lambda_n}(\hat{D}_3) \\
            \mathcal{L}_{\Lambda_n}(\hat{D}_4) \\
            \vdots \\
            \mathcal{L}_{\Lambda_n}(\hat{D}_{2n-1}) \\
            \mathcal{L}_{\Lambda_n}(\hat{D}_{2n})
        \end{pmatrix}
        = \mathsf{L}_n
        \begin{pmatrix}
            \hat{D}_1 \\
            \hat{D}_2 \\
            \hat{D}_3 \\
            \hat{D}_4 \\
            \vdots \\
            \hat{D}_{2n-1} \\
            \hat{D}_{2n}
        \end{pmatrix},
        \quad
        \mathsf{L}_n
        \coloneqq
        \begin{pmatrix}
            -1 & 1 & 0 & 0 & \cdots & 0 & 0 \\
            0 & -1 & 1 & 0 & \cdots & 0 & 0 \\
            0 & 0 & -1 & 1 & \cdots & 0 & 0 \\
            0 & 0 & 0 & -1 & \ddots & 0 & 0 \\
            \vdots & \vdots & \vdots & \vdots & \ddots & \ddots & \vdots \\
            0 & 0 & 0 & 0 & \cdots & -1 & 1 \\
            0 & 0 & 0 & 0 & \cdots & 0 & 0
        \end{pmatrix}
        \in \mathrm{M}_{2n}.
    \end{align}
\end{proof}

From this, the dimension of the vector space to be considered reduces to $2n$, while the original $C^*$-algebra $\mathfrak{A}_{\Lambda_n}$ is of $2^{4n}$ dimension.
The Jordan canonical form of $\mathsf{L}_n$ is given by
\begin{gather*}
    \mathsf{L}_n
    = \mathsf{S}_n^{-1}\mathsf{J}_n\mathsf{S}_n, \\
    \mathsf{S}_n^{-1}
    \coloneqq 
    \left( \begin{array}{ccc|c}
         &  &  & 1 \\
         & I_{2n-1} &  & \vdots \\
         &  &  & 1 \\ \hline
        0 & \cdots & 0 & 1
    \end{array} \right), 
    \quad
    \mathsf{S}_n
    = \left( \begin{array}{ccc|c}
         &  &  & -1 \\
         & I_{2n-1} &  & \vdots \\
         &  &  & -1 \\ \hline
        0 & \cdots & 0 & 1
    \end{array} \right), \\
    \mathsf{J}_n
    \coloneqq
    \begin{pmatrix}
        -1 & 1 & 0 & \cdots & 0 & 0 & 0 \\
        0 & -1 & 1 & \cdots & 0 & 0 & 0 \\
        0 & 0 & -1 & \ddots & 0 & 0 & 0 \\
        \vdots & \vdots & \vdots & \ddots & \ddots & \vdots & \vdots \\
        0 & 0 & 0 & \cdots & -1 & 1 & 0 \\
        0 & 0 & 0 & \cdots & 0 & -1 & 0 \\ 
        0 & 0 & 0 & \cdots & 0 & 0 & 0
    \end{pmatrix}, 
\end{gather*}
so we have
\begin{align*}
    \gamma_t^{\Lambda_n}(\hat{\sigma}_{j=1}^3)
    &= e^{t\mathcal{L}_{\Lambda_n}}(\hat{D}_1)
    = \begin{pmatrix}
        1 & 0 & \cdots & 0
    \end{pmatrix}\mathsf{S}_n^{-1} e^{t\mathsf{J}_n} \mathsf{S}_n\begin{pmatrix}
        \hat{D}_1 & \hat{D}_2 & \cdots & \hat{D}_{2n}
    \end{pmatrix}^{\mathsf{T}} \\
    &= e^{-t}\sum_{m=1}^{2n-1}\frac{t^{m-1}}{(m-1)!}\hat{D}_m + \left( 1 - e^{-t}\sum_{m=1}^{2n-1}\frac{t^{m-1}}{(m-1)!} \right)\hat{D}_{2n}.
\end{align*}

Now we take an initial state $\omega_0$ satisfying
\begin{align*}
    \omega_0(\hat{D}_{2j-1})
    = 1,
    \quad
    \omega_0(\hat{D}_{2j})
    = 0,
    \quad
    j\in\mathbb{N}.
\end{align*}
Then, we get
\begin{align*}
    \omega_0\circ\gamma_t^\Gamma(\hat{\sigma}_{j=1}^3)
    = \lim_{n\to\infty} \omega_0\circ\gamma_t^{\Lambda_n}(\hat{\sigma}_{j=1}^3)
    = \lim_{n\to\infty} e^{-t}\sum_{j=1}^n \frac{t^{2j-2}}{(2j-2)!}
    = e^{-t}\cosh t
\end{align*}
for $t\ge 0$.
In this model, it is unclear whether the NESS is unique for the initial state $\omega_0$ or not.
Yet, since $e^{-t}\cosh t$ converges to a common value $\lim_{t\to\infty}e^{-t}\cosh t=1/2$, we have 
\begin{align*}
    \omega_{\sss}(\hat{\sigma}_{j=1}^3)
    = \frac{1}{2}
\end{align*}
for any NESS $\omega_{\sss}$ with the initial state $\omega_0$.
In the same way, we obtain $\omega_{\sss}^{\text{ave}}(\hat{\sigma}_{j=1}^3)=1/2$ for any TANESS $\omega_{\sss}^{\text{ave}}$ with $\omega_0$.

We can also calculate another order of limits in Eqs.~\eqref{eq:LTLthenTDL} and \eqref{eq:LTLthenTDL2}.
Taking the long-time limit before the thermodynamic limit brings 
\begin{align*}
    \lim_{n\to\infty}\lim_{t\to\infty}\omega_0\circ\gamma_t^{\Lambda_n}(\hat{\sigma}_{j=1}^3)
    = \lim_{n\to\infty}\sum_{j=1}^n\frac{1}{(2j-2)!}\lim_{t\to\infty}t^{2j-2}e^{-t}
    = 0
\end{align*}
and
\begin{align*}
    \lim_{n\to\infty}\lim_{T\to\infty}\frac{1}{T}\int_0^T\omega_0\circ\gamma_t^{\Lambda_n}(\hat{\sigma}_{j=1}^3)\dif t
    = \lim_{n\to\infty}\sum_{j=1}^n\frac{1}{(2j-2)!}\lim_{T\to\infty}\frac{1}{T}\int_0^T t^{2j-2}e^{-t}\dif t
    = 0.
\end{align*}
In the present model, thus, the thermodynamic limit and the long-time limit do not commute and hence Eqs.~\eqref{eq:LTLthenTDL} and \eqref{eq:LTLthenTDL2} in Theorems~\ref{thm:main_theorem} and \ref{thm:main_theorem2} do not give the NESS and the TANESS with the initial state $\omega_0$.

Before moving on to another topic, let us discuss the differences between this model and another related model.
We find Eq.~\eqref{eq:reduced_Liouvillian} to be similar with the master equation for unidirectional random walk of a single particle.
An extension of such a classical stochastic process to many-particle systems is well-known as the totally asymmetric simple exclusion process (TASEP).
However, our model differs from the TASEP in some senses.
First, the TASEP has no counterpart of the observables $\hat{\sigma}^1,\hat{\sigma}^2$ in our model, which give rise to the non-commutative structure in the quantum spin system $\mathfrak{A}$.
Second, many-body correlations among $\hat{\sigma}^3$'s caused by the Liouvillian $\mathcal{L}_{\Lambda_n}$ do not reproduce the exclusion process in the TASEP.

\subsection{Spectrum of Liouvillian}
As seen in the previous subsection, our model should violate the assumption (A1) or (A2) in Theorem~\ref{thm:main_theorem}, and (B1) or (B2) in Theorem~\ref{thm:main_theorem2}.
To check this directly, we calculate the exact spectrum and a lower bound of the condition number of each Liouvillian $\mathcal{L}_{\Lambda_n}$.

We first point out that the Liouvillians $\mathcal{L}_{\Lambda_n}$ in our model have a triangular structure for a certain basis.
For $\boldsymbol{\mu}\in\{0,1,2,3\}^n$ with $n\in\mathbb{N}$, we use a notation $\hat{\sigma}^{\boldsymbol{\mu}}=\hat{\sigma}^{\mu_1}\otimes\cdots\otimes\hat{\sigma}^{\mu_n}\in B((\mathbb{C}^2)^{\otimes n})$.

\begin{lem}\label{lem:L_Lambda_n_is_triangular}
    Let $\boldsymbol{\mu}\prec\boldsymbol{\nu}$ be the lexicographical order of $\boldsymbol{\mu},\boldsymbol{\nu}\in\{0,1,2,3\}^n$, i.e., $\boldsymbol{\mu}\prec\boldsymbol{\nu}$ if and only if $\boldsymbol{\mu}\neq\boldsymbol{\nu}$ and $\mu_j<\nu_j$ hold for $j=\min\{i\in\{1,2,\ldots,n\}\mid\mu_i\neq\nu_i\}$.
    For any $n\in\mathbb{N}$, if $\boldsymbol{\mu},\boldsymbol{\nu}\in\{0,1,2,3\}^{2n}$ are ordered by $\boldsymbol{\mu}\prec\boldsymbol{\nu}$, then we have
    \begin{align*}
        \braket{\hat{\sigma}^{\boldsymbol{\mu}},\mathcal{L}_{\Lambda_n}(\hat{\sigma}^{\boldsymbol{\nu}})}_{\text{HS}}
        = 0,
    \end{align*}
    where $\braket{\bullet,\bullet}_{\text{HS}}$ is the Hilbert-Schmidt inner product defined by 
    \begin{align*}
        \braket{\hat{A},\hat{B}}_{\text{HS}}
        \coloneqq\frac{1}{2^n}\tr(\hat{A}^*\hat{B})
    \end{align*}
    for $\hat{A},\hat{B}\in B((\mathbb{C}^2)^{\otimes n})$.
\end{lem}
Note that $\hat{\sigma}^{\boldsymbol{\mu}}$ is normalized with respect to the Hilbert-Schmidt inner product, i.e., 
\begin{align*}
    \braket{\hat{\sigma}^{\boldsymbol{\mu}},\hat{\sigma}^{\boldsymbol{\mu}}}_{\text{HS}}=1.
\end{align*}
\begin{proof}
    We define $\mathcal{D}\colon B(\mathbb{C}^4)\to B(\mathbb{C}^4)$ by
    \begin{align*}
        \mathcal{D}(\hat{A})
        \coloneqq \frac{1}{2}(\hat{G}^*[\hat{A},\hat{G}] + \text{h.c.}),
        \quad
        \hat{A}\in B(\mathbb{C}^4),
    \end{align*}
    where $\hat{G}$ is given in Eq.~\eqref{eq:def_G}.
    A straightforward calculation yields the action of $\mathcal{D}$ to the basis $(\hat{\sigma}^{\mu_1}\otimes\hat{\sigma}^{\mu_2})_{\mu_1,\mu_2\in\{0,1,2,3\}}$ of $B(\mathbb{C}^4)$ as
    \begin{gather}
        \label{eq:action_of_D_top}
        \mathcal{D}(\hat{\sigma}^0\otimes\hat{\sigma}^0)
        = 0,
        \quad
        \mathcal{D}(\hat{\sigma}^0\otimes\hat{\sigma}^1)
        = -\frac{1}{2}\hat{\sigma}^0\otimes\hat{\sigma}^1,
        \\
        \mathcal{D}(\hat{\sigma}^0\otimes\hat{\sigma}^2)
        = -\frac{1}{2}\hat{\sigma}^0\otimes\hat{\sigma}^2,
        \quad
        \mathcal{D}(\hat{\sigma}^0\otimes\hat{\sigma}^3)
        = 0,
        \\
        \mathcal{D}(\hat{\sigma}^1\otimes\hat{\sigma}^0)
        = -\frac{1}{2}\hat{\sigma}^1\otimes\hat{\sigma}^0,
        \quad
        \mathcal{D}(\hat{\sigma}^1\otimes\hat{\sigma}^1)
        = 0,
        \\
        \mathcal{D}(\hat{\sigma}^1\otimes\hat{\sigma}^2)
        = 0,
        \quad
        \mathcal{D}(\hat{\sigma}^1\otimes\hat{\sigma}^3)
        = -\frac{1}{2}\hat{\sigma}^1\otimes\hat{\sigma}^3,
        \\
        \mathcal{D}(\hat{\sigma}^2\otimes\hat{\sigma}^0)
        = -\frac{1}{2}\hat{\sigma}^2\otimes\hat{\sigma}^0,
        \quad
        \mathcal{D}(\hat{\sigma}^2\otimes\hat{\sigma}^1)
        = -\hat{\sigma}^2\otimes\hat{\sigma}^1 + \hat{\sigma}^1\otimes\hat{\sigma}^2,
        \\
        \mathcal{D}(\hat{\sigma}^2\otimes\hat{\sigma}^2)
        = -\hat{\sigma}^2\otimes\hat{\sigma}^2 - \hat{\sigma}^1\otimes\hat{\sigma}^1,
        \quad
        \mathcal{D}(\hat{\sigma}^2\otimes\hat{\sigma}^3)
        = -\frac{1}{2}\hat{\sigma}^2\otimes\hat{\sigma}^3,
        \\
        \mathcal{D}(\hat{\sigma}^3\otimes\hat{\sigma}^0)
        = -\hat{\sigma}^3\otimes\hat{\sigma}^0 + \hat{\sigma}^0\otimes\hat{\sigma}^3,
        \quad
        \mathcal{D}(\hat{\sigma}^3\otimes\hat{\sigma}^1)
        = -\frac{1}{2}\hat{\sigma}^3\otimes\hat{\sigma}^1,
        \\
        \mathcal{D}(\hat{\sigma}^3\otimes\hat{\sigma}^2)
        = -\frac{1}{2}\hat{\sigma}^3\otimes\hat{\sigma}^2,
        \quad
        \mathcal{D}(\hat{\sigma}^3\otimes\hat{\sigma}^3)
        = -\hat{\sigma}^3\otimes\hat{\sigma}^3 + \hat{\sigma}^0\otimes\hat{\sigma}^0.
        \label{eq:action_of_D_bottom}
    \end{gather}
    Thereby, we notice that if $(\mu_1,\mu_2)\prec(\nu_1,\nu_2)$ then we have
    \begin{align*}
        \braket{\hat{\sigma}^{\mu_1}\otimes\hat{\sigma}^{\mu_2},\mathcal{D}(\hat{\sigma}^{\nu_1}\otimes\hat{\sigma}^{\nu_2})}_{\text{HS}}=0.
    \end{align*}

    In view of Eq.~\eqref{eq:expression_of_L_Lambda_n}, $\mathcal{L}_{\Lambda_n}(\hat{\sigma}^{\boldsymbol{\nu}})$ reads 
    \begin{align*}
        \mathcal{L}_{\Lambda_n}(\hat{\sigma}^{\boldsymbol{\nu}})
        = \sum_{m=1}^{2n-1}\left(\bigotimes_{k=1}^{m-1}\hat{\sigma}^{\nu_k}\right)\otimes\mathcal{D}(\hat{\sigma}^{\nu_m}\otimes\hat{\sigma}^{\nu_{m+1}})\otimes\left(\bigotimes_{\ell=m+2}^{2n}\hat{\sigma}^{\nu_\ell}\right),
    \end{align*}
    leading to
    \begin{align*}
        \braket{\hat{\sigma}^{\boldsymbol{\mu}},\mathcal{L}_{\Lambda_n}(\hat{\sigma}^{\boldsymbol{\nu}})}_{\text{HS}}
        = \sum_{m=1}^{2n-1}\left(\prod_{k=1}^{m-1}\delta_{\mu_k,\nu_k}\right)\braket{\hat{\sigma}^{\mu_m}\otimes\hat{\sigma}^{\mu_{m+1}},\mathcal{D}(\hat{\sigma}^{\nu_m}\otimes\hat{\sigma}^{\nu_{m+1}})}_{\text{HS}}\left(\prod_{\ell=m+2}^{2n}\delta_{\mu_\ell,\nu_\ell}\right).
    \end{align*}
    When $\boldsymbol{\mu}\prec\boldsymbol{\nu}$, we have some $\kappa\in\{1,2,\ldots,2n\}$ satisfying both $\mu_m=\nu_m$ for all $1\le m<\kappa$ and $\mu_\kappa<\nu_\kappa$.
    Hence, it follows that
    \begin{align*}
        \braket{\hat{\sigma}^{\boldsymbol{\mu}},\mathcal{L}_{\Lambda_n}(\hat{\sigma}^{\boldsymbol{\nu}})}_{\text{HS}}
        = \braket{\hat{\sigma}^{\mu_{\kappa-1}}\otimes\hat{\sigma}^{\mu_\kappa},\mathcal{D}(\hat{\sigma}^{\nu_{\kappa-1}}\otimes\hat{\sigma}^{\nu_\kappa})}_{\text{HS}}
        + \braket{\hat{\sigma}^{\mu_\kappa}\otimes\hat{\sigma}^{\mu_{\kappa+1}},\mathcal{D}(\hat{\sigma}^{\nu_\kappa}\otimes\hat{\sigma}^{\nu_{\kappa+1}})}_{\text{HS}}.
    \end{align*}
    Since $(\mu_{\kappa-1},\mu_{\kappa})\prec(\nu_{\kappa-1},\nu_\kappa)$ and $(\mu_\kappa,\mu_{\kappa+1})\prec(\nu_\kappa,\nu_{\kappa+1})$ hold, we finally get
    \begin{align*}
        \braket{\hat{\sigma}^{\boldsymbol{\mu}},\mathcal{L}_{\Lambda_n}(\hat{\sigma}^{\boldsymbol{\nu}})}_{\text{HS}}
        = 0
    \end{align*}
    for any $\boldsymbol{\mu},\boldsymbol{\nu}\in\{0,1,2,3\}^{2n}$ with $\boldsymbol{\mu}\prec\boldsymbol{\nu}$.
\end{proof}

This lemma implies that the representation matrix of each $\mathcal{L}_{\Lambda_n}$ are triangular when we take the basis as $(\hat{\sigma}^{\boldsymbol{\mu}})_{\boldsymbol{\mu}\in\{0,1,2,3\}^{2n}}$ aligned with the order determined by $\prec$.
Accordingly, all the eigenvalues of $\mathcal{L}_{\Lambda_n}$ are given by the diagonal elements $\braket{\hat{\sigma}^{\boldsymbol{\mu}},\mathcal{L}_{\Lambda_n}(\hat{\sigma}^{\boldsymbol{\mu}})}_{\text{HS}}$.
On the basis of such an argument, we arrive at the following theorem on the spectral gaps.

\begin{thm}\label{thm:gaps_for_Lambda_n}
    For $n\in\mathbb{N}$, the Liouvillian $\mathcal{L}_{\Lambda_n}$ has four-fold degenerate eigenvalue $0$ and the gaps
    \begin{align*}
        \Delta_{\Lambda_n}
        = \Delta^{\text{p}}_{\Lambda_n}
        = \frac{1}{2},
        \quad
        \Delta^{\text{ex}}_{\Lambda_n}
        = 1.
    \end{align*}
\end{thm}
\begin{proof}
    It is sufficient to calculate $\braket{\hat{\sigma}^{\boldsymbol{\mu}},\mathcal{L}_{\Lambda_n}(\hat{\sigma}^{\boldsymbol{\mu}})}_{\text{HS}}$ for any $\boldsymbol{\mu}\in\{0,1,2,3\}^{2n}$.
    Here we introduce a symbol
    \begin{align*}
        \delta_\nu(\mu)
        = \begin{dcases*}
            1 & if $\mu=\nu$; \\
            0 & if $\mu\neq\nu$.
        \end{dcases*}
    \end{align*}
    Using some computations in the proof of Lemma~\ref{lem:L_Lambda_n_is_triangular}, we obtain
    \begin{align*}
        &\braket{\hat{\sigma}^{\boldsymbol{\mu}},\mathcal{L}_{\Lambda_n}(\hat{\sigma}^{\boldsymbol{\mu}})}_{\text{HS}} \notag \\
        &= \sum_{m=1}^{2n-1}\braket{\hat{\sigma}^{\mu_m}\otimes\hat{\sigma}^{\mu_{m+1}},\mathcal{D}(\hat{\sigma}^{\mu_m}\otimes\hat{\sigma}^{\mu_{m+1}})}_{\text{HS}} \notag \\
        &= \sum_{m=1}^{2n-1}\Bigg[ -\frac{1}{2}(\delta_0(\mu_m)+\delta_3(\mu_m))(\delta_1(\mu_{m+1})+\delta_2(\mu_{m+1})) \notag \\
        &\qquad\qquad- \frac{1}{2}(\delta_1(\mu_m)+\delta_2(\mu_m))(\delta_0(\mu_{m+1})+\delta_3(\mu_{m+1})) \notag \\
        &\qquad\qquad- \delta_2(\mu_m)(\delta_1(\mu_{m+1})+\delta_{2}(\mu+1)) - \delta_3(\mu_m)(\delta_0(\mu_{m+1})+\delta_3(\mu_{m+1})) \Bigg] \notag \\
        &= -\frac{1}{2}\sum_{m=1}^{2n-1}\Big[ \delta_0(\mu_m)(\delta_1(\mu_{m+1})+\delta_2(\mu_{m+1}))
        + \delta_1(\mu_m)(\delta_0(\mu_{m+1})+\delta_3(\mu_{m+1})) \notag \\
        &\qquad\qquad\quad+ \delta_2(\mu_m)(1+\delta_1(\mu_{m+1})+\delta_2(\mu_{m+1})) + \delta_3(\mu_m)(1+\delta_0(\mu_{m+1})+\delta_3(\mu_{m+1})) \Big].
    \end{align*}
    Since all terms in the sum is a nonnegative integer, $\braket{\hat{\sigma}^{\boldsymbol{\mu}},\mathcal{L}_{\Lambda_n}(\hat{\sigma}^{\boldsymbol{\mu}})}_{\text{HS}}$ can be zero or a negative half-integer, i.e., 
    \begin{align*}
        \sigma_{B(\mathfrak{A}_{\Lambda_n})}(\mathcal{L}_{\Lambda_n})
        \subset -\frac{1}{2}\mathbb{Z}_{\ge 0}
        = \left\{ 0, -\frac{1}{2},-1,-\frac{3}{2},\ldots \right\}.
    \end{align*}

    To obtain $\braket{\hat{\sigma}^{\boldsymbol{\mu}},\mathcal{L}_{\Lambda_n}(\hat{\sigma}^{\boldsymbol{\mu}})}_{\text{HS}}=0$, it is necessary to see either $(\mu_m,\mu_{m+1})=(0,0),(0,3),(1,1),$ or $(1,2)$ for any $m=1,2,\ldots,2n$. 
    This condition is satisfied only in four cases $\boldsymbol{\mu}=(0,\ldots,0,0),(0,\ldots,0,3),(1,\ldots,1,1),$ and $(1,\ldots,1,2)$.
    Conversely, we can check $\mathcal{L}_{\Lambda_n}(\hat{\sigma}^{\boldsymbol{\mu}})=0$ for the four types of $\boldsymbol{\mu}$.
    Hence, the eigenvalue 0 of $\mathcal{L}_{\Lambda_n}$ is four-fold degenerate.

    Next, we will show that $-1/2$ is a semisimple eigenvalue of $\mathcal{L}_{\Lambda_n}$. 
    It is easy to check that $-1/2$ is an eigenvalue. 
    For example, we can show that $\mathcal{L}_{\Lambda_n}((\hat{\sigma}^0)^{\otimes(2n-1)}\otimes\hat{\sigma}^1)=(-1/2)(\hat{\sigma}^0)^{\otimes(2n-1)}\otimes\hat{\sigma}^1$.
    To see the semisimplicity, it is sufficient to prove that for any $\boldsymbol{\mu}\in\{0,1,2,3\}^{2n}$ satisfying $\braket{\hat{\sigma}^{\boldsymbol{\mu}},\mathcal{L}_{\Lambda_n}(\hat{\sigma}^{\boldsymbol{\mu}})}_{\text{HS}}=-1/2$, we have $\mathcal{L}_{\Lambda_n}(\hat{\sigma}^{\boldsymbol{\mu}})=(-1/2)\hat{\sigma}^{\boldsymbol{\mu}}$.
    When $\braket{\hat{\sigma}^{\boldsymbol{\mu}},\mathcal{L}_{\Lambda_n}(\hat{\sigma}^{\boldsymbol{\mu}})}_{\text{HS}}=-1/2$ holds, we have some $\kappa\in\{1,2,\ldots,2n-1\}$ such that
    \begin{align}\label{eq:diagonal_element_is_-1/2}
        \braket{\hat{\sigma}^{\mu_\kappa}\otimes\hat{\sigma}^{\mu_{\kappa+1}},\mathcal{D}(\hat{\sigma}^{\mu_\kappa}\otimes\hat{\sigma}^{\mu_{\kappa+1}})}_{\text{HS}}
        = -\frac{1}{2}
    \end{align}
    and
    \begin{align}\label{eq:diagonal_element_is_0}
        \braket{\hat{\sigma}^{\mu_m}\otimes\hat{\sigma}^{\mu_{m+1}},\mathcal{D}(\hat{\sigma}^{\mu_m}\otimes\hat{\sigma}^{\mu_{m+1}})}_{\text{HS}}
        = 0
    \end{align}
    for any $m\neq \kappa$.
    From Eqs.~\eqref{eq:action_of_D_top}-\eqref{eq:action_of_D_bottom}, Eq.~\eqref{eq:diagonal_element_is_-1/2} follows $\mathcal{D}(\hat{\sigma}^{\mu_\kappa}\otimes\hat{\sigma}^{\mu_{\kappa+1}})=(-1/2)\hat{\sigma}^{\mu_\kappa}\otimes\hat{\sigma}^{\mu_{\kappa+1}}$ and Eq.~\eqref{eq:diagonal_element_is_0} does $\mathcal{D}(\hat{\sigma}^{\mu_m}\otimes\hat{\sigma}^{\mu_{m+1}})=0$. 
    Thus, $\braket{\hat{\sigma}^{\boldsymbol{\mu}},\mathcal{L}_{\Lambda_n}(\hat{\sigma}^{\boldsymbol{\mu}})}_{\text{HS}}=-1/2$ results in $\mathcal{L}_{\Lambda_n}(\hat{\sigma}^{\boldsymbol{\mu}})=(-1/2)\hat{\sigma}^{\boldsymbol{\mu}}$.
    Consequently, we get $\Delta_{\Lambda_n}=\Delta^{\text{p}}_{\Lambda_n}=1/2$ and $\Delta^{\text{ex}}_{\Lambda_n}\ge 1$.
    Recalling Lemma~\ref{lem:invariant_subspace} and its proof, we find that $1$ is a non-semisimple eigenvalue of $\mathcal{L}_{\Lambda_n}$ and obtain $\Delta^{\text{ex}}_{\Lambda_n}=1$.
\end{proof}

We can extend this results to general $\Lambda\Subset\Gamma=\mathbb{N}$.

\begin{cor}\label{cor:gaps_for_general_Lambda}
    For any nonempty finite subset $\Lambda\Subset\Gamma$, the gaps of $\mathcal{L}_\Lambda$ are 
    \begin{align}\label{eq:gaps_for_general_Lambda}
        \Delta_\Lambda
        = \Delta^{\mathrm{p}}_\Lambda
        = \frac{1}{2},
        \quad
        \Delta^{\text{ex}}_\Lambda
        = 1.
    \end{align}
    In particular, for Eqs.~\eqref{eq:line_gap}, \eqref{eq:point_gap}, and \eqref{eq:existence_of_ex_gap}, we can choose
    \begin{align*}
        \Delta
        = \Delta^{\text{p}}
        = \frac{1}{2},
        \quad
        \Delta^{\text{ex}}
        = 1.
    \end{align*}
\end{cor}
\begin{proof}
    For any $\emptyset\neq\Lambda\Subset\Gamma$, there exists a unique decomposition 
    \begin{align*}
        \Lambda
        = \coprod_{\eta=1}^\Xi Z_\eta
    \end{align*}
    with the set $Z_\eta=\{n\in\mathbb{N}\mid n_\eta^0\le n\le n_\eta^1\}$ of successive integers satisfying $n_\eta^1<n_{\eta+1}^0$ if $\Xi\ge 2$.
    The Liouvillian $\mathcal{L}_\Lambda$ over $\mathfrak{A}_\Lambda=\bigotimes_{\eta=1}^\Xi\mathfrak{A}_{Z_\eta}$ is also decomposed into the direct sum
    \begin{align*}
        \mathcal{L}_\Lambda
        = \bigoplus_{\eta=1}^\Xi\mathcal{L}_{Z_\eta},
    \end{align*}
    so the spectrum of $\mathcal{L}_\Lambda$ is simply the union of the spectra of $\mathcal{L}_{Z_\eta}$, $\eta=1,2,\ldots,\Xi$, i.e.,
    \begin{align*}
        \sigma_{B(\mathfrak{A}_\Lambda)}(\mathcal{L}_\Lambda)
        = \bigcup_{\eta=1}^\xi \sigma_{B(\mathfrak{A}_{Z_\eta})}(\mathcal{L}_{Z_\eta}),
        \quad
        \sigma^{\text{ex}}_{B(\mathfrak{A}_\Lambda)}(\mathcal{L}_\Lambda)
        = \bigcup_{\eta=1}^\xi \sigma^{\text{ex}}_{B(\mathfrak{A}_{Z_\eta})}(\mathcal{L}_{Z_\eta}).
    \end{align*}
    Since $\mathcal{L}_{Z_\eta}$ is equivalent to $\mathcal{L}_{\Lambda_{n(\eta)}}$ with $n(\eta)=|Z_\eta|=n_\eta^1-n_\eta^0+1$ as a linear map, from Theorem~\ref{thm:gaps_for_Lambda_n} we have Eq.~\eqref{eq:gaps_for_general_Lambda}.
\end{proof}

\subsection{Condition number}

The results in Corollary~\ref{cor:gaps_for_general_Lambda} implies that the present model satisfies the assumptions (A1) and (B1) in Theorems~\ref{thm:main_theorem} and \ref{thm:main_theorem2}.
Therefore, both assumptions (A2) and (B2) should not be satisfied, i.e., for any $\epsilon\in(0,\Delta^{\text{ex}}=1)$ and some $\Lambda'\Subset\Gamma$, the condition number $\kappa_{\Lambda'}(\mathcal{V}_{\Lambda}^\epsilon)$ should be unbounded with respect to $\Lambda'\subset\Lambda\Subset\Gamma$.
In this subsection, we will check this by showing a lower bound of the condition number diverges to infinity as $n\to\infty$.

To evaluate the condition number for our model, we determine the generalized eigenspace associated with the eigenvalue $-1$ of $\mathcal{L}_{\Lambda_n}$.
\begin{lem}\label{lem:generalized_eigenspace}
    For $n\in\mathbb{N}$, the generalized eigenspace of the Liouvillian $\mathcal{L}_{\Lambda_n}$ associated with the eigenvalue $-1$ is 
    \begin{align*}
        \operatorname{span}_\mathbb{C}\{\hat{D}_m\hat{E}^\alpha\hat{X}^\beta,\hat{F}_q \mid m\in\{1,2,\ldots,2n-1\},\alpha,\beta\in\{0,1\},q\in\{1,2,\ldots,K_n\}\}
    \end{align*}
    with some positive integer $K_n$, where
    \begin{gather*}
        \hat{D}_m
        \coloneqq (\hat{\sigma}^0)^{\otimes (m-1)}\otimes\hat{\sigma}^3\otimes(\hat{\sigma}^0)^{\otimes(2n-m)}\in\mathfrak{A}_{\Lambda_n}, \\
        \hat{E}
        \coloneqq \hat{D}_{2n}
        = (\hat{\sigma}^0)^{\otimes(2n-1)}\otimes\hat{\sigma}^3,
        \quad
        \hat{X}
        \coloneqq (\hat{\sigma}^2)^{\otimes 2n}.
    \end{gather*}
    The operators $\hat{F}_q\in\mathfrak{A}_{\Lambda_n}$ are the eigenoperators of $\mathcal{L}_{\Lambda_n}$, i.e., $\mathcal{L}_{\Lambda_n}(\hat{F}_q)=-\hat{F}_q$.

    Moreover, the subspace
    \begin{align*}
        \mathfrak{I}_{\Lambda_n}
        \coloneqq \operatorname{span}_\mathbb{C}\{\hat{D}_m\hat{E}^\alpha\hat{X}^\beta,\hat{F}_q \mid m\in\{1,2,\ldots,2n\},\alpha,\beta\in\{0,1\},q\in\{1,2,\ldots,K_n\}\}
    \end{align*}
    of $\mathfrak{A}_{\Lambda_n}$ is an invariant subspace of $\mathcal{L}_{\Lambda_n}$.
\end{lem}
\begin{proof}
    We first compute the algebraic multiplicity of the eigenvalue $-1$. From Lemma~\ref{lem:L_Lambda_n_is_triangular}, we just have to count how many $\boldsymbol{\mu}\in\{0,1,2,3\}^{2n}$ satisfy $\braket{\hat{\sigma}^{\boldsymbol{\mu}},\mathcal{L}_{\Lambda_n}(\hat{\sigma}^{\boldsymbol{\mu}})}_{\text{HS}}=-1$.
    There are two cases to realize $\braket{\hat{\sigma}^{\boldsymbol{\mu}},\mathcal{L}_{\Lambda_n}(\hat{\sigma}^{\boldsymbol{\mu}})}_{\text{HS}}=-1$ as follows: 
    \begin{enumerate}[(i)]
        \item There exists $\kappa\in\{1,2,\ldots,2n-1\}$ such that it holds that
        \begin{align}\label{eq:diagonal_element_is_-1}
            \braket{\hat{\sigma}^{\mu_\kappa}\otimes\hat{\sigma}^{\mu_{\kappa+1}},\mathcal{D}(\hat{\sigma}^{\mu_\kappa}\otimes\hat{\sigma}^{\mu_{\kappa+1}})}_{\text{HS}}
            = -1
        \end{align}
        and 
        \begin{align}\label{eq:diagonal_element_is_0_condnum}
            \braket{\hat{\sigma}^{\mu_m}\otimes\hat{\sigma}^{\mu_{m+1}},\mathcal{D}(\hat{\sigma}^{\mu_m}\otimes\hat{\sigma}^{\mu_{m+1}})}_{\text{HS}}
            = 0
        \end{align}
        for any $m\neq\kappa$;
        \item There exist $\kappa_1,\kappa_2\in\{1,2,\ldots,2n-1\}$ with $\kappa_1 \neq \kappa_2$ such that it holds that
        \begin{align*}
            \braket{\hat{\sigma}^{\mu_m}\otimes\hat{\sigma}^{\mu_{m+1}},\mathcal{D}(\hat{\sigma}^{\mu_m}\otimes\hat{\sigma}^{\mu_{m+1}})}_{\text{HS}}
            = -\frac{1}{2}
        \end{align*}
        for $m=\kappa_1,\kappa_2$ and
        \begin{align*}
            \braket{\hat{\sigma}^{\mu_m}\otimes\hat{\sigma}^{\mu_{m+1}},\mathcal{D}(\hat{\sigma}^{\mu_m}\otimes\hat{\sigma}^{\mu_{m+1}})}_{\text{HS}}
            = 0
        \end{align*}
        for $m\neq \kappa_1,\kappa_2$.
    \end{enumerate}
    We begin with the case (i).
    To obtain Eq.~\eqref{eq:diagonal_element_is_-1} it is necessary and sufficient to satisfy $(\mu_\kappa,\mu_{\kappa+1})=(3,0),(3,3),(2,1),$ or $(2,2)$ while Eq.~\eqref{eq:diagonal_element_is_0_condnum} is equivalent to $(\mu_\kappa,\mu_{\kappa+1})=(0,0),(0,3),(1,1),$ or $(1,2)$.
    Therefore, all possible $\boldsymbol{\mu}\in\{0,1,2,3\}^{2n}$ to realize the case (i) are following $4(2n-1)$ patterns:
    \begin{itemize}
        \item $\mu_\kappa = 3$ and $\mu_m = 0$ for some $\kappa\in\{1,2,\ldots,2n-1\}$ and any $m\neq\kappa$;
        \item $\mu_\kappa = 3$, $\mu_{2n}=3$, and $\mu_m = 0$ for some $\kappa\in\{1,2,\ldots,2n-1\}$ and any $m\neq\kappa,2n$;
        \item $\mu_\kappa = 2$ and $\mu_m = 1$ for some $\kappa\in\{1,2,\ldots,2n-1\}$ and any $m\neq\kappa$;
        \item $\mu_\kappa = 2$, $\mu_{2n}=2$, and $\mu_m = 1$ for some $\kappa\in\{1,2,\ldots,2n-1\}$ and any $m\neq\kappa,2n$.
    \end{itemize}
    Corresponding to the $\boldsymbol{\mu}$'s, we have $4(2n-1)$ operators $\hat{D}_m,\hat{D}_m\hat{E},\hat{D}_m\hat{X},\hat{D}_m\hat{E}\hat{X}\in\mathfrak{A}_{\Lambda_n}$ ($m=1,2,\ldots,2n-1$), where $\hat{D}_m$ was defined in Lemma~\ref{lem:invariant_subspace} and $\hat{E},\hat{X}$ are introduced by
    \begin{align*}
        \hat{E}
        = \hat{D}_{2n}
        = (\hat{\sigma}^0)^{\otimes(2n-1)}\otimes\hat{\sigma}^3,
        \quad
        \hat{X}
        = (\hat{\sigma}^2)^{\otimes 2n}.
    \end{align*}
    In fact, the unitary operators $\hat{E},\hat{X}$ define $\mathbb{Z}_2$ strong symmetries of the Liouvillian $\mathcal{L}_{\Lambda_n}$ (cf. \cite{Buca2012}),
    \begin{align*}
        [\hat{L}(Z),\hat{E}]
        = [\hat{L}(Z),\hat{X}]
        = 0,
        \quad
        Z\subset\Lambda_n,
        \quad
        \hat{E}^2
        =\hat{X}^2
        = \hat{I}.
    \end{align*}
    From the symmetries and Eq.~\eqref{eq:reduced_Liouvillian}, it follows that
    \begin{align*}
        \begin{pmatrix}
                \mathcal{L}_{\Lambda_n}(\hat{D}_1\hat{E}^\alpha\hat{X}^\beta) \\
                \mathcal{L}_{\Lambda_n}(\hat{D}_2\hat{E}^\alpha\hat{X}^\beta) \\
                \vdots \\
                \mathcal{L}_{\Lambda_n}(\hat{D}_{2n}\hat{E}^\alpha\hat{X}^\beta)
            \end{pmatrix}
            = \mathsf{L}_n
            \begin{pmatrix}
                \hat{D}_1\hat{E}^\alpha\hat{X}^\beta \\
                \hat{D}_2\hat{E}^\alpha\hat{X}^\beta \\
                \vdots \\
                \hat{D}_{2n}\hat{E}^\alpha\hat{X}^\beta
            \end{pmatrix}
    \end{align*}
    for $\alpha,\beta\in\{0,1\}$.
    Thus, we find four operators $\hat{D}_{2n}\hat{E}^\alpha\hat{X}^\beta$ ($\alpha,\beta\in\{0,1\}$) are associated with the eigenvalue 0 of $\mathcal{L}_{\Lambda_n}$ and $4(2n-1)$ operators $\hat{D}_m\hat{E}^\alpha\hat{X}^\beta$ ($m=1,2,\ldots,2n-1$) are elements of the generalized eigenspace with the eigenvalue $-1$ of $\mathcal{L}_{\Lambda_n}$:
    \begin{align*}
        \mathcal{L}_{\Lambda_n}(\hat{D}_{2n}\hat{E}^\alpha\hat{X}^\beta)
        = 0,
        \quad
        (\mathcal{L}_{\Lambda_n})^{2n-m}((\hat{D}_m-\hat{D}_{2n})\hat{E}^\alpha\hat{X}^\beta)
        = -(\hat{D}_m-\hat{D}_{2n})\hat{E}^\alpha\hat{X}^\beta.
    \end{align*}
    On the other hand, if $\boldsymbol{\mu}\in\{0,1,2,3\}^{2n}$ belongs to the case (ii), 
    we have $\mathcal{L}_{\Lambda_n}(\hat{\sigma}^{\boldsymbol{\mu}})=-\hat{\sigma}^{\boldsymbol{\mu}}$ in the same way as the proof of Theorem~\ref{thm:gaps_for_Lambda_n}.
    Let $K_n$ be the number of $\boldsymbol{\mu}$'s satisfying the conditions in the case (ii) and $\hat{F}_1,\hat{F}_2,\ldots,\hat{F}_{K_n}$ the distinct operators of the form $\hat{\sigma}^{\boldsymbol{\mu}}$ satisfying $\mathcal{L}_{\Lambda_n}(\hat{F}_q)=-\hat{F}_q$ for $q=1,2,\ldots,K_n$.
    After all, we get the full generalized eigenspace 
    \begin{align*}
        \operatorname{span}_\mathbb{C}\{\hat{D}_m\hat{E}^\alpha\hat{X}^\beta,\hat{F}_q \mid m\in\{1,2,\ldots,2n-1\},\alpha,\beta\in\{0,1\},q\in\{1,2,\ldots,K_n\}\},
    \end{align*}
    whose dimension $4(2n-1)+K_n$ equals the algebraic multiplicity of the eigenvalue $-1$ of $\mathcal{L}_{\Lambda_n}$.
    Note that $4(2n-1)+K_n$ operators $\hat{D}_m\hat{E}^\alpha\hat{X}^\beta,\hat{F}_q$ are entirely orthonormal because they are of the form $e^{\ii\theta}\hat{\sigma}^{\boldsymbol{\mu}}$ with $\theta\in\{0,\pi/2,\pi\}$.
    Also, the subspace
    \begin{align*}
        \mathfrak{I}_{\Lambda_n}
        \coloneqq \operatorname{span}_\mathbb{C}\{\hat{D}_m\hat{E}^\alpha\hat{X}^\beta,\hat{F}_q \mid m\in\{1,2,\ldots,2n\},\alpha,\beta\in\{0,1\},q\in\{1,2,\ldots,K_n\}\}
    \end{align*}
    of $\mathfrak{A}_{\Lambda_n}$ is an invariant subspace of $\mathcal{L}_{\Lambda_n}$.
\end{proof}

Let us move on to evaluation of the condition number.
Since we just have to show $\sup_{n}\kappa_{\Lambda'}(\mathcal{V}_{\Lambda_n}^\epsilon) = \infty$ for a fixed $\Lambda'\Subset\Gamma$, we particularly take $\Lambda'=\{1\}$.
In addition, the map $\mathcal{V}_{\Lambda_n}^\epsilon$ that transforms $\mathcal{L}_{\Lambda_n}/(1-\epsilon)$ into the Jordan canonical form $\mathcal{J}_{\Lambda_n}^\epsilon$ is not unique; if $\mathcal{J}_{\Lambda_n}^\epsilon$ is invariant under a similarity transformation $\mathcal{J}_{\Lambda_n}^\epsilon\mapsto\mathcal{S}\circ\mathcal{J}_{\Lambda_n}^\epsilon\circ\mathcal{S}^{-1}$, $\mathcal{S}\circ\mathcal{V}_{\Lambda_n}^\epsilon$ also transforms $\mathcal{L}_{\Lambda_n}/(1-\epsilon)$ into $\mathcal{J}_{\Lambda_n}^\epsilon$. 
In the evaluation of the condition number $\kappa_{\{1\}}(\mathcal{V}_{\Lambda_n}^\epsilon)$, we should minimize it with respect to all the possible choices of $\mathcal{V}_{\Lambda_n}^\epsilon$.
Consequently, the statement to be proven is as follows.

\begin{thm}
    For $n\in\mathbb{N}$ and $\epsilon\in(0,1)$, let $\mathcal{J}_{\Lambda_n}^\epsilon\in\mathrm{M}_{2^{4n}}$ be the Jordan canonical form of the normalized Liouvillian $\mathcal{L}_{\Lambda_n}/(1-\epsilon)$.
    Then, we have
    \begin{align*}
        \inf\left\{ \kappa_{\{1\}}(\mathcal{V}_{\Lambda_n}^\epsilon) \middle|\ \mathcal{V}_{\Lambda_n}^\epsilon\in B(\mathfrak{A}_{\Lambda_n},\mathbb{C}^{2^{4n}}),\ \mathcal{J}_{\Lambda_n}^\epsilon = \mathcal{V}_{\Lambda_n}\circ\frac{\mathcal{L}_{\Lambda_n}}{1-\epsilon}\circ(\mathcal{V}_{\Lambda_n}^\epsilon)^{-1} \right\}
        \ge \frac{\sqrt{1+(1-\epsilon)^2}}{(1-\epsilon)^{2n-2}}
    \end{align*}
    for $n\ge 2$.
\end{thm}

\begin{proof}
    Using the result in Lemma~\ref{lem:generalized_eigenspace}, we consider the direct-sum decomposition $\mathfrak{A}_{\Lambda_n}\simeq\mathfrak{I}_{\Lambda_n}\oplus(\mathfrak{A}_{\Lambda_n}/\mathfrak{I}_{\Lambda_n})$.
    The Liouvillian $\mathcal{L}_{\Lambda_n}$ is also decomposed into
    \begin{align*}
        \mathcal{U}_{\Lambda_n} \mathcal{L}_{\Lambda_n} \mathcal{U}_{\Lambda_n}^*
        = \begin{pmatrix}
            \mathcal{L}_{\Lambda_n}\restriction\mathfrak{I}_{\Lambda_n} & \mathcal{M}_{\mathfrak{I}_{\Lambda_n}} \\
            0 & [\mathcal{L}_{\Lambda_n}]
        \end{pmatrix}
    \end{align*}
    for some unitary $\mathcal{U}_{\Lambda_n}\colon\mathfrak{A}_{\Lambda_n}\to\mathfrak{I}_{\Lambda_n}\oplus(\mathfrak{A}_{\Lambda_n}/\mathfrak{I}_{\Lambda_n})$, where $\mathcal{L}_{\Lambda_n}\restriction\mathfrak{I}_{\Lambda_n}\colon\mathfrak{I}_{\Lambda_n}\to\mathfrak{I}_{\Lambda_n}$ is the restriction of $\mathcal{L}_{\Lambda_n}$ onto the invariant subspace $\mathfrak{I}_{\Lambda_n}$ and $[\mathcal{L}_{\Lambda_n}]\colon\mathfrak{A}_{\Lambda_n}/\mathfrak{I}_{\Lambda_n}\to\mathfrak{A}_{\Lambda_n}/\mathfrak{I}_{\Lambda_n}$ is the linear map induced by $\mathcal{L}_{\Lambda_n}$ on the quotient space $\mathfrak{A}_{\Lambda_n}/\mathfrak{I}_{\Lambda_n}$.
    By the construction of $\mathfrak{I}_{\Lambda_n}$, the spectra of $\mathcal{L}_{\Lambda_n}\restriction\mathfrak{I}_{\Lambda_n}$ and $[\mathcal{L}_{\Lambda_n}]$ are $\{0,-1\}$ and $\sigma_{B(\mathfrak{A}_{\Lambda_n})}(\mathcal{L}_{\Lambda_n})\setminus\{0,-1\}$, respectively.
    Let $d_{\text{inv}}=4(2n-1)+K_n$ be the dimension of $\mathfrak{I}_{\Lambda_n}$ and $d_{\text{quot}}=2^{4n}-d_{\text{inv}}$ that of $\mathfrak{A}_{\Lambda_n}/\mathfrak{I}_{\Lambda_n}$.
    We transform $(1-\epsilon)^{-1}\mathcal{L}_{\Lambda_n}\restriction\mathfrak{I}_{\Lambda_n}$ and $(1-\epsilon)^{-1}[\mathcal{L}_{\Lambda_n}]$ into the Jordan canonical forms, 
    \begin{align*}
        \mathcal{J}_{\Lambda_n}^{\text{inv},\epsilon}
        = \mathcal{V}_{\Lambda_n}^{\text{inv},\epsilon}\circ\frac{\mathcal{L}_{\Lambda_n}\restriction\mathfrak{I}_{\Lambda_n}}{1-\epsilon}\circ(\mathcal{V}_{\Lambda_n}^{\text{inv},\epsilon})^{-1},
        \quad
        \mathcal{J}_{\Lambda_n}^{\text{quot},\epsilon}
        = \mathcal{V}_{\Lambda_n}^{\text{quot},\epsilon}\circ\frac{[\mathcal{L}_{\Lambda_n}]}{1-\epsilon}\circ(\mathcal{V}_{\Lambda_n}^{\text{quot},\epsilon})^{-1}
    \end{align*}
    with invertible operators $\mathcal{V}_{\Lambda_n}^{\text{inv},\epsilon}\colon\mathfrak{I}_{\Lambda_n}\to\mathbb{C}^{d_{\text{inv}}}$ and $\mathcal{V}_{\Lambda_n}^{\text{quot},\epsilon}\colon\mathfrak{A}_{\Lambda_n}/\mathfrak{I}_{\Lambda_n}\to\mathbb{C}^{d_{\text{quot}}}$.
    Using these operators, we further perform the similarity transformation
    \begin{align*}
        &(\mathcal{V}_{\Lambda_n}^{\text{inv},\epsilon}\oplus\mathcal{V}_{\Lambda_n}^{\text{quot},\epsilon})\circ\mathcal{U}_{\Lambda_n}\circ\frac{1}{1-\epsilon}\mathcal{L}_{\Lambda_n}\circ\mathcal{U}_{\Lambda_n}^* \circ (\mathcal{V}_{\Lambda_n}^{\text{inv},\epsilon}\oplus\mathcal{V}_{\Lambda_n}^{\text{quot},\epsilon})^{-1} \notag \\
        &= \begin{pmatrix}
            \mathcal{J}_{\Lambda_n}^{\text{inv},\epsilon} & \mathcal{V}_{\Lambda_n}^{\text{inv},\epsilon}\circ\mathcal{M}_{\mathfrak{I}_{\Lambda_n}}\circ(\mathcal{V}_{\Lambda_n}^{\text{quot},\epsilon})^{-1} \\
            0 & \mathcal{J}_{\Lambda_n}^{\text{quot},\epsilon}
        \end{pmatrix}.
    \end{align*}
    Since $\mathcal{J}_{\Lambda_n}^{\text{inv},\epsilon}$ and $\mathcal{J}_{\Lambda_n}^{\text{quot},\epsilon}$ have entirely distinct spectra,
     the Sylvester equation
    \begin{align*}
        \mathcal{X}_{\Lambda_n}^\epsilon\mathcal{J}_{\Lambda_n}^{\text{quot},\epsilon} -\mathcal{J}_{\Lambda_n}^{\text{inv},\epsilon}\mathcal{X}_{\Lambda_n}^\epsilon
        + \mathcal{V}_{\Lambda_n}^{\text{inv},\epsilon}\circ\mathcal{M}_{\mathfrak{I}_{\Lambda_n}}\circ(\mathcal{V}_{\Lambda_n}^{\text{quot},\epsilon})^{-1}
        = 0
    \end{align*}
    has a unique solution $\mathcal{X}_{\Lambda_n}^\epsilon$ (cf. \cite{Rosenblum1956,Bhatia1997}).
    Therefore, the operator defined by
    \begin{align}\label{eq:V_in_the_model}
        \mathcal{V}_{\Lambda_n}^\epsilon
        = \begin{pmatrix}
            I & \mathcal{X}_{\Lambda_n}^\epsilon \\
            0 & I
        \end{pmatrix}
        \circ (\mathcal{V}_{\Lambda_n}^{\text{inv},\epsilon}\oplus\mathcal{V}_{\Lambda_n}^{\text{quot},\epsilon})\circ\mathcal{U}_{\Lambda_n}
        = \begin{pmatrix}
            \mathcal{V}_{\Lambda_n}^{\text{inv},\epsilon} & \mathcal{X}_{\Lambda_n}^\epsilon\mathcal{V}_{\Lambda_n}^{\text{quot},\epsilon} \\
            0 & \mathcal{V}_{\Lambda_n}^{\text{quot},\epsilon}
        \end{pmatrix}\circ\mathcal{U}_{\Lambda_n}
    \end{align}
    transforms $(1-\epsilon)^{-1}\mathcal{L}_{\Lambda_n}$ into the Jordan canonical form,
    \begin{align*}
        \mathcal{V}_{\Lambda_n}^\epsilon\circ\frac{1}{1-\epsilon}\mathcal{L}_{\Lambda_n}^\epsilon\circ(\mathcal{V}_{\Lambda_n}^\epsilon)^{-1}
        = \begin{pmatrix}
            \mathcal{J}_{\Lambda_n}^{\text{inv},\epsilon} & 0 \\
            0 & \mathcal{J}_{\Lambda_n}^{\text{quot},\epsilon}
        \end{pmatrix},
    \end{align*}
    up to orders of Jordan cells.
    Again, since $\mathcal{J}_{\Lambda_n}^{\text{inv},\epsilon}$ and $\mathcal{J}_{\Lambda_n}^{\text{quot},\epsilon}$ have entirely distinct spectra, the degrees of freedom of the map that transforms $\mathcal{L}_{\Lambda_n}/(1-\epsilon)$ into the Jordan canonical form $\mathcal{J}_{\Lambda_n}^{\text{inv},\epsilon}\oplus\mathcal{J}_{\Lambda_n}^{\text{quot},\epsilon}$ are evaluated as
    \begin{align*}
         &\left\{\tilde{\mathcal{V}}_{\Lambda_n}^\epsilon\in B(\mathfrak{A}_{\Lambda_n},\mathbb{C}^{2^{4n}}) \middle|\ \mathcal{J}_{\Lambda_n}^\epsilon = \tilde{\mathcal{V}}_{\Lambda_n}\circ\frac{1}{1-\epsilon}\mathcal{L}_{\Lambda_n}\circ(\tilde{\mathcal{V}}_{\Lambda_n}^\epsilon)^{-1} \right\} \\
         &\subset \{ (\mathcal{S}^{\text{inv}}\oplus\mathcal{S}^{\text{quot}})\circ\mathcal{V}_{\Lambda_n^\epsilon} \mid \mathcal{S}^{\text{inv}}\in\mathrm{GL}(\mathbb{C}^{d_{\text{inv}}}),\ \mathcal{S}^{\text{quot}}\in\mathrm{GL}(\mathbb{C}^{d_{\text{quot}}}) \}.
    \end{align*}
    
    We intend to evaluate the condition number $\kappa_{\{1\}}(\tilde{\mathcal{V}}_{\Lambda_n}^\epsilon)$ of $\tilde{\mathcal{V}}_{\Lambda_n}^\epsilon=(\mathcal{S}^{\text{inv}}\oplus\mathcal{S}^{\text{quot}})\circ\mathcal{V}_{\Lambda_n^\epsilon}$ for the map $\mathcal{V}_{\Lambda_n}^\epsilon$ given by Eq.~\eqref{eq:V_in_the_model} and some matrices $\mathcal{S}^{\text{inv}}\in\mathrm{GL}(\mathbb{C}^{d_{\text{inv}}})$, $\mathcal{S}^{\text{quot}}\in\mathrm{GL}(\mathbb{C}^{d_{\text{quot}}})$.
    Since
    \begin{align*}
        \|\tilde{\mathcal{V}}_{\Lambda_n}^\epsilon\restriction\mathfrak{A}_{\{1\}}\|_{\mathfrak{A}_{\{1\}}\to\mathbb{C}^{2^{4n}}}
        &= \left\| \begin{pmatrix}
            \mathcal{S}^{\text{inv}}\mathcal{V}_{\Lambda_n}^{\text{inv},\epsilon} & \mathcal{S}^{\text{inv}}\mathcal{X}_{\Lambda_n}^\epsilon\mathcal{V}_{\Lambda_n}^{\text{quot},\epsilon} \\
            0 & \mathcal{S}^{\text{quot}}\mathcal{V}_{\Lambda_n}^{\text{quot},\epsilon}
        \end{pmatrix}\restriction\mathfrak{A}_{\{1\}} \right\|_{\mathfrak{A}_{\{1\}}\to\mathbb{C}^{2^{4n}}} \\
        &\ge \|\mathcal{S}^{\text{inv}}\mathcal{V}_{\Lambda_n}^{\text{inv},\epsilon}\restriction\mathfrak{A}_{\{1\}}\cap\mathfrak{I}_{\Lambda_n}\|_{\mathfrak{A}_{\{1\}}\cap\mathfrak{I}_{\Lambda_n}\to\mathbb{C}^{d_{\text{inv}}}}
    \end{align*}
    and 
    \begin{align*}
        \|(\tilde{\mathcal{V}}_{\Lambda_n}^\epsilon)^{-1}\|_{\mathbb{C}^{2^{4n}}\to\mathfrak{A}_{\Lambda_n}}
        &= \left\| \begin{pmatrix}
            (\mathcal{V}_{\Lambda_n}^{\text{inv},\epsilon})^{-1}(\mathcal{S}^{\text{inv}})^{-1} & -(\mathcal{V}_{\Lambda_n}^{\text{inv},\epsilon})^{-1}\mathcal{X}_{\Lambda_n}^\epsilon(\mathcal{S}^{\text{inv}})^{-1} \\
            0 & (\mathcal{V}_{\Lambda_n}^{\text{quot},\epsilon})^{-1}(\mathcal{S}^{\text{quot}})^{-1}
        \end{pmatrix} \right\|_{\mathbb{C}^{2^{4n}}\to\mathfrak{A}_{\Lambda_n}} \\
        &\ge \|(\mathcal{V}_{\Lambda_n}^{\text{inv},\epsilon})^{-1}(\mathcal{S}^{\text{inv}})^{-1}\|_{\mathbb{C}^{d_{\text{inv}}}\to\mathfrak{I}_{\Lambda_n}},
    \end{align*}
    we have
    \begin{align*}
        \kappa_{\{1\}}(\tilde{\mathcal{V}}_{\Lambda_n}^\epsilon)
        \ge \|\mathcal{S}^{\text{inv}}\mathcal{V}_{\Lambda_n}^{\text{inv},\epsilon}\restriction\mathfrak{A}_{\{1\}}\cap\mathfrak{I}_{\Lambda_n}\|_{\mathfrak{A}_{\{1\}}\cap\mathfrak{I}_{\Lambda_n}\to\mathbb{C}^{d_{\text{inv}}}}\|(\mathcal{V}_{\Lambda_n}^{\text{inv},\epsilon})^{-1}(\mathcal{S}^{\text{inv}})^{-1}\|_{\mathbb{C}^{d_{\text{inv}}}\to\mathfrak{I}_{\Lambda_n}}.
    \end{align*}
    To obtain the explicit form of $\mathcal{J}_{\Lambda_n}^{\text{inv},\epsilon}$, $\mathcal{V}_{\Lambda_n}^{\text{inv},\epsilon}$, and $(\mathcal{V}_{\Lambda_n}^{\text{inv},\epsilon})^{-1}$, 
    we write down the action of $\mathcal{L}_{\Lambda_n}$ on $\mathfrak{I}_{\Lambda_n}$ as follows:
    \begin{align*}
        &\mathcal{L}_{\Lambda_n}\left( \begin{pmatrix}
            \hat{F}_1 & \cdots & \hat{F}_{K_n}
        \end{pmatrix}\oplus\bigoplus_{\alpha,\beta\{0,1\}}\begin{pmatrix}
            \hat{D}_{2n}\hat{E}^\alpha\hat{X}^\beta & \cdots & \hat{D}_1\hat{E}^\alpha\hat{X}^\beta
        \end{pmatrix} \right) \notag \\
        &= \left[ \begin{pmatrix}
            \hat{F}_1 & \cdots & \hat{F}_{K_n}
        \end{pmatrix}\oplus\bigoplus_{\alpha,\beta\{0,1\}}\begin{pmatrix}
            \hat{D}_{2n}\hat{E}^\alpha\hat{X}^\beta & \cdots & \hat{D}_1\hat{E}^\alpha\hat{X}^\beta
        \end{pmatrix} \right]\left( I_{K_n}\oplus\mathsf{M}_n^{\oplus 4}\right),
    \end{align*}
    where 
    \begin{align*}
        \mathsf{M}_n
        = \begin{pmatrix}
            0 & 1 & 0 & \cdots & 0 & 0 \\
            0 & -1 & 1 & \cdots & 0 & 0 \\
            0 & 0 & -1 & \ddots & 0 & 0 \\
            \vdots & \vdots & \vdots & \ddots & \vdots & \vdots \\
            0 & 0 & 0 & \cdots & -1 & 1 \\
            0 & 0 & 0 & \cdots & 0 & -1
        \end{pmatrix}.
    \end{align*}
    The Jordan canonical form of $(1-\epsilon)^{-1}\mathsf{M}_n$ is 
    \begin{align*}
        \mathsf{J}_{n}^\epsilon
        = \begin{pmatrix}
            0 & 0 & 0 & \cdots & 0 & 0 \\
            0 & -\frac{1}{1-\epsilon} & 1 & \cdots & 0 & 0 \\
            0 & 0 & -\frac{1}{1-\epsilon} & \cdots & 0 & 0 \\
            \vdots & \vdots & \vdots & \ddots & \vdots & \vdots \\
            0 & 0 & 0 & \cdots & -\frac{1}{1-\epsilon} & 1 \\
            0 & 0 & 0 & \cdots & 0 & -\frac{1}{1-\epsilon}
        \end{pmatrix}
        = \mathsf{V}_n^\epsilon\frac{\mathsf{M}_n}{1-\epsilon}(\mathsf{V}_n^\epsilon)^{-1}
    \end{align*}
    with
    \begin{align*}
        \mathsf{V}_n^\epsilon
        = \left(\begin{array}{c|ccccc}
            1 & 1 & 1 & 1 & \cdots & 1 \\
            \hline
             0 & 1 & & & & \\
             0 & & \frac{1}{1-\epsilon} & & & \\
             0 & & & \frac{1}{(1-\epsilon)^2} & & \\
             \vdots & & & & \ddots & \\
             0 & & & & & \frac{1}{(1-\epsilon)^{2n-2}}
        \end{array}\right)
    \end{align*}
    and 
    \begin{align*}
        (\mathsf{V}_n^\epsilon)^{-1}
        = \left(\begin{array}{c|ccccc}
            1 & -1 & -(1-\epsilon) & -(1-\epsilon)^2 & \cdots & -(1-\epsilon)^{2n-2} \\
            \hline
             0 & 1 & & & & \\
             0 & & 1-\epsilon & & & \\
             0 & & & (1-\epsilon)^2 & & \\
             \vdots & & & & \ddots & \\
             0 & & & & & (1-\epsilon)^{2n-2}
        \end{array}\right).
    \end{align*}
    Therefore, we get the following expressions
    \begin{gather*}
        \mathcal{J}_{\Lambda_n}^{\text{inv},\epsilon}
        = \frac{I_{K_n}}{1-\epsilon}\oplus(\mathsf{J}_n^\epsilon)^{\oplus 4}, \\
        \mathcal{V}_{\Lambda_n}^{\text{inv},\epsilon}(\hat{A})
        = \begin{pmatrix}
            \braket{\hat{F}_1,\hat{A}}_{\text{HS}} \\
            \vdots \\
            \braket{\hat{F}_{K_n},\hat{A}}_{\text{HS}}
        \end{pmatrix}\oplus\bigoplus_{\alpha,\beta\in\{0,1\}}\mathsf{V}_n^\epsilon\begin{pmatrix}
            \braket{\hat{D}_{2n}\hat{E}^\alpha\hat{X}^\beta,\hat{A}}_{\text{HS}} \\
            \vdots \\
            \braket{\hat{D}_1\hat{E}^\alpha\hat{X}^\beta,\hat{A}}_{\text{HS}}
        \end{pmatrix},
        \quad
        \hat{A}\in\mathfrak{I}_{\Lambda_n}, \\
        (\mathcal{V}_{\Lambda_n}^{\text{inv},\epsilon})^{-1}(\boldsymbol{u})
        = \left[ \begin{pmatrix}
            \hat{F}_1 & \cdots & \hat{F}_{K_n}
        \end{pmatrix}\oplus\bigoplus_{\alpha,\beta\in}\{0,1\}\begin{pmatrix}
            \hat{D}_{2n}\hat{E}^\alpha\hat{X}^\beta & \cdots & \hat{D}_1\hat{E}^\alpha\hat{X}^\beta
        \end{pmatrix}(\mathsf{V}_n^\epsilon)^{-1} \right]\boldsymbol{u},
        \quad
        \boldsymbol{u}\in\mathbb{C}^{d_{\text{inv}}}.
    \end{gather*}
    
    The concrete form of $\mathcal{J}_{\Lambda_n}^{\text{inv},\epsilon}$ makes it possible to determine the degrees of freedom of similarity transformation $\mathcal{S}^{\text{inv}}\circ\mathcal{J}_{\Lambda_n}^{\text{inv},\epsilon}\circ(\mathcal{S}^{\text{inv}})^{-1}=\mathcal{J}_{\Lambda_n}^{\text{inv},\epsilon}$ as 
    \begin{align*}
        &\left\{ \mathcal{S}^{\text{inv}}\in\mathrm{GL}(\mathbb{C}^{d_{\text{inv}}}) \middle|\ \mathcal{S}^{\text{inv}}\circ\mathcal{J}_{\Lambda_n}^{\text{inv},\epsilon}\circ(\mathcal{S}^{\text{inv}})^{-1}=\mathcal{J}_{\Lambda_n}^{\text{inv},\epsilon} \right\} \\
        &= \{ \mathsf{R}_{K_n}\oplus (S_{a,b}\oplus T_{a,b}I_{2n-1})_{a,b=1}^4 \mid \\
        &\qquad\mathsf{R}_{K_n}\in\mathrm{GL}(\mathbb{C}^{K_n}),\ \mathsf{S}=(S_{a,b})_{a,b=1}^4\in\mathrm{GL}(\mathbb{C}^4),\ \mathsf{T}=(T_{a,b})_{a,b=1}^4\in\mathrm{GL}(\mathbb{C}^4) \}
    \end{align*}
    for $n\ge 2$. 
    Note that for $n=1$ the set of $\mathcal{S}^{\text{inv}}$ becomes homomorphic to $\mathrm{GL}(\mathbb{C}^{K_n+4})\times\mathrm{GL}(\mathbb{C}^4)$ because the eigenvalue $-1$ of $\mathcal{L}_{\{1\}}$ is semisimple.
    Therefore, for $n\ge 2$, we have
    \begin{align*}
        &\inf\left\{ \kappa_{\{1\}}(\tilde{\mathcal{V}}_{\Lambda_n}^\epsilon) \middle|\ \tilde{\mathcal{V}}_{\Lambda_n}^\epsilon\in B(\mathfrak{A}_{\Lambda_n},\mathbb{C}^{2^{4n}}),\ \mathcal{J}_{\Lambda_n}^\epsilon = \tilde{\mathcal{V}}_{\Lambda_n}\circ\frac{\mathcal{L}_{\Lambda_n}}{1-\epsilon}\circ(\tilde{\mathcal{V}}_{\Lambda_n}^\epsilon)^{-1} \right\} \\
        &\ge \inf\{ \|\mathcal{S}^{\text{inv}}\mathcal{V}_{\Lambda_n}^{\text{inv},\epsilon}\restriction\mathfrak{A}_{\{1\}}\cap\mathfrak{I}_{\Lambda_n}\|_{\mathfrak{A}_{\{1\}}\cap\mathfrak{I}_{\Lambda_n}\to\mathbb{C}^{d_{\text{inv}}}}\|(\mathcal{V}_{\Lambda_n}^{\text{inv},\epsilon})^{-1}(\mathcal{S}^{\text{inv}})^{-1}\|_{\mathbb{C}^{d_{\text{inv}}}\to\mathfrak{I}_{\Lambda_n}} \mid \\
        &\qquad \mathcal{S}^{\text{inv}} = \mathsf{R}_{K_n}\oplus (S_{a,b}\oplus T_{a,b}I_{2n-1})_{a,b=1}^4,\ \mathsf{R}_{K_n}\in\mathrm{GL}(\mathbb{C}^{K_n}),\ \mathsf{S},\mathsf{T}\in\mathrm{GL}(\mathbb{C}^4) \} \\
        &= \inf\Bigg\{ \frac{\|\mathcal{S}^{\text{inv}}\mathcal{V}_{\Lambda_n}^{\text{inv},\epsilon}(\hat{A})\|_2}{\|\hat{A}\|} \frac{\|(\mathcal{V}_{\Lambda_n}^{\text{inv},\epsilon})^{-1}(\boldsymbol{u})\|}{\|\mathcal{S}^{\text{inv}}\boldsymbol{u}\|_2} \Bigg|  \hat{A}\in\mathfrak{A}_{\{1\}}\cap\mathfrak{I}_{\Lambda_n},\ \boldsymbol{u}\in\mathbb{C}^{d_{\text{inv}}}, \\
        &\qquad\mathcal{S}^{\text{inv}} = \mathsf{R}_{K_n}\oplus (S_{a,b}\oplus T_{a,b}I_{2n-1})_{a,b=1}^4,\ \mathsf{R}_{K_n}\in\mathrm{GL}(\mathbb{C}^{K_n}),\ \mathsf{S},\mathsf{T}\in\mathrm{GL}(\mathbb{C}^4) \Bigg\}.
    \end{align*}
    
    To evaluate the right hand side from below, we substitute a concrete operator $\hat{A}\in\mathfrak{A}_{\{1\}}\cap\mathfrak{I}_{\Lambda_n}$ and a vector $\boldsymbol{u}\in\mathbb{C}^{d_{\text{inv}}}$ for $\|\mathcal{S}^{\text{inv}}\mathcal{V}_{\Lambda_n}^{\text{inv},\epsilon}(\hat{A})\|_2/\|\hat{A}\|$ and $\|(\mathcal{V}_{\Lambda_n}^{\text{inv},\epsilon})^{-1}(\boldsymbol{u})\|/\|\mathcal{S}^{\text{inv}}\boldsymbol{u}\|_2$ with a matrix $\mathcal{S}^{\text{inv}} = \mathsf{R}_{K_n}\oplus (S_{a,b}\oplus T_{a,b}I_{2n-1})_{a,b=1}^4,\ \mathsf{R}_{K_n}\in\mathrm{GL}(\mathbb{C}^{K_n}),\ \mathsf{S},\mathsf{T}\in\mathrm{GL}(\mathbb{C}^4)$.
    Let $\boldsymbol{0}_m$ denote the zero vector of length $m$.
    For $\hat{A}=\hat{D}_1-\hat{D}_2\in\mathfrak{A}_{\{1\}}\cap\mathfrak{I}_{\Lambda_n}$,
    \begin{align*}
        \|\hat{A}\|
        = \|\hat{\sigma}^3\otimes\hat{\sigma}^0-\hat{\sigma}^0\otimes\hat{\sigma}^3\|
        = 2
    \end{align*}
    and 
    \begin{align*}
        \|\mathcal{S}^{\text{inv}}\mathcal{V}_{\Lambda_n}^{\text{inv},\epsilon}(\hat{A})\|_2
        &= \left\|(S_{a,b}\oplus T_{a,b}I_{2n-1})_{a,b=1}^4 \left(\mathsf{V}_n^\epsilon\begin{pmatrix}
            \boldsymbol{0}_{2n-2} \\
            -1 \\
            1
        \end{pmatrix}\oplus\boldsymbol{0}_{2n}\oplus\boldsymbol{0}_{2n}\oplus\boldsymbol{0}_{2n}\right)\right\|_2 \\
        &= \left\|\begin{pmatrix}
            T_{1,1} \\
            T_{2,1} \\
            T_{3,1} \\
            T_{4,1}
        \end{pmatrix}\otimes\begin{pmatrix}
            \frac{-1}{(1-\epsilon)^{2n-3}} \\
            \frac{1}{(1-\epsilon)^{2n-2}}
        \end{pmatrix}\right\|_2 \\
        &= \sqrt{\sum_{a=1}^4|T_{a,1}|^2}\frac{\sqrt{1+(1-\epsilon)^2}}{(1-\epsilon)^{2n-2}}
    \end{align*}
    hold.
    For $\boldsymbol{u}=\boldsymbol{0}_{K_n}\oplus\begin{pmatrix} 0 \\ 1 \\ \boldsymbol{0}_{2n-2} \end{pmatrix}\oplus\boldsymbol{0}_{2n}\oplus\boldsymbol{0}_{2n}\oplus\boldsymbol{0}_{2n}\in\mathbb{C}^{d_{\text{inv}}}$, 
    \begin{align*}
        \|(\mathcal{V}_{\Lambda_n}^{\text{inv},\epsilon})^{-1}(\boldsymbol{u})\|
        &= \left\|\begin{pmatrix}
            \hat{D}_{2m} & \cdots & \hat{D}_1
        \end{pmatrix}(\mathsf{V}_n^\epsilon)^{-1}\begin{pmatrix}
            0 \\
            1 \\
            \boldsymbol{0}_{2n-2}
        \end{pmatrix}\right\| \\
        &= \|\hat{D}_{2m}-\hat{D}_{2m-1}\|
        = \|\hat{\sigma}^0\otimes\hat{\sigma}^3-\hat{\sigma}^3\otimes\hat{\sigma}^0\|
        = 2
    \end{align*}
    and 
    \begin{align*}
        \|\mathcal{S}^{\text{inv}}\boldsymbol{u}\|_2
        = \left\|\begin{pmatrix}
            T_{1,1} \\
            T_{2,1} \\
            T_{3,1} \\
            T_{4,1}
        \end{pmatrix}\right\|
        = \sqrt{\sum_{a=1}^4|T_{a,1}|^2}
    \end{align*}
    hold.
    
    Consequently, we conclude that the condition number diverges to infinity as $n\to\infty$ from the inequality
    \begin{align*}
        &\inf\left\{ \kappa_{\{1\}}(\tilde{\mathcal{V}}_{\Lambda_n}^\epsilon) \middle|\ \tilde{\mathcal{V}}_{\Lambda_n}^\epsilon\in B(\mathfrak{A}_{\Lambda_n},\mathbb{C}^{2^{4n}}),\ \mathcal{J}_{\Lambda_n}^\epsilon = \tilde{\mathcal{V}}_{\Lambda_n}\circ\frac{\mathcal{L}_{\Lambda_n}}{1-\epsilon}\circ(\tilde{\mathcal{V}}_{\Lambda_n}^\epsilon)^{-1} \right\} \\
        &\ge \inf\left\{\frac{\sqrt{\sum_{a=1}^4|T_{a,1}|^2}\frac{\sqrt{1+(1-\epsilon)^2}}{(1-\epsilon)^{2n-2}}}{2}\frac{2}{\sqrt{\sum_{a=1}^4|T_{a,1}|^2}} \middle|\ \mathsf{T}\in\mathrm{GL}(\mathbb{C}^4) \right\} \\
        &= \frac{\sqrt{1+(1-\epsilon)^2}}{(1-\epsilon)^{2n-2}}.
    \end{align*}
    This result is consistent with the fact that our model should violate the conditions (A2) in Theorem~\ref{thm:main_theorem} and (B2) in Theorem~\ref{thm:main_theorem2}.
\end{proof}


\section*{Acknowledgement}
We thank Tomohiro Sasamoto for valuable comments.
We would like to thank Mariam Ughrelidze and Lorenza Viola for informing us of the related work \cite{Ughrelidze2024}, which uses the condition number to analyze stability and spectral features in quadratic bosonic Liouvillians.
This work was partially done while N. Hara was at Kyoto University.
We appreciate the Mathematics-Based Creation of Science (MACS) Program at Graduate School of Science, Kyoto University for providing us with an opportunity to start this work.

\section*{Data Availability}
All data generated or analyzed during this study are available within this article.

\section*{Statements and Declarations}
\subsection*{Funding}
K. Shimomura is supported by JST CREST Grant No.~JPMJCR19T2 and JSPS KAKENHI Grant No.~JP25KJ1632.
S. Kusuoka is supported by JSPS KAKENHI Grant No.~JP19H00643, No.~JP21H00988, No.~JP22H00099, and No.~JP23K20801.

\subsection*{Competing Interests}
The authors have no conflict of interest to declare that are relevant to the content of this article.

\appendix

\section{Net in a nutshell}\label{sec:net}
In this Appendix, we summarize basic properties of the net as required in the main text.
The description in this appendix is based on \cite{Kelley1955}.

\begin{dfn}
    A non-empty set $\mathcal{M}$ is a \emph{directed set} if there is a relation $\le$ on $\mathcal{M}$ satisfying three conditions as follows.
    \begin{itemize}
        \item $\mu\le\mu$ for each $\mu\in\mathcal{M}$,
        \item if $\mu_1\le\mu_2$ and $\mu_2\le\mu_3$ then $\mu_1\le\mu_3$,
        \item for $\mu_1,\mu_2\in\mathcal{M}$, there is some $\mu_3\in\mathcal{M}$ satisfying $\mu_1\le\mu_3$ and $\mu_2\le\mu_3$.
    \end{itemize}
    For example, a set of all non-empty finite subsets of a non-empty set can be directed by the inclusion.
    A \emph{net} in a set $X$ is a function from some directed set $\mathcal{M}$ to $X$.
    We often write a net $\mathcal{M}\to X$, $\mu\mapsto x_\mu$ by $(x_\mu)_{\mu\in\mathcal{M}}$ or simply $(x_\mu)$.
    A \emph{subnet} of a net $P\colon\mathcal{M}\to X$ is the composition $P\circ\varphi$, where $\varphi\colon\mathcal{N}\to\mathcal{M}$ is a \emph{cofinal} function from some directed set $\mathcal{N}$ to $\mathcal{M}$, i.e., for each $\mu\in\mathcal{M}$ there is some $\nu\in\mathcal{N}$ such that $\mu\le\varphi(\nu)$.
    We often write a subnet of a net $(x_\mu)_{\mu\in\mathcal{M}}$ by $(x_{\mu_\nu})_{\nu\in\mathcal{N}}$ or simply $(x_{\mu_\nu})$, where the function $\nu\in\mathcal{N}\mapsto\mu_\nu\in\mathcal{M}$ is increasing and cofinal in $\mathcal{M}$.
\end{dfn}
Note that some textbooks (such as \cite{willard1970}) define the subnet in a different way; they require that the function $\varphi\colon\mathcal{N}\to\mathcal{M}$ is \emph{increasing}, i.e., $\varphi(\nu_1)\le\varphi(\nu_2)$ whenever $\nu_1\le\nu_2$ for $\nu_1,\nu_2\in\mathcal{N}$, as well as cofinal.

The convergence of a net is defined in the following way.

\begin{dfn}
    A net $(x_\mu)_{\mu\in\mathcal{M}}$ in a topological space $X$ \emph{converges} to $x\in X$ (witten $x_\mu\xrightarrow{\mu}x$ or $\lim_\mu x_\mu = x$), if for each neighborhood $U$ of $x$, there is some $\mu_0\in\mathcal{M}$ such that $x_\mu\in U$ whenever $\mu_0\le\mu$.
    A net has $y\in X$ as a \emph{cluster point} if it has a subnet converging to $y$.
\end{dfn}

In particular, the convergence of a net $(x_\mu)$ in a metric spece $(X,\operatorname{dist})$ agrees with the definition that $x_\mu\xrightarrow{\mu}x$ if for any positive number $\epsilon$ there is some $\mu_0\in\mathcal{M}$ such that $\mu_0\le\mu$ implies $\operatorname{dist}(x_\mu,x)<\epsilon$.
The assertion below holds for general topological spaces.
\begin{thm}[\cite{Kelley1955} p.74 (c)]\label{thm:subnet_of_subnet}
    If every subnet of a net $(x_\mu)$ in a topological space $X$ has a subnet converging to $x\in X$, then $(x_\mu)$ converges to $x$.
\end{thm}

Furthermore, the compactness of a topological space is characterized by using the concept of nets.
\begin{thm}
    A topological space $X$ is compact, if and only if each net in $X$ has a cluster point.
\end{thm}

\section{\texorpdfstring{Some properties of $C^*$-algebra}{TEXT}}\label{sec:C*-algebra}
In this appendix, we summarize some properties of $C^*$-algebra used in the main text.

A Banach space $\mathfrak{A}$ is said to be a \emph{Banach $*$-algebra} if $\mathfrak{A}$ is a $\mathbb{C}$-algebra with the involution $*$.
A unital $C^*$-algebra is a unital Banach $*$-algebra with the unit $\hat{I}$ satisfying the \emph{$C^*$-condition}:
\begin{align*}
    \|\hat{A}^*\hat{A}\|
    = \|\hat{A}\|^2
\end{align*}
for any $\hat{A}\in\mathfrak{A}$. 
An element $\hat{A}$ of a $C^*$-algebra $\mathfrak{A}$ is defined to be positive if there exists an element $\hat{B}\in\mathfrak{A}$ such that $\hat{A}=\hat{B}^*\hat{B}$.

A linear functional $\psi$ over a $C^*$-algebra $\mathfrak{A}$ is defined to be positive if 
\begin{align*}
    \psi(\hat{A}^*\hat{A})
    \ge 0
\end{align*}
for any $\hat{A}\in\mathfrak{A}$. 
A state $\omega$ over $\mathfrak{A}$ is a positive linear functional over a $C^*$-algebra $\mathfrak{A}$ with
\begin{align*}
    \|\omega\|
    \coloneqq \sup\left\{|\omega(\hat{A})|\mid\hat{A}\in\mathfrak{A},|\hat{A}\|=1\right\}
    = 1.
\end{align*}

\begin{thm}[\cite{Bratteli1987} Proposition 2.3.11]
    Let $\omega$ be a linear functional over a unital $C^*$-algebra $\mathfrak{A}$.
    The following conditions are equivalent:
    \begin{enumerate}[(1)]
        \item $\omega$ is positive;
        \item $\omega$ is continuous and satisfies
        \begin{align*}
            \|\omega\|
            = \omega(\hat{I}).
        \end{align*}
    \end{enumerate}
\end{thm}

For the dual space $\mathfrak{A}^*$, the weak-$*$ topology on $\mathfrak{A}^*$ is defined by the family of seminorms $p_{\hat{a}}:\omega \mapsto |\omega(\hat{A})|$ for $\hat{A} \in \mathfrak{A}$. Let $E_\mathfrak{A} \subseteq \mathfrak{A}^*$ be a set of states over $\mathfrak{A}$. The following Banach-Alaoglu theorem shows that $E_\mathfrak{A}$ is weakly-$*$ compact.
\begin{thm}[{\cite[Theorem 2.3.15]{Bratteli1987}}]\label{thm:compactness_of_states_set}
    The set $E_\mathfrak{A}$ of states over a unital $C^*$-algebra $\mathfrak{A}$ is weakly-$*$ compact.
\end{thm}

\printbibliography

\end{document}